%% file: main_v2.tex
\pdfoutput=1

\documentclass{siamonline1116}

\usepackage{subfigure}
\usepackage{afterpage}
\usepackage{amsmath}

\input{ex_shared}

\graphicspath{
  {images/},{prebuiltimages/}
}

\ifpdf
\hypersetup{
  pdftitle={\TheTitle},
  pdfauthor={\TheAuthors}
}
\fi

\usepackage{color, colortbl}
\definecolor{Gray}{gray}{0.9}

\newcommand{\ourcodeurl}{https://bitbucket.org/cdeledalle/ggmm-epll}

\allowdisplaybreaks

\begin{document}
\sloppy

\maketitle

\begin{abstract}
  Patch priors have become an important component of 
  image restoration.
  A powerful approach in this category of restoration algorithms is the popular Expected Patch Log-Likelihood (EPLL)
  algorithm. EPLL uses a Gaussian mixture model (GMM) prior learned on clean image patches 
  as a way to regularize degraded patches.
  In this paper, we show that
  a generalized Gaussian mixture model (GGMM) captures the underlying distribution of patches better
  than a GMM.
  Even though GGMM is a powerful prior to combine with EPLL, the non-Gaussianity of its components
  presents major challenges to be applied to a computationally intensive process of
  image restoration. Specifically,
  each patch has to undergo a patch classification
  step and a shrinkage step. These two steps can be efficiently solved with a GMM prior but are computationally
  impractical when using a GGMM prior. In this paper, we provide approximations and computational recipes for fast evaluation of these two steps, so that
  EPLL can embed a GGMM prior on an image with more than tens of thousands of patches.
  Our main contribution is to analyze the accuracy of our approximations
  based on thorough theoretical analysis.
  Our evaluations indicate that
  the GGMM prior is consistently a better fit for modeling image patch distribution and
  performs better on average
  in image denoising task.
\end{abstract}

\begin{keywords}
  Generalized Gaussian distribution, Mixture models, Image denoising, Patch priors.
\end{keywords}

\begin{AMS}
  68U10, 62H35, 94A08 
\end{AMS}

\section{Introduction}
Image restoration is the process of recovering the underlying clean image from its degraded or corrupted observation(s). The images captured by common imaging systems often contain corruptions such as noise, optical or motion blur due to sensor limitations and/or environmental conditions. For this reason, image restoration algorithms have widespread applications in medical imaging, satellite imaging, surveillance, and general consumer imaging applications.
Priors on natural images play an important role in image restoration algorithms. Image priors are used to denoise or regularize ill-posed restoration problems such as deblurring and super-resolution, to name just a few. Early attempts in designing image priors relied on
modeling local pixel gradients with
Gibbs distributions \cite{geman1984stochastic},
Laplacian distributions (total variation) \cite{rudin1992nonlinear,yang2009efficient},
hyper-Laplacian distribution \cite{krishnan2009fast},
generalized Gaussian distribution \cite{bouman1993generalized},
or Gaussian mixture models \cite{fergus2006removing}.
Concurrently, priors have also been designed by modeling coefficients
of an image in a transformed domain using
generalized Gaussian \cite{mallat1989theory,moulin1999analysis,chang2000adaptive,do2002wavelet} or scaled mixture of Gaussian \cite{portilla2003image}
priors for wavelet or curvelet coefficients \cite{boubchir2005multivariate}.
Alternatively, modeling the distribution of patches of an image
(\ie small windows usually of size $8 \times 8$)
has proven to be a powerful solution. In particular,
popular patch techniques rely on non-local self-similarity \cite{buades2005non},
fields of experts \cite{roth2005fields},
learned patch dictionaries \cite{aharon2006rm,elad2006image,rubinstein2013analysis},
sparse or low-rank properties of stacks of similar patches
\cite{dabov2007image,deledalle2011image,lebrun2013nonlocal,ji2011robust},
patch re-occurrence priors \cite{michaeli2014blind},
or more recently mixture models patch priors
\cite{Zoran11,yu2012solving,van2014student,teodoro2015single,houdard2017,singh2017,niknejad2017class}.


Of these approaches, a successful approach introduced by Zoran and Weiss \cite{Zoran11} is to model patches 
of clean natural images using Gaussian Mixture Model (GMM) priors. The agility of this model lies in the fact that a prior learned on clean image patches can be effectively employed to restore a wide range of inverse problems. It is also easily extendable to include other constraints such as sparsity or multi-resolution patches \cite{sulam2015expected, papyan2016multi}. The use of GMMs for patch priors make these methods computationally tractable and flexible. Although GMM patch prior is effective and popular, in this article, we argue that a generalized Gaussian mixture model (GGMM) is a better fit for image patch prior modeling. Compared to a Gaussian model, a generalized Gaussian distribution (GGD) has an extra degree of freedom controlling the shape of the distribution and it encompasses Gaussian and Laplacian models.

Beyond image restoration tasks, GGDs have been used in several different fields
of image and signal processing, including
watermark detection \cite{cheng2003robust}, texture retrieval \cite{do2002wavelet},
voice activity detection \cite{gazor2003speech} and
MP3 audio encoding \cite{dominguez2003practical}, to cite just a few.
In these tasks, GGDs are used to characterize or model the prior distribution of
clean signals, for instance, from their DCT coefficients or frequency subbands
for natural images \cite{westerink1991subband,tanabe1992subband,muller1993distribution,aiazzi1999estimation} or videos \cite{sharifi1995estimation},
gradients for X-ray images \cite{bouman1993generalized},
wavelet coefficients for natural
\cite{mallat1989theory,moulin1999analysis,cheng2003robust},
textured \cite{do2002wavelet},
or ultrasound images \cite{achim2001novel},
tangential wavelet coefficients for three-dimensional mesh data \cite{lavu2003estimation},
short time windows for speech signals \cite{kokkinakis2005exponent}
or frequency subbands for speech \cite{gazor2003speech} or audio signals \cite{dominguez2003practical}.

In this paper, we go one step further and use multi-variate GGD with one scale and one shape parameter per dimension.
The superior patch prior modeling capability of such a GGMM over a GMM is illustrated in \Cref{fig:ggmm_vs_gmm}.
The figure shows histograms of six orthogonal 1-D projections of subset of clean patches onto the eigenvectors of the covariance matrix of a single component of the GMM.
To illustrate the difference in the shapes ($\nu$) and scales ($\lambda$) of the distributions of each dimension,
we have chosen a few projections corresponding to both the most and the least significant eigenvalues.
It can be seen that GGD is a better fit on the obtained histograms than a Gaussian model.
Additionally, different dimensions of the patch follow a different GGD (\ie has a different shape and scale parameter).
Hence, it does not suffice to model all the feature dimensions of 
a given cluster of patches as Laplacian or Gaussian. Therefore, we propose to model patch priors as GGMM distributed with a separate shape and scale parameters for each feature dimension of a GGD component. This differs from the recent related approach in \cite{niknejad2017class} that considered GGMM where each component has a fixed shape parameter for all directions.

\begin{figure}
  \includegraphics[width=.49\linewidth]{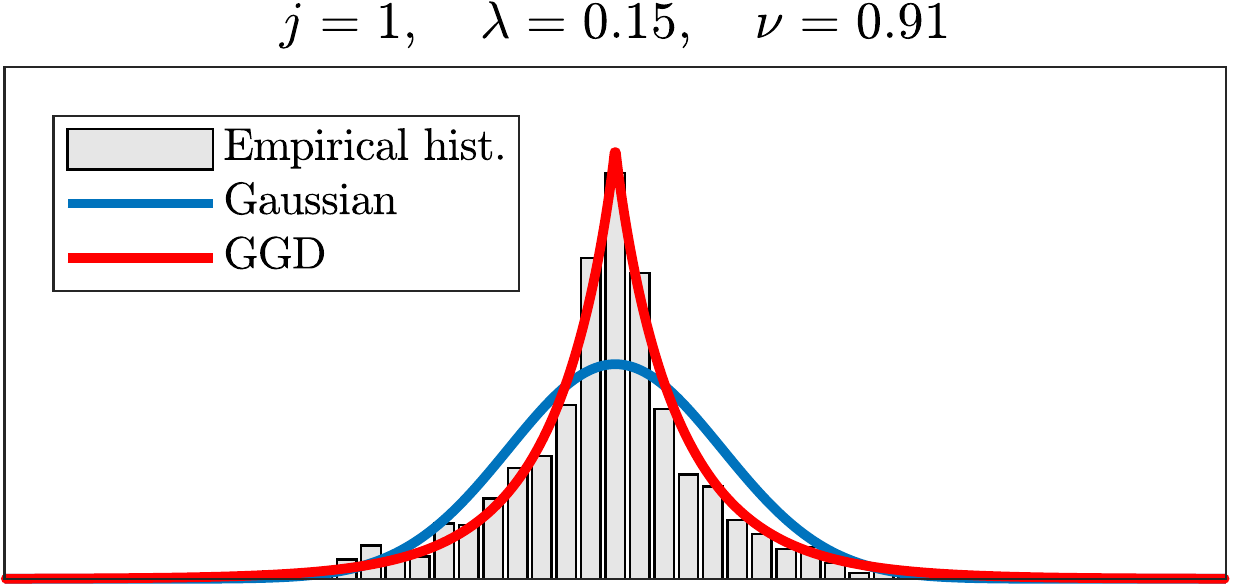}\hfill%
  \includegraphics[width=.49\linewidth]{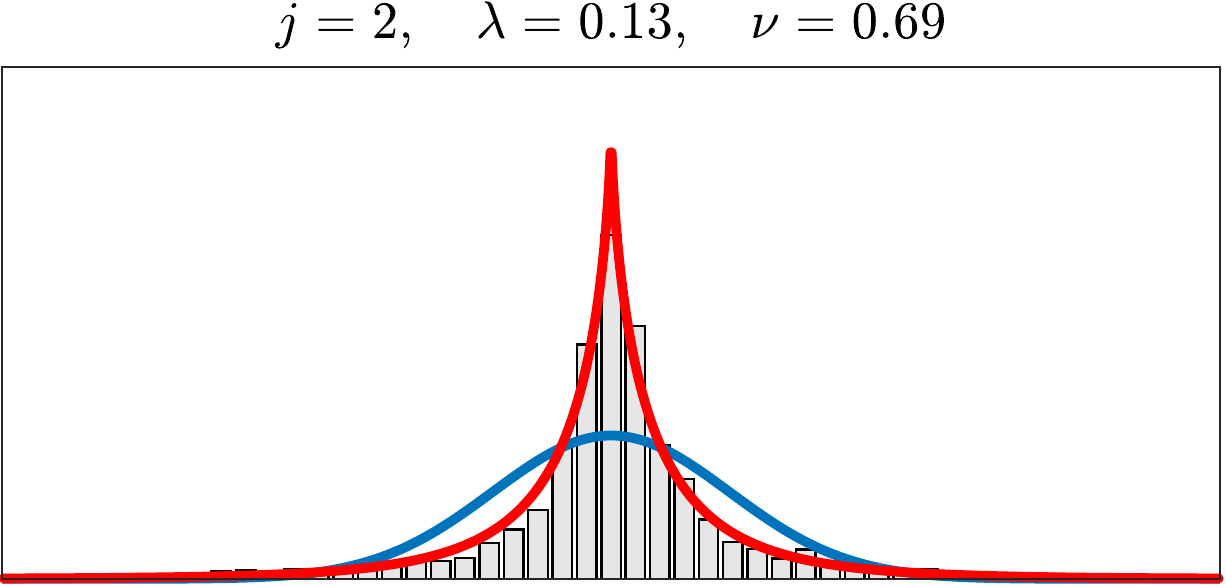}\\[1em]
  \includegraphics[width=.49\linewidth]{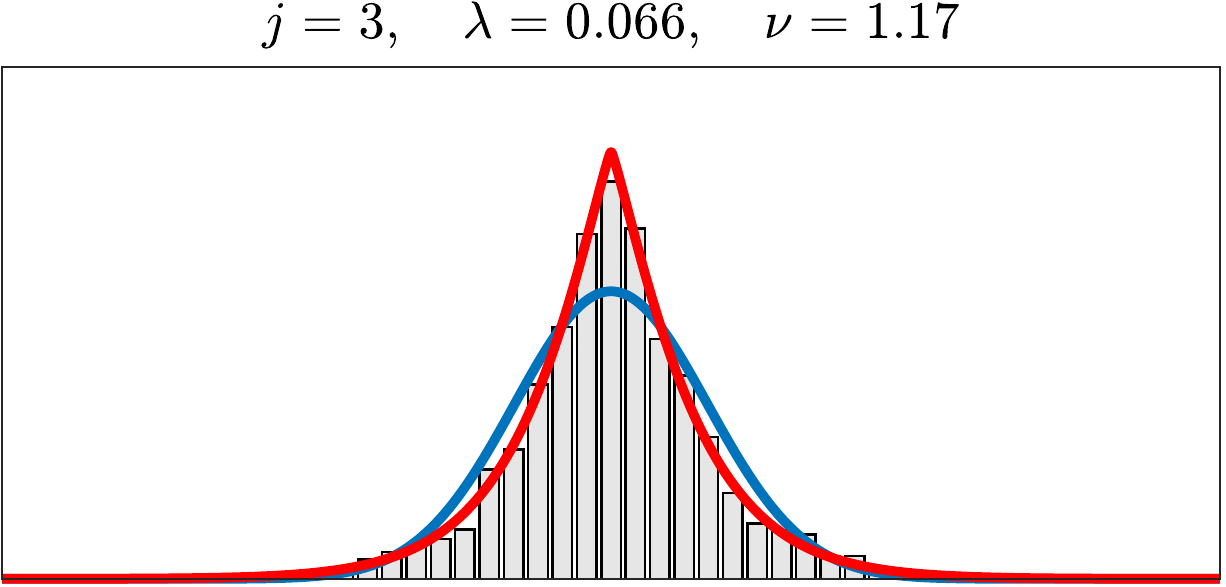}\hfill%
  \includegraphics[width=.49\linewidth]{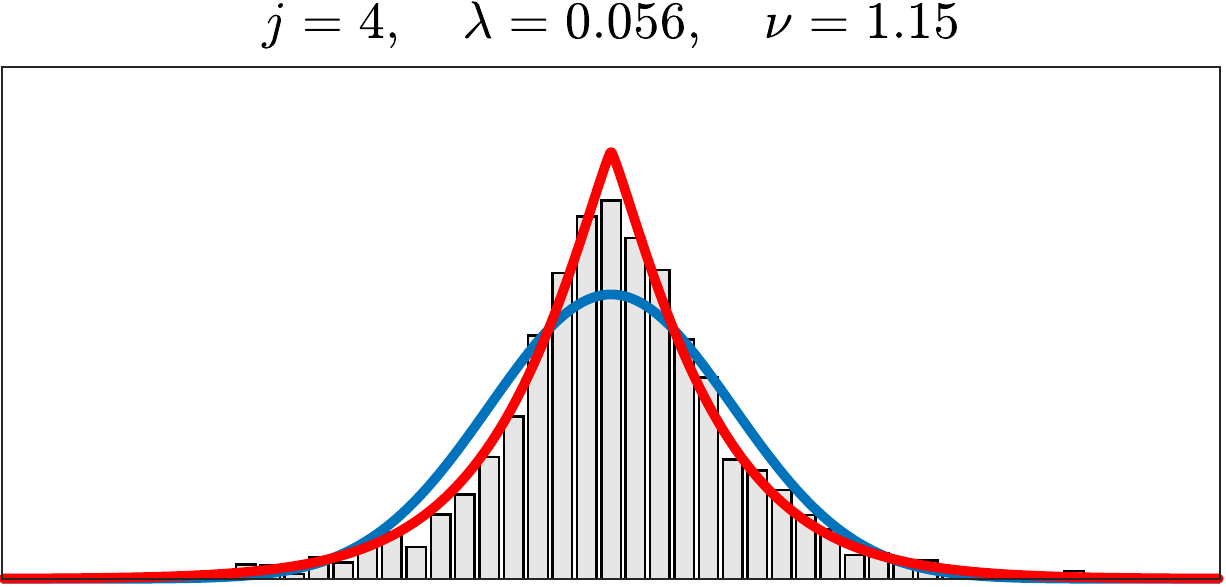}\\[1em]
  \includegraphics[width=.49\linewidth]{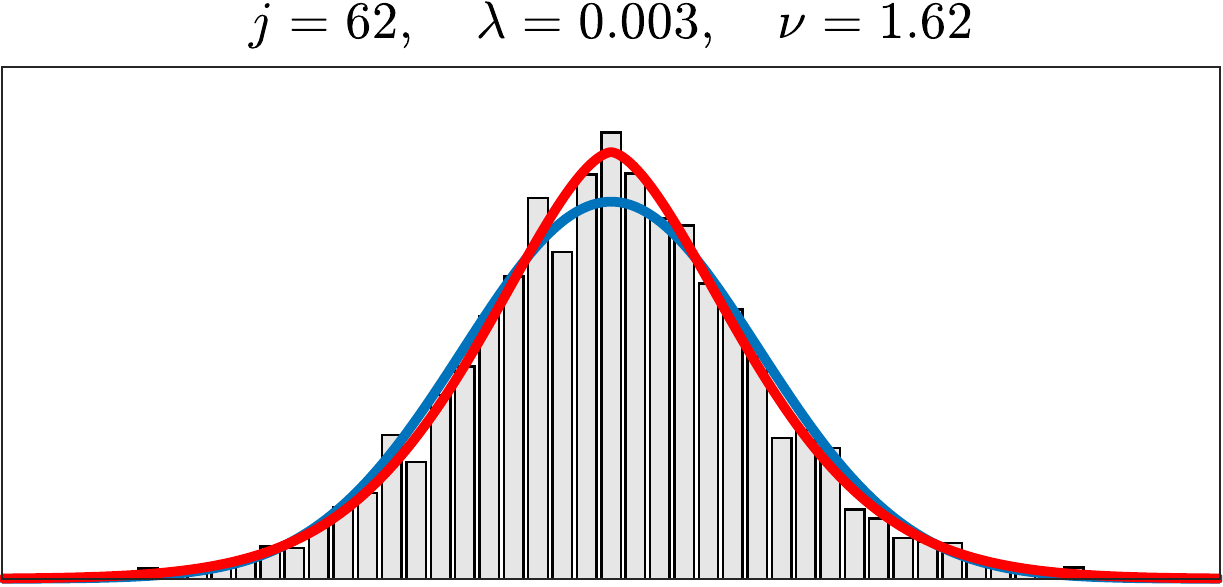}\hfill%
  \includegraphics[width=.49\linewidth]{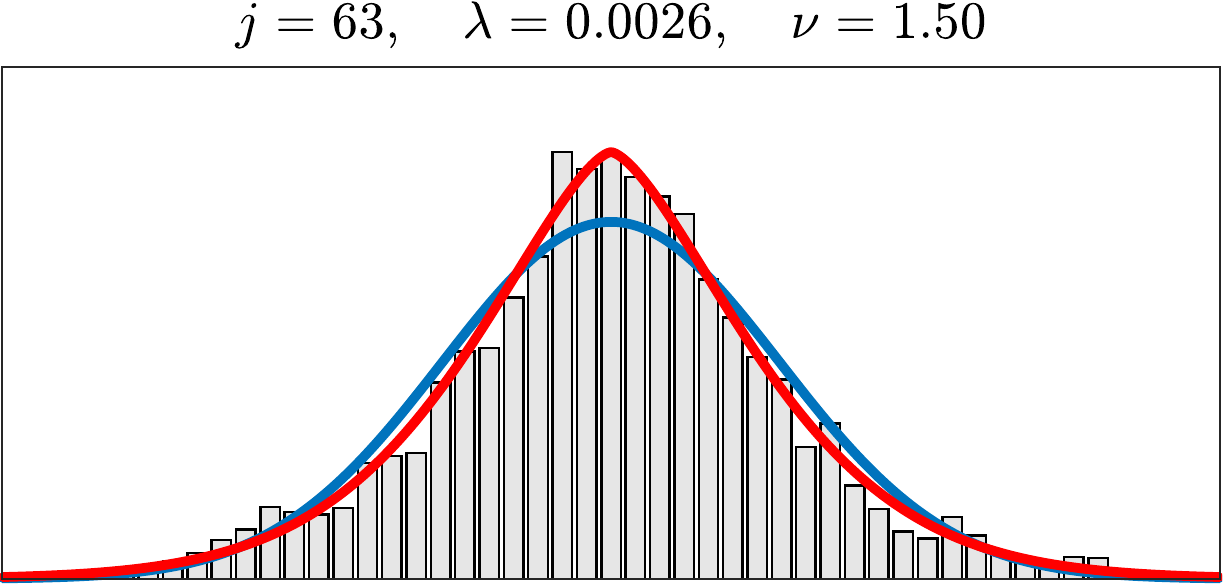}%
  \caption{Histograms of the projection of $200,000$ clean patches on
    6 eigenvectors $j=1, 2, 3, 4, 62$ and $63$
    of the covariance matrix of one component $k$ of the mixture
    (with weight $w_k = 1.3$\%).
    The contribution of each clean patch in the histograms
    is given by its membership values
    onto this component $k$ (as obtained during the E-Step of EM).
    For each histogram, a generalized Gaussian distribution
    was adjusted by estimating the parameters $\lambda$ and $\nu$
    by moment estimation (as obtained during the M-Step of our modified EM exposed in \Cref{sec:learning_ggmm}).
    For comparisons, we have also provided illustrations of the
    best fit obtained with a Gaussian distribution.
  }
  \vskip1em
  \label{fig:ggmm_vs_gmm}
\end{figure}

\paragraph{Contributions}
The goal of this paper is to measure the improvements obtained in image denoising tasks by incorporating a GGMM in EPLL algorithm.
Unlike \cite{niknejad2017class}, that incorporates a GGMM prior in a posterior mean estimator
based on importance sampling, we directly extend the maximum {\it a posteriori}
formulation of Zoran and Weiss \cite{Zoran11} for the case of GGMM priors.
While such a GGMM prior has the ability to capture the underlying distribution
of clean patches more closely, we will show that it introduces
two major computational challenges in this case. 
The first one can be thought of as a classification task in which a
noisy patch is assigned to one of the components of the mixture.
The second one corresponds to an estimation task where a
noisy patch is denoised given that it belongs to one of the components of the mixture.
%
Due to the interaction of the noise distribution with the GGD prior,
we first show that these two tasks lead to 
a group of one-dimensional integration and optimization problems, respectively.
Specifically, for $x \in \RR$, these problems are of the following forms 
\begin{align}
\int_\RR \exp\left(
-\frac{(t - x)^2}{2 \sigma^2} - \frac{|t|^\nu}{\lambda_\nu^\nu}
\right) \; \d t
\quad\qandq\quad
\uargmin{t \in \RR}
\frac{(t - x)^2}{2 \sigma^2} + \frac{|t|^\nu}{\lambda_\nu^\nu}
~,
\end{align}
for some $\nu > 0$, $\sigma > 0$ and $\lambda_\nu > 0$.
In general, they do not admit closed-form solutions but
some particular solutions or approximations
have been derived for the estimation/optimization
problem \cite{moulin1999analysis,chaux2007variational}.
By contrast, up to our knowledge, little is known for approximating
the classification/integration one
(only crude approximations were proposed in \cite{soury2015new}).

Our contributions are both theoretical- and application-oriented.
The major contribution of this paper, which is and theoretical in nature,
is to develop an accurate
approximation for the classification/integration problem.
In particular, we show that our approximation error vanishes
for $x \to 0$ and $x \to \pm\infty$
when $\nu = 1$, see
\Cref{thm:assignment_nu1_left} and
\Cref{thm:assignment_nu1_right}.
We next generalize this result for $\frac23 < \nu < 2$ in
\Cref{thm:assignment_general_left} and
\Cref{thm:assignment_general_right}.
On top of that,
we prove that the two problems enjoy
some important desired properties in
\Cref{prop:disc_basic_prop} and \Cref{prop:shrink_basic_prop}.
These theoretical results allow the two quantities to be approximated
by functions that can be quickly evaluated in order to be incorporated in fast algorithms.
Our last contribution is experimental and concerns
the performance evaluation 
of the proposed model
in image denoising scenario.
For reproducibility, we have released
our implementation at \url{\ourcodeurl}.

\paragraph{Potential impacts beyond image denoising}
It is important to note that
the two main contributions presented in this work,
namely approximations for classification and estimation problems,
are general techniques that are relevant to a wider area
of research problems than image restoration. In particular, our contributions apply
to any problems where (i) the underlying clean data are
modeled by a GGD or a GGMM, whereas (ii) the observed samples
are corrupted by Gaussian noise. They are especially relevant
in machine learning scenarios where a GGMM is
trained on clean data but data provided during testing time is noisy.
That is,
our approximations can be used to extend the applicability of
the aforementioned \textit{clean} GGD/GGM based approaches to the less than ideal
testing scenario where the data is corrupted by noise.
For instance, using the techniques introduced in this paper, one could directly use
the GGD based voice activity model of \cite{gazor2003speech}
into the likelihood ratio test based detector of \cite{gazor2003soft}
(which relies on solutions of the integration problem
but was \textit{limited} to Laplacian distributions).
In fact, we suspect that many studies in signal processing may have limited themselves
to Gaussian or Laplacian signal priors because of the complicated integration
problem arising from the intricate
interaction of GGDs with Gaussian noise.
In this paper, we demonstrate that this difficulty can be efficiently
overcome with our approximations. This leads us to believe that
the impact of the approaches presented in this paper
will not only be useful for image restoration
but also aid a wider field of general signal processing applications.


\paragraph{Organization}
After explaining the considered patch prior based restoration framework
in \Cref{sec:background},
we derive our GGMM based restoration scheme in \Cref{sec:ggmm}.
The approximations of the classification and estimation problems
are studied in \Cref{sec:discrepancy}
and \Cref{sec:shrinkage}, respectively.
Finally, we present numerical experiments and results in \Cref{sec:exp_eval}.


\section{Background}
\label{sec:background}
In this section we provide a detailed overview of
the use of patch-based priors in Expected Patch Log-Likelihood (EPLL) framework
and its usage under GMM priors.

\subsection{Image restoration with patch based priors}
\label{sec:image_restoration}
We consider the problem of estimating an image
$\bu \in \RR^N$ ($N$ is the number of pixels)
from noisy linear observations $\bv = \Aa \bu + \bw$, where
$\Aa: \RR^N \to \RR^M$ is a linear operator
and $\bw \in \RR^M$ is a noise component assumed to be white and Gaussian with variance $\sigma^2$.
In this paper, we will focus on standard denoising problems where
$\Aa$ is the identity matrix, but in more general settings,
it can account for loss of information such as blurring.
Typical examples for operator $\Aa$ are: a low pass filter (for {\it deconvolution}),
a masking operator (for {\it inpainting}), or a projection
on a random subspace
(for {\it compressive sensing}).
To reduce noise and stabilize the inversion of $\Aa$,
some \textit{prior} information is used for the estimation of $\bu$.
Recent techniques \cite{elad2006image,Zoran11,sulam2015expected} include
this {\it prior} information as a model for the distribution of patches found in natural clean images.
We consider the EPLL framework \cite{Zoran11} that
restores an image by maximum {\it{a posteriori}} estimation over all patches,
corresponding to the following minimization problem:
\begin{align}\label{eq:epll}
  \uargmin{\bu \in \RR^n} \frac{P}{2 \sigma^2} \norm{\Aa \bu - \bv}^2 - \sum_{i=1}^N \log p\left( \Pp_i \bu \right)
\end{align}
where $\Pp_i : \RR^N \to \RR^{P}$ is the linear operator
extracting a patch with $P$ pixels centered at the pixel with location $i$
(typically, $P=8 \!\times\! 8$), and $p(.)$ is the {\it a priori} probability
density function (\ie the statistical model of noiseless patches
in natural images).
Since $i$ scans all of the $N$ pixels of the image,
all patches contribute to the loss and many patches overlap.
Allowing for overlapping is important because otherwise there would appear blocking artifacts.
While the first term in eq.~\eqref{eq:epll} ensures that
$\Aa \bu$ is close to the observations $\bv$
(this term is the negative log-likelihood under the white Gaussian noise assumption),
the second term regularizes the solution $\bu$ by favoring an image such that all of its patches fit the {\it prior} model of patches in natural images.

\paragraph{Optimization with half-quadratic splitting}
Problem \eqref{eq:epll} is a large optimization problem where $\Aa$ couples all unknown pixel values of $\bu$ and the patch prior is often chosen non-convex.
Our method follow the choice made by EPLL
of using a classical technique, known as half-quadratic splitting \cite{geman1995nonlinear,krishnan2009fast},
that introduces
$N$ auxiliary unknown vectors $\bz_i \in \RR^P$, and 
alternatively consider the penalized optimization problem, 
for $\beta > 0$, as
\begin{align}\label{eq:qhs}
  \uargmin{\substack{\bu \in \RR^n \\ \bz_1, \ldots, \bz_N \in \RR^P}}
  \frac{P}{2 \sigma^2} \norm{\Aa \bu - \bv}^2
  + \frac{\beta}{2} \sum_{i\in \Ii} \norm{\Pp_i \bu - \bz_i}^2
  -  \sum_{i\in \Ii}\log p\left( \bz_i \right).
\end{align}
When $\beta \to \infty$, the problem \eqref{eq:qhs} is 
equivalent to the original problem \eqref{eq:epll}. In practice, an increasing sequence of $\beta$ is considered, and the optimization is performed by alternating between the minimization for $\bu$ and $\bz_i$.
Though little is known about the convergence of this algorithm,
few iterations produce remarkable results, in practice.
We follow the EPLL settings prescribed in \cite{Zoran11} by performing 5 iterations
of this algorithm with parameter $\beta$ set to $\frac{1}{\sigma^2} \{ 1, 4, 8, 16, 32 \}$ for each
iteration, respectively. The algorithm is initialized using $\hat \bu = \bv$
for the first estimate.
\paragraph*{Minimization with respect to $\bu$}
Considering all $\bz_i$ to be fixed, optimizing \eqref{eq:qhs} for $\bu$
corresponds to solving a linear inverse problem with a Tikhonov regularization.
It has an explicit solution known as the linear minimum mean square estimator
(or often referred to as Wiener filtering) which is obtained as:
\begin{align}\label{eq:qhs_optim_imest_sol}
  \hat \bu &= \uargmin{\bu \in \RR^n}
  \frac{P}{2 \sigma^2} \norm{\Aa \bu - \bv}^2
  + \frac{\beta}{2} \sum_{i\in \Ii} \norm{\Pp_i \bu - \hat \bz_i}^2 \nonumber\\
  &=
  \left(\Aa^t\Aa+\frac{\beta\sigma^2}{P}\sum_{i \in \Ii}\Pp_i^t\Pp_i\right)^{-1} \left(\mathcal{A}^t \bv + \frac{\beta\sigma^2}{P}\sum_{i \in \Ii}\Pp_i^t \hat \bz_i\right),
\end{align}
where $\Pp_i^t\Pp_i$ is a diagonal matrix whose $i$-th diagonal element corresponds to
the number of patches overlapping the pixel of index $i$.

\paragraph*{Minimization with respect to $\bz_i$}
Considering $\bu$ to be fixed, optimizing \eqref{eq:qhs} for $\bz_i$
leads to:
\begin{align}
  \label{eq:qhs_optim_paest}
    \hat \bz_i &\leftarrow \uargmin{\bz_i \in \RR^P}
    \frac{\beta}{2} \norm{\tilde \bz_i - \bz_i}^2
    - \log p(\bz_i)
    \qwhereq
    \tilde \bz_i = \Pp_i \hat\bu~,
\end{align}
which corresponds to the maximum \textit{a posterior} (MAP) denoising problem
under the patch prior $p$ of a patch
$\tilde \bz_i$ contaminated by Gaussian noise with variance $1/\beta$.
The solution of this optimization problem
strongly depends on the properties of the chosen patch prior.


\medskip

Our algorithm will follow the exact same procedure
as EPLL by alternating between eq.~\eqref{eq:qhs_optim_imest_sol}\footnote{%
  Since our study focuses only on denoising, we will consider $\Aa = \Id_N$.}
and \eqref{eq:qhs_optim_paest}.
In our proposed method, we will be using a generalized
Gaussian mixture model (GGMM) to represent patch prior.
The general scheme we will adopt to solve eq.~\eqref{eq:qhs_optim_paest}
under GGMM prior is inspired from the one proposed by Zoran \& Weiss \cite{Zoran11}
in the simpler case of Gaussian mixture model (GMM) prior.
For this reason, we will now introduce this simpler case of GMM prior
before exposing the technical challenges arising from the use of
GGMM prior in \Cref{sec:ggmm}.



\subsection{Patch denoising with GMM priors}
The authors of \cite{Zoran11, yu2012solving} suggested
using a zero-mean Gaussian mixture model (GMM) prior%
\footnote{To enforce the zero-mean assumption, patches are
  first centered on zero, then denoised using eq.~\eqref{eq:qhs_optim_paest},
  and, finally, their initial means are added back.
  In fact, one can show that it corresponds to modeling
  $p(\bz - \bar{\bz})$ with a GMM
  where $\bar z_j = \tfrac1P \sum_i z_i$ for all $1 \leq j \leq P$.
  \label{fn:zero_mean}
}, that, 
for any patch $\bz \in \RR^{P}$, is given by 
\begin{align}
p\left( \bz \right)
=
\sum_{k=1}^K w_k \Nn_P(\bz; 0_P, \bSigma_k)
\end{align}
where $K$ is the number of components,
$w_k>0$ are weights such that $\sum_k w_k\!=\!1$,
and $\Nn_P(0_P, \bSigma_k)$ denotes
the multi-variate Gaussian distribution with zero-mean and
covariance $\bSigma_k \in \RR^{P \times P}$.
A $K$-component GMM prior models image patches as being spread over $K$ clusters that have ellipsoid shapes where each coefficient (of each component) follows a Gaussian distribution, \ie bell-shaped with small tails.
In \cite{Zoran11}, the parameters $w_k$ and $\bSigma_k$ of the GMM
are learned using the Expectation Maximization algorithm \cite{dempster1977maximum}
on a dataset of 2 million clean patches of size $8 \times 8$ pixels that are
randomly extracted from the training images of the Berkeley Segmentation Database (BSDS) \cite{martin2001database}.
The GMM learned in \cite{Zoran11} has $K = 200$ zero-mean Gaussian mixture components.

Due to the multi-modality of the GMM prior,
introducing this prior in eq.~\eqref{eq:qhs_optim_paest} makes
the optimization problem highly non-convex:
\begin{align}\label{eq:patch_denoising_gmm_prior}
    \hat \bz &\leftarrow \uargmin{\bz \in \RR^P}
    \frac{\beta}{2} \norm{\tilde \bz - \bz}^2
    - \log \left[ \sum_{k=1}^K w_k \Nn(\bz; 0_P, \bSigma_k) \right]~.
\end{align}
To circumvent this issue, the EPLL framework introduced by
Zoran et al. \cite{Zoran11} uses an approximation%
\footnote{An alternative investigated in \cite{teodoro2015single}
  is to replace the MAP denoising problem in \eqref{eq:qhs_optim_paest}
  by the minimum mean square error (MMSE) estimator ({\it a.k.a.}, posterior mean).
  The MMSE estimator is defined as an integration problem
  and has a closed-form solution in case of GMM priors.
  In our experimental scenarios, this estimator did not lead
  to significative improvements compared to MAP
  and we thus did not pursue this idea.
}.
%
In a nutshell, the approximated approach of EPLL provides a way to
avoid the
intractability of mixture models, thus making them
available to model patch priors \cite{Zoran11}.
\Cref{fig:epll-scheme} provides an illustration of the EPLL framework
and the steps involved in solving the optimization problem \eqref{eq:patch_denoising_gmm_prior} namely:
\begin{itemize}
\item[(i)] compute the posterior $p(k|\tilde{\bz})$\footnote{In \cite{Zoran11}, this is referred to as ``conditional mixing weight''.} of each Gaussian component, $k=1,2 \dots K$ for the given noisy patch $\tilde \bz$ with assumed noise variance of $1/\beta$,
\item[(ii)] select the component $k^\star$ that best explains the given patch $\tilde \bz$,
\item[(iii)] perform whitening by projecting $\tilde \bz$ over the main directions of that cluster
(given by the eigenvectors of $\bSigma_{k^\star}$), and
\item[(iv)] apply a linear shrinkage on the coefficients with respect to
  the noise variance $1/\beta$ and the spread of the cluster (encoded by the eigenvalues).
\end{itemize}
The details of each of these steps will be discussed in \Cref{sec:den_ggmm}
as part of the development of the proposed model which is more general.


\begin{figure}
\label{fig:epll-scheme}
  \includegraphics[width=\linewidth]{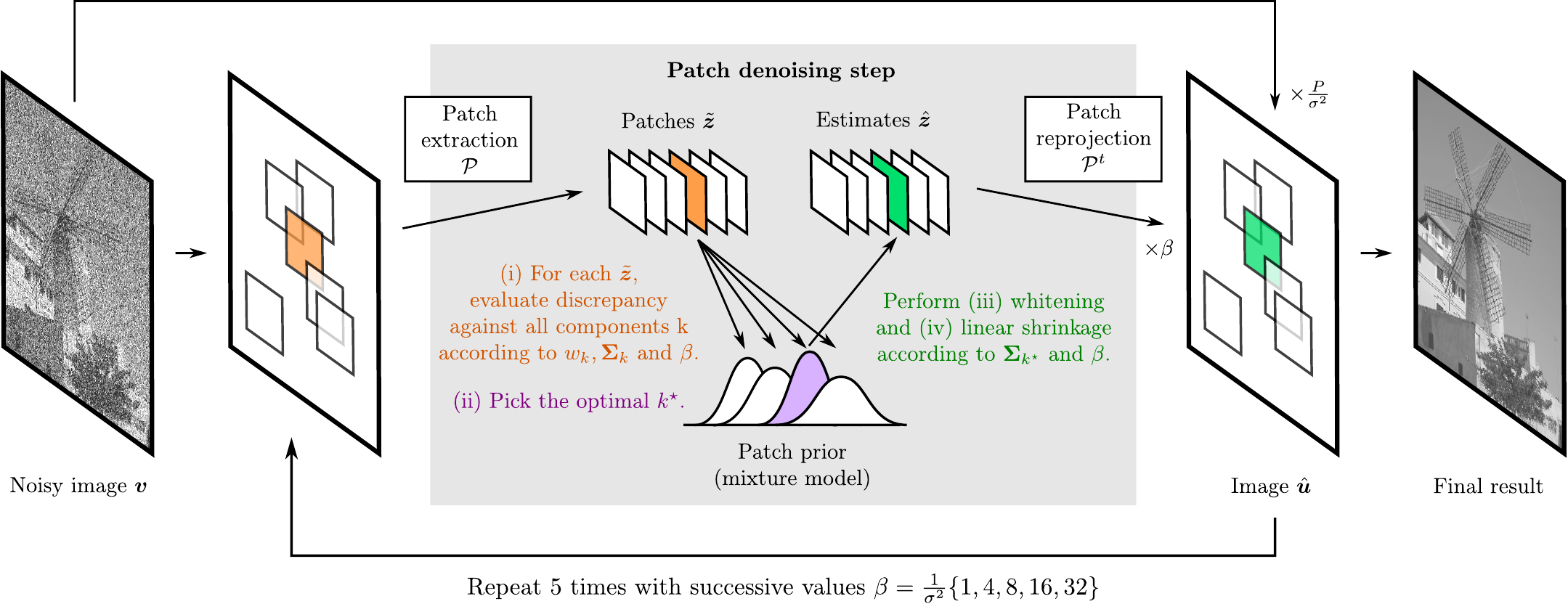}\\[-1em]
  \caption{Illustration of EPLL framework for image denoising with a GMM prior.
    A large collection of (overlapping) patches are first extracted.
    For each patch, an optimal Gaussian component is picked
    based on a measure of discrepancy with the given patch.
    This Gaussian component is next used as a prior model
    to denoise the given patch by linear shrinkage in the corresponding
    eigenspace and depending on $\beta$.
    All estimated patches are finally aggregated together, weighted
    by $\beta$ and combined with the original noisy image to produce a first estimate.
    The procedure is repeated 5 times with increasing values of $\beta$.
  }
  \vspace{1em}
\end{figure}

\medskip

In this paper, we suggest using a mixture of
generalized Gaussian distributions that will enable image patches
to be spread over clusters that are bell shaped in some directions
but can be peaky with large tails in others. 
While the use of GMM priors leads to piece-wise linear estimator (PLE)
as a function of $\bz$ (see \cite{yu2012solving}),
our GGMM prior will lead to 
a piecewise non-linear shrinkage estimator.

\section{Generalized Gaussian Mixture Models}
\label{sec:ggmm}

In this paper, we aim to learn $K$ orthogonal transforms
such that each of them can map a subset (cluster) of clean patches
into independent zero-mean coefficients. Instead of assuming the coefficient
distributions to be bell shaped, we consider that both the scale and
the shape of these distributions may vary from one coordinate to another
(within the same transform).
Our motivation to assume such a highly flexible model is based on the observation illustrated in \Cref{fig:ggmm_vs_gmm}.
Given one of such transform and its corresponding cluster of patches,
we have displayed the histogram of the patch coefficients
for six different coordinates.
It can be clearly observed that the shape of the distribution varies
depending on the coordinate.
Some of them are peaky with heavy tails, and, therefore, would not be faithfully captured by a Gaussian distribution, as done in EPLL \cite{Zoran11}.
By contrast, some others have a bell shape, and so would not be captured properly
by a peaky and heavy tailed distribution, as done for instance by sparse models
\cite{mallat1989theory,moulin1999analysis,aharon2006rm,elad2006image,elad2007analysis,rubinstein2013analysis,sulam2015expected}.
This shows that one cannot simultaneously decorrelate and sparsify
a cluster of clean patches for all coordinates.
Since some of the coordinates reveal sparsity while
some others reveal Gaussianity, we propose to use a more
flexible model that can capture such variations.
We propose using a multi-variate zero-mean
generalized Gaussian mixture model (GGMM)
\begin{align}
  p\left( \bz \right) =
  \sum_{k=1}^K w_k \Gg(\bz; 0_P, \bSigma_k, \bnu_k)
  \label{eq:prior}
\end{align}
where $K$ is the number of components and
$w_k>0$ are weights such that $\sum_k w_k\!=\!1$.
The notation $\Gg(0_P, \bSigma, \bnu)$ denotes the $P$-dimensional
generalized Gaussian distribution (GGD)
with zero-mean, covariance $\bSigma \in \RR^{P \times P}$ (symmetric positive definite)
and shape parameter $\bnu \in \RR^P$, whose expression is
\begin{align}
  \label{eq:ggmm_vect}
  &
  \Gg(\bz; 0_P, \bSigma, \bnu)
  =
  \frac{\Kk }{2 |\bSigma_{\bnu}|^{1/2}}
  \exp\left[- {\norm{\bSigma_{\bnu}^{-1/2} \bz}_{\bnu}^{\bnu}} \right]
  \qwithq
  \norm{\bx}_{\bnu}^{\bnu} = \sum_{j=1}^P |x_j|^{\nu_j},
  \\
  \qwhereq &
  \Kk =
  \prod_{j=1}^P
  \frac{\nu_j}{\Gamma(1/\nu_j)}
  \qandq
  \bSigma_{\bnu}^{1/2} =
  \bSigma^{1/2}
  \begin{pmatrix}
    \sqrt{\frac{\Gamma(1/\nu_1)}{\Gamma(3/\nu_1)}}\\
    & \ddots \\
    & & \sqrt{\frac{\Gamma(1/\nu_P)}{\Gamma(3/\nu_P)}}
  \end{pmatrix}
  .
\end{align}
Denoting the eigen decomposition of matrix
$\bSigma$ by $\bSigma = \bU \bLambda \bU^t$
such that $\bU \in \RR^{P \times P}$ is unitary and
$\bLambda = \diag(\lambda_1, \lambda_2, \ldots, \lambda_P)^2$ is
diagonal with positive diagonal elements $\lambda_j^2$, 
$\bSigma^{1/2}$ in the above expression is defined
as $\bSigma^{1/2} = \bU \bLambda^{1/2}$
and $\bSigma^{-1/2} =  \bLambda^{-1/2} \bU^t$ is its inverse.

When $\bnu$ is a constant vector with all entries
equal to $\nu_j=2$, $\Gg(0_P, \bSigma, \bnu)$
is the multi-variate Gaussian distribution $\Nn(0_P, \bSigma)$ (as used in EPLL \cite{Zoran11}).
When all $\nu_j=1$, it
is the multi-variate Laplacian distribution and the subsequent
GGMM is a Laplacian Mixture Model (LMM).
When all $\nu_j<1$, it
is the multi-variate hyper-Laplacian distribution and the subsequent
GGMM is a hyper-Laplacian Mixture Model (HLMM).
Choosing $K=1$ with a constant vector $\bnu$ corresponds to
$\ell_\nu$ regularization \cite{elad2007analysis,rubinstein2013analysis}.
But as motivated earlier, unlike classical multivariate GGD models
\cite{boubchir2005multivariate,pascal2013parameter,niknejad2017class},
we allow for the entries of $\bnu$
to vary from one coordinate $j$ to another.
To the best of our knowledge, 
the proposed work is the first one to consider
this fully flexible model.

\begin{prop}\label{prop:sep}
  The multi-variate zero-mean GGD can be decomposed as
  \begin{align}
    &\Gg(\bz; 0_P, \bSigma, \bnu)
    = \prod_{j = 1}^P \Gg((\bU^t \bz)_j ; 0, \lambda_j, \nu_j)
  \end{align}
  where, for $\bx = \bU^t \bz$, the distribution of each of its components is given as
  \begin{align}
    \nonumber
    &
    \Gg(x ; 0, \lambda, \nu)
    =
    \frac{\kappa}{2 \lambda_\nu}
    \exp\left[- \left(\frac{|x|}{\lambda_\nu}\right)^\nu \right]
    \\
    \nonumber
    \qwhereq &
    \kappa =
    \frac{\nu}{\Gamma(1/\nu)}
    \qandq
    \lambda_\nu = \lambda \sqrt{\frac{\Gamma(1/\nu)}{\Gamma(3/\nu)}}~,
  \end{align}
  where $x \mapsto \Gg(x ; 0, \lambda, \nu)$ is a real, even, unimodal, bounded
  and continuous probability density function. It is also differentiable everywhere
  except for $x=0$ when $\nu \leq 1$.
\end{prop}

The proof follows directly by injecting the eigen decomposition of $\bSigma$
in \eqref{eq:ggmm_vect} and basic properties of $x \mapsto |x|^\nu$.
\Cref{prop:sep} shows that,
for each of the $K$ clusters,
eq.~\eqref{eq:ggmm_vect}, indeed, models
a prior that is separable in a coordinate system
obtained by applying the whitening transform $\bU^t$.
Not only is the prior separable for each coordinate $j$,
but the shape ($\nu_j$) and scale ($\lambda_j$)
of the distribution may vary.

\medskip

Before detailing the usage of GGMM priors in EPLL framework, 
we digress briefly to explain the procedure we used for training
such a mixture of generalized Gaussian distributions
where different scale and shape parameters
are learned for each feature dimension.

\bigskip

\subsection{Learning GGMMs}\label{sec:learning_ggmm}
Parameter estimation is carried out using
a modified version of the Expectation-Maximization (EM) algorithm \cite{dempster1977maximum}.
EM is an iterative algorithm that performs
at each iteration two steps, namely
Expectation step (E-Step) and Maximization stop (M-Step),
and is known to monotonically increase the model
likelihood and converge to a local optimum.
For applying EM to learn a GGMM, we leverage standard strategies used for parameter
estimation for GGD and/or GGMM that are reported in previous works
\cite{mallat1989theory,birney1995modeling,sharifi1995estimation,aiazzi1999estimation,dominguez2003practical,kokkinakis2005exponent,boubchir2005multivariate,mohamed2009generalized,pascal2013parameter,krupinski2015approximated}.
Our M-Step update for the shape parameter $\nu$ is inspired from
Mallat's strategy \cite{mallat1989theory}
using statistics of the first absolute and second moments of GGDs.
Since this strategy uses the method of moments for $\nu$
instead of maximum likelihood estimation,
we refer to this algorithm as modified EM and the M-Step as Moment step.
We also noticed that shape parameters $\nu < .3$ lead to numerical issues
and $\nu > 2$ leads to local minima with several
degenerate components. 
For this reason,  at each step, we impose
the constraint that the learned shape parameters
satisfy $\nu \in [.3, 2]$. This observation is
consistent with earlier works that have attempted to learn
GGMM shape parameters from data \cite{roenko2014estimation}.
Given $n$ training clean patches of size $P$
and an initialization for the $K$ parameters $w_k > 0$, $\bSigma_k \in \RR^{P \times P}$ and $\bnu_k \in \RR^P$, for $k=1,\ldots,K$,
our modified EM algorithm iteratively
alternates between the following two steps:
\medskip
\begin{itemize}
\item {\bf Expectation step (E-Step)}
\begin{itemize} [leftmargin=*]
\item For all components $k=1,\ldots,K$ and training samples $i=1,\ldots,n$, compute:

  \begin{align*}
    \xi_{k,i} \leftarrow \frac{w_k \Gg(\bz_i; 0_P, \bSigma_k, \bnu_k)}{\sum_{l=1}^K w_l \Gg(\bz_i; 0_P, \bSigma_l, \bnu_l)}~.
  \end{align*}
\end{itemize}

\item {\bf Moment step (M-Step)}
\begin{itemize} [leftmargin=*]
\item For all components $k=1,\ldots,K$, update:
  \begin{align*}
    w_k & \leftarrow \frac{\sum_{i=1}^n \xi_{k,i}}{\sum_{l=1}^K \sum_{i=1}^n \xi_{l,i}}
    \qandq
    \bSigma_{k} \leftarrow \frac{\sum_{i=1}^n \xi_{k,i} \bz_i \bz_i^t}{\sum_{i=1}^n \xi_{k,i}}~.
  \end{align*}
\item Perform eigen decomposition of $\bSigma_{k}$: $$\bSigma_{k} = \bU_{\!k} \bLambda_k \bU_{\!k}^t
  \mbox{ \quad where \quad } \bLambda_k = \diag(\lambda_{k,1}, \lambda_{k,2}, \ldots, \lambda_{k,P})^2~.
  $$
\item For all components $k=1,\ldots,K$ and dimensions $j=1,\ldots,P$, compute:
  \begin{align*}
    \chi_{k,j} &\leftarrow \frac{\sum_{i=1}^n \xi_{k,i} |(\bU_{\!k}^t \bz_i)_j|}{\sum_{i=1}^n \xi_{k,i}}
    \qandq
    (\bnu_{k})_j \leftarrow \Pi_{[.3, 2]}\left[F^{-1}\left(\frac{\chi_{k,j}^2}{\lambda_{k,j}^2}\right)\right]~.
  \end{align*}
\end{itemize}
 \end{itemize}
where $\Pi_{[a, b]}[x] = \min(\max(x, a), b)$
and $F(x) = \frac{\Gamma(2/x)^2}{\Gamma(3/x)\Gamma(1/x)}$
is a monotonic invertible function that was introduced in \cite{mallat1989theory}
(we used a lookup table to perform its inversion as done in
\cite{birney1995modeling,sharifi1995estimation}).
Note that $\chi_{k,j}^2$ and $\lambda_{k,j}^2$ corresponds to
the first absolute and second moments for component $k$ and dimension $j$, respectively.

For consistency purposes, we keep the training data and the number of mixture
components in the models the same as that used in the
original EPLL algorithm \cite{Zoran11}. Specifically, we train our models on
$n= 2$ million clean patches randomly extracted from Berkeley Segmentation Dataset (BSDS) \cite{martin2001database}.
We learn $K = 200$ zero-mean generalized Gaussian mixture components from patches
of size $P = 8 \times 8$.
We opted for a warm-start training
by initializing our GGMM model with the GMM model from \cite{Zoran11} and
with initial values of shape parameters as 2.
We run our modified EM algorithm for 100 iterations.
As observed in \Cref{fig:ggmm_vs_gmm},
the obtained GGMM models the underlying distributions
of a cluster of clean patches much better than a GMM.
In addition, we will see
in \Cref{sec:exp_eval} that our GGMM estimation did not lead to overfitting
as it is also a better fit than a GMM for unseen clean patches.

\subsection{Patch denoising with GGMM priors}\label{sec:den_ggmm}
We now explain why solving step \eqref{eq:qhs_optim_paest} in EPLL
is non-trivial when using a GGMM patch prior.
In this case, for a noisy patch $\tilde\bz$ with variance $\sigma^2$,
equation \eqref{eq:qhs_optim_paest} becomes
\begin{align}\label{eq:patch_denoising_ggmm_prior}
    \hat \bz &\leftarrow \uargmin{\bz \in \RR^P}
    \frac{1}{2 \sigma^2} \norm{\tilde \bz - \bz}^2
    - \log \left[ \sum_{k=1}^K w_k \Gg(\bz; 0_P, \bSigma_k, \bnu_k) \right]~.
\end{align}
As for GMMs, due to the multi-modality of the GGMM prior,
this optimization problem is highly non-convex.
To circumvent this issue, we follow the strategy used by EPLL \cite{Zoran11}
in the specific case of Gaussian mixture model prior.
The idea is to restrict the sum involved in the logarithm in
eq.~\eqref{eq:patch_denoising_ggmm_prior} to only
one component $k^\star$.

If we consider the best $k^\star$ to be given
(the strategy to select the best  $k^\star$ will be discussed
next), then eq.~\eqref{eq:patch_denoising_ggmm_prior} is approximated
by the following simpler problem
\begin{multline}
    \hat \bz \leftarrow \uargmin{\bz \in \RR^P}
    \left\{ \frac{\norm{\tilde \bz - \bz}^2}{2\sigma^2}
    - \log \Gg(\bz; 0_P, \bSigma_{k^\star}, \bnu_{k^\star})
    =
    \frac{\norm{\tilde \bz - \bz}^2}{2\sigma^2}
    + \norm{\bSigma_{\bnu_{k^\star}}^{-1/2} \bz }_\nu^\nu
    \right\}~.
\end{multline}
The main advantage of this simplified version is that, by virtue of \Cref{prop:sep},
the underlying optimization
becomes tractable and can be separated into $P$ one-dimensional
optimization problems, as:
\begin{align}
  &\hat \bz = \bU_{\!k^\star} \hat \bx
  \qwhereq \hat x_j = s(\tilde x_j; \sigma, \lambda_{k^\star,j}, \nu_{k^\star,j})
  \qwithq
  \tilde \bx = \bU_{\! k^\star}^t \tilde \bz\\
  \qandq
  &
  \label{eq:shrinkage_func}
  s(x; \sigma, \lambda, \nu)
  =
  s_{\sigma,\lambda}^\nu(x)
  \in
  \uargmin{t \in \RR}
  \frac{1}{2\sigma^2} (x - t)^2
  + \frac{|t|^\nu}{\lambda_\nu^\nu}
  \qwhereq
  \lambda_\nu = \lambda \sqrt{\frac{\Gamma(1/\nu)}{\Gamma(3/\nu)}}~,
\end{align}
where for all $k$, $\nu_{k,j} = (\bnu_k)_j$ and $\lambda_{k,j} = (\blambda_k)_j$.
While the problem is not necessarily convex, its solution
$s_{\sigma,\lambda}^\nu$ is always uniquely defined almost everywhere (see, \Cref{sec:shrinkage}). We call this almost everywhere real function
$s_{\sigma,\lambda}^\nu : \RR \to \RR$ shrinkage function.
When $\nu = 2$, it is a linear function that is often referred to as Wiener shrinkage.
When $\nu \neq 2$, as we will discuss in \Cref{sec:shrinkage},
it is a non-linear shrinkage function 
that can be computed in closed form for some cases or
with some approximations.

\begin{figure}
  \includegraphics[width=\linewidth]{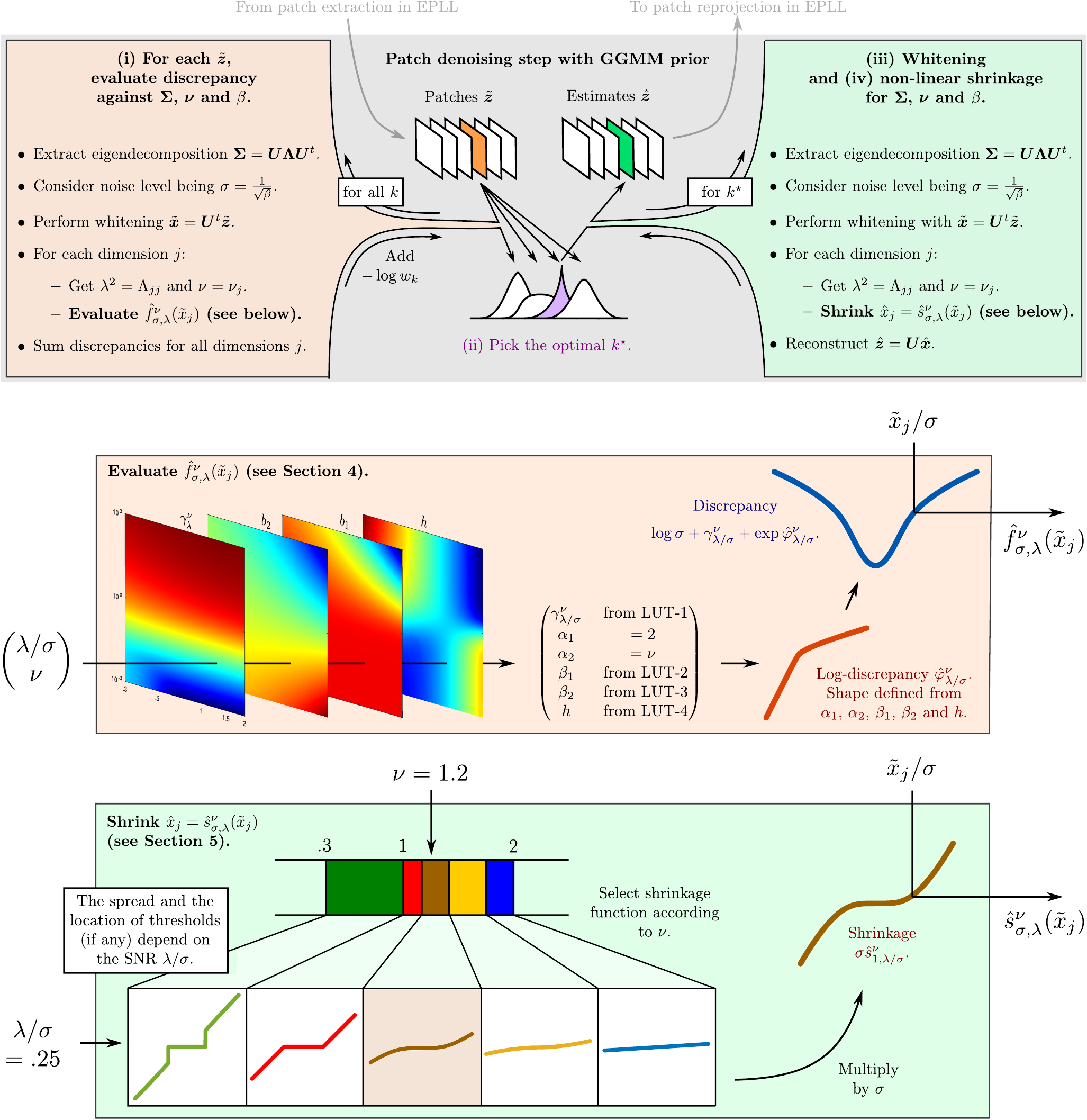}\\[-1em]
  \caption{Illustration of our extension of EPLL to GGMM priors.
    The general procedure, illustrated in the top row,
    is similar to the original EPLL scheme described
    in \Cref{fig:epll-scheme} but relies on generalized Gaussian distributions instead
    of Gaussian distributions.
    The shape of the discrepancy function, illustrated in the second row,
    depends on the given scale and shape parameters ($\lambda$ and $\nu$)
    of the GGD components. In \Cref{sec:discrepancy},
    we will see that it can be approximated based on six parameters,
    four of them retrieved from lookup tables (LUTs).
    Finally, the shrinkage function, illustrated in the bottom row,
    can be non-linear and depends on the selected GGD component.
    In \Cref{sec:shrinkage}, we will see that it can be approximated
    by one of five predefined parametric functions depending on the range
    in which the scale parameter $\nu$ lies.
    The values $\nu = 1.2$ and $\lambda/\sigma = .25$, shown in the bottom row, were chosen
    for the sake of illustration.
  }
  \label{fig:ggmm_epll}
  \vspace{1em}
\end{figure}

Now, we address the question of finding a strategy for choosing a relevant component
$k^\star$ to replace the mixture distribution inside the logarithm.
The optimal component $k^\star$ can be obtained by maximizing the posterior as
\begin{align}
  k^\star
  &\in \uargmax{1 \leq k \leq K}
  p(k \;|\; \tilde\bz)
  =
  \uargmax{1 \leq k \leq K}
  w_{k} p(\tilde \bz \;|\; k)
  =
  \uargmin{1 \leq k \leq K}
  -\log w_{k}
  -\log p(\tilde \bz \;|\; k)
\end{align}
where the weights
of the GGMM corresponds to the prior probability $w_{k} = p(k)$.
We next use the fact that the patch $\tilde \bz$ (conditioned on $k$)
can be expressed as $\tilde \bz = \bz + \bn$ where $\bz$ and $\bn$ are
two independent random variables from distributions
$\Gg(0_P, \bSigma_k, \bnu_k)$ and $\Nn(0_P, \sigma^2\Id_P)$
respectively. It follows that the distribution of $\tilde \bz$
is the convolution of these latter two, and then
\begin{align}
  -\log p(\tilde \bz \;|\; k)
  =
  -\log
  \int_{\RR^P}
  \Gg(\tilde \bz - \bz; 0_P, \bSigma_k, \bnu_k) \cdot \Nn(\bz; 0_P, \sigma^2\Id_P)
  \;\d \bz~.
\end{align}
We next use \Cref{prop:sep} to separate this integration problem
into $P$ one-dimensional integration problems. We obtain
\begin{align}
  &
  -\log
  p(\tilde \bz \;|\; k)
  =
  \sum_{j=1}^P f((\bU_{\!k}^t \tilde \bz)_j; \sigma, \lambda_{k,j}, \nu_{k,j})
\end{align}
where, for $\tilde \bx = \bU_{\!k}^t \tilde \bz$, the integration problem
of each of its components reads as
\begin{align}
  f(x; \sigma, \lambda, \nu)
  =
  f_{\sigma,\lambda}^\nu(x)
  =
  -\log
  \int_{\RR}
  \Gg(x - t; 0, \lambda, \nu) \cdot \Nn(t; 0, \sigma^2)
  \;\d t~.
  \label{eq:discrepancy_func}
\end{align}
We call the real function
$f_{\sigma,\lambda}^\nu : \RR \to \RR$ the \textit{discrepancy function} which
measures the goodness of fit of a GGD to the noisy value $x$.
When $\nu=2$, this function is quadratic with $x$.
For $\nu \neq 2$, as we will discuss in \Cref{sec:discrepancy}, it is a non-quadratic
function, that can be efficiently approximated based on
an in-depth analysis of its asymptotic behavior.

\bigskip

\Cref{fig:ggmm_epll} illustrates the details of the patch denoising step under
the GGMM-EPLL framework. It shows that the method relies on fast approximations
$\hat{f}_{\sigma,\lambda}^\nu$ and $\hat{s}_{\sigma,\lambda}^\nu$ of the discrepancy
and shrinkage functions, respectively.

\bigskip

The next two sections are dedicated to the analysis and approximations
of the discrepancy function $f_{\sigma,\lambda}^\nu$ and
the shrinkage function $s_{\sigma,\lambda}^\nu$, respectively.

\section{Discrepancy function: analysis and approximations}
\label{sec:discrepancy}

From its definition given in eq.~\eqref{eq:discrepancy_func},
the discrepancy function reads for $\nu > 0$, $\sigma >0$
and $\lambda >0$, as
\begin{align}
  f_{\sigma,\lambda}^\nu(x)
  =&
  -\log \frac{1}{\sqrt{2 \pi} \sigma} \frac{\nu}{2 \lambda_\nu \Gamma(1/\nu)}
  -\log
  \int_{-\infty}^\infty
  \exp\left( -\frac{(x - t)^2}{2 \sigma^2} \right)
  \exp\left[ -\left(\frac{|t|}{\lambda_\nu}\right)^\nu \right] \; \d t~.
\end{align}
It corresponds to the negative logarithm of the distribution
of the sum of a zero-mean generalized Gaussian and
a zero-mean Gaussian random variables.
When $\nu=2$, the generalized Gaussian random variable becomes
Gaussian, and the resulting distribution is also Gaussian
with zero-mean and variance $\sigma^2 + \lambda^2$, and then
  \begin{align}\label{eq:discrepancy_l2}
    f_{\sigma,\lambda}^2(x)
    =
    \frac12\left[
      \log 2\pi
      +
      \log(\sigma^2 + \lambda^2)
      +
      \frac{x^2}{\sigma^2 + \lambda^2}
      \right]~.
  \end{align}
  \begin{rem}
    For $\nu = 2$, a direct consequence of \eqref{eq:discrepancy_l2} is that
    $-\log p(\tilde \bz \;|\; k)$ is as an affine function
    of the Mahanalobis distance between $\tilde \bz$ and $0_P$
    for the covariance matrix $\bSigma_k + \sigma^2 \Id_P$:
    \begin{align}
      -\log
      p(\tilde \bz \;|\; k)
      =
      \frac12\left[ P \log 2 \pi + \log |\bSigma_k + \sigma^2 \Id_P|
        +
        \tilde \bz^t (\bSigma_k + \sigma^2 \Id_P)^{-1} \tilde \bz
        \right]~.
    \end{align}
  \end{rem}
When $\nu=1$, the generalized Gaussian random variable becomes
Laplacian, and the distribution resulting from the convolution
also has a closed form which
leads to the following discrepancy function
\begin{align}\label{eq:discrepancy_l1}
  f_{\sigma, \lambda}^1(x)
  &=
  \log (2 \sqrt{2} \lambda)
  -
  \frac{\sigma^2}{\lambda^2}
  -
  \log
  \left[
    e^{\frac{\sqrt{2} x}{\lambda}}
    \erfc\left(
    \frac{x}{\sqrt{2} \sigma}
    +
    \frac{\sigma}{\lambda}
    \right)
    +
    e^{-\frac{\sqrt{2} x}{\lambda}}
    \erfc\left(
    -\frac{x}{\sqrt{2}\sigma}
    +
    \frac{\sigma}{\lambda}
    \right)
    \right]~,
\end{align}
refer to \Cref{proof:discrepancy_l1} for derivation
(note that this expression is given in \cite{gazor2003soft}).

To the best of our knowledge, there are no simple expressions for
other values of $\nu$. One solution proposed by \cite{soury2015new}
is to express this in terms of the bi-variate Fox-H function \cite{mittal1972integral}. This, rather cumbersome expression, is computationally demanding. In practice, this special function requires numerical integration
techniques over complex lines \cite{peppas2012new}, and is thus difficult to numerically evaluate it efficiently.
Since, in our application, we need to evaluate this function
a large number of times, we cannot utilize this solution.

In \cite{soury2015new}, the authors have also proposed to approximate
this non-trivial distribution by another GGD. For fixed values
of $\sigma$, $\lambda$ and $\nu$,
they proposed three different numerical techniques to estimate its parameters
$\lambda'$ and $\nu'$ that best approximate either
the kurtosis, the tail or the cumulative distribution function.
Based on their approach, the discrepancy function
$f_{\sigma,\lambda}^\nu(x)$ would thus be a power function of the form $|x|^{\nu'}$.

In this paper, we show that $f_{\sigma,\lambda}^\nu$ does, indeed, asymptotically behave
as a power function for small and large values of $x$,
but the exponent can be quite different for these two asymptotics.
We believe that these different behaviors are important to be preserved
in our application context. For this reason, $f_{\sigma,\lambda}^\nu$ cannot
be modeled as a power function through a GGD distribution.
Instead, we provide an alternative solution that
is able to capture the correct behavior for both of these
asymptotics, and that also permits fast computation.

\bigskip

\subsection{Theoretical analysis}

In this section, we perform a thorough theoretical analysis of the
discrepancy function, in order to approximate it accurately.
Let us first introduce some basic properties
regarding the discrepancy function.
\begin{prop}\label{prop:disc_basic_prop}
  Let $\nu >0$, $\sigma>0$, $\lambda >0$
  and $f_{\sigma,\lambda}^\nu$ as defined in eq.~\eqref{eq:discrepancy_func}.
  The following relations hold true
  \begin{align}
    \label{eq:disc_reduction}
    \tag{reduction}
     f_{\sigma,\lambda}^\nu(x)
    &=
     \log \sigma
     +
     f_{1,\lambda/\sigma}^\nu(x/\sigma)~,
    \\
    \label{eq:disc_symmetry}
    \tag{even}
    f_{\sigma,\lambda}^\nu(x)
    &= f_{\sigma,\lambda}^\nu(-x)~,\\
    \label{eq:disc_unimodal}
    \tag{unimodality}
    |x| \geq |y| &\Leftrightarrow f_{\sigma,\lambda}^\nu(|x|) \geq f_{\sigma,\lambda}^\nu(|y|)~,
    \\
    \label{eq:disc_bound}
    \tag{lower bound at 0}
    \min_{x\in \RR} f_{\sigma,\lambda}^\nu(x)
    &= f_{\sigma,\lambda}^\nu(0) > -\infty~.
  \end{align}
\end{prop}
The proofs can be found in \Cref{proof:disc_basic_prop}.
Based on \Cref{prop:disc_basic_prop}, we can now express the
discrepancy function $f_{\sigma,\lambda}^\nu(x):\RR\to\RR$ in terms of
a constant $\gamma_\lambda^\nu$ and
another function $\varphi_{\lambda}^{\nu}:\RR_+^* \to \RR$,
both of which can be parameterized by
only two parameters $\lambda>0$ and $\nu>0$, as
\begin{align}\label{eq:f_vs_varphi}
  f_{\sigma,\lambda}^\nu(x)
  &=
    \log \sigma
    +
    \gamma_{\lambda/\sigma}^\nu
    +
    \choice{
    e^{\varphi_{\lambda/\sigma}^\nu(|x/\sigma|)}
    &
    \ifq x \ne 0~,\\
    0
    &
    \otherwise~,
  }
  \\
  \qwhereq
  \varphi_\lambda^\nu(x)
  &=
  \log\left[ f_{1,\lambda}^\nu(x) - \gamma_\lambda^\nu \right]
  \qandq
  \gamma_\lambda^\nu = f_{1,\lambda}^\nu(0)~.
\end{align}
We call $\varphi_{\lambda}^{\nu}$ the log-discrepancy function.

\begin{figure}%
  \includegraphics[width=.32\linewidth,viewport=0 -4 305 345,clip]{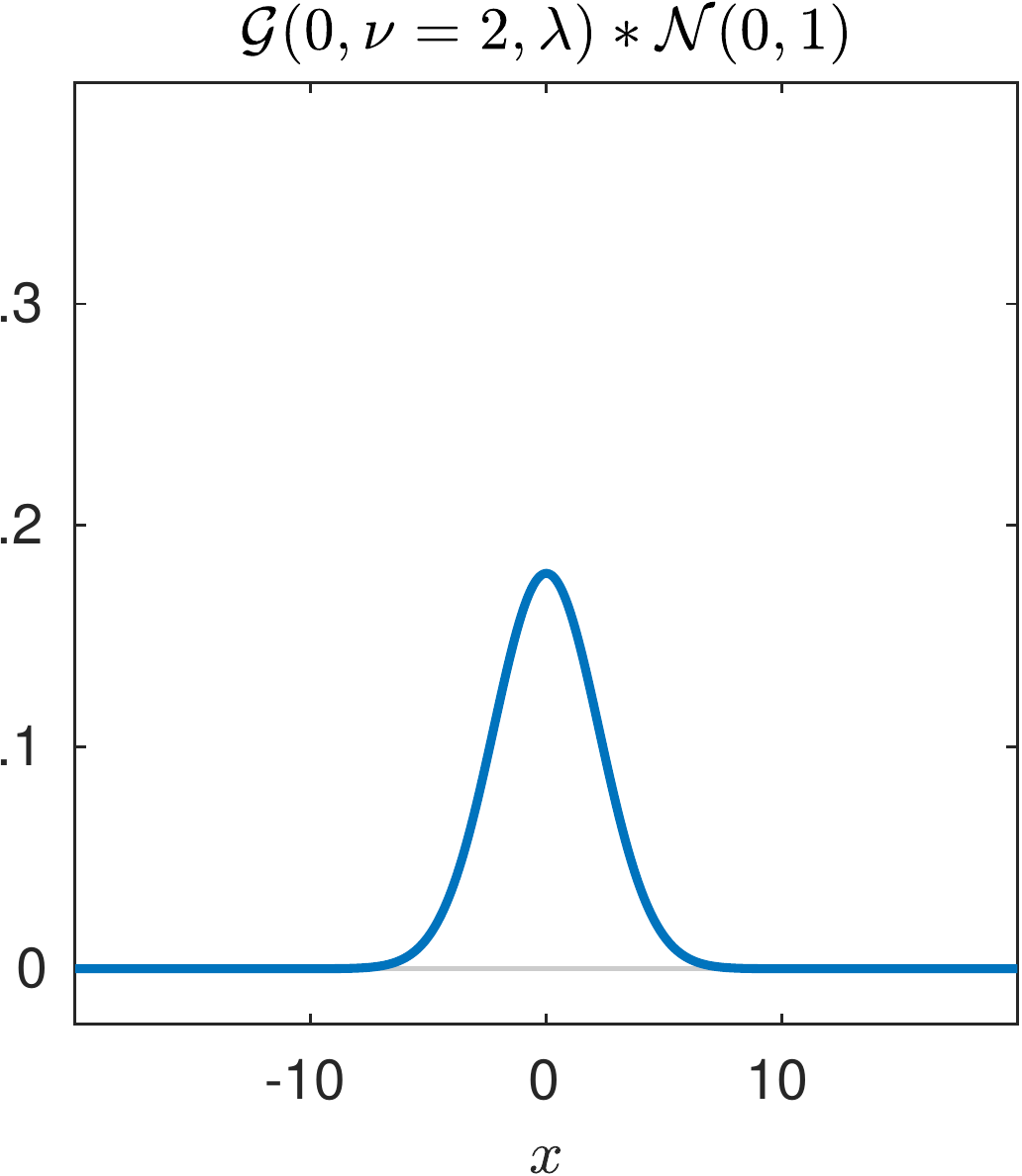}\hfill%
  \includegraphics[width=.32\linewidth,viewport=0 -4 300 345,clip]{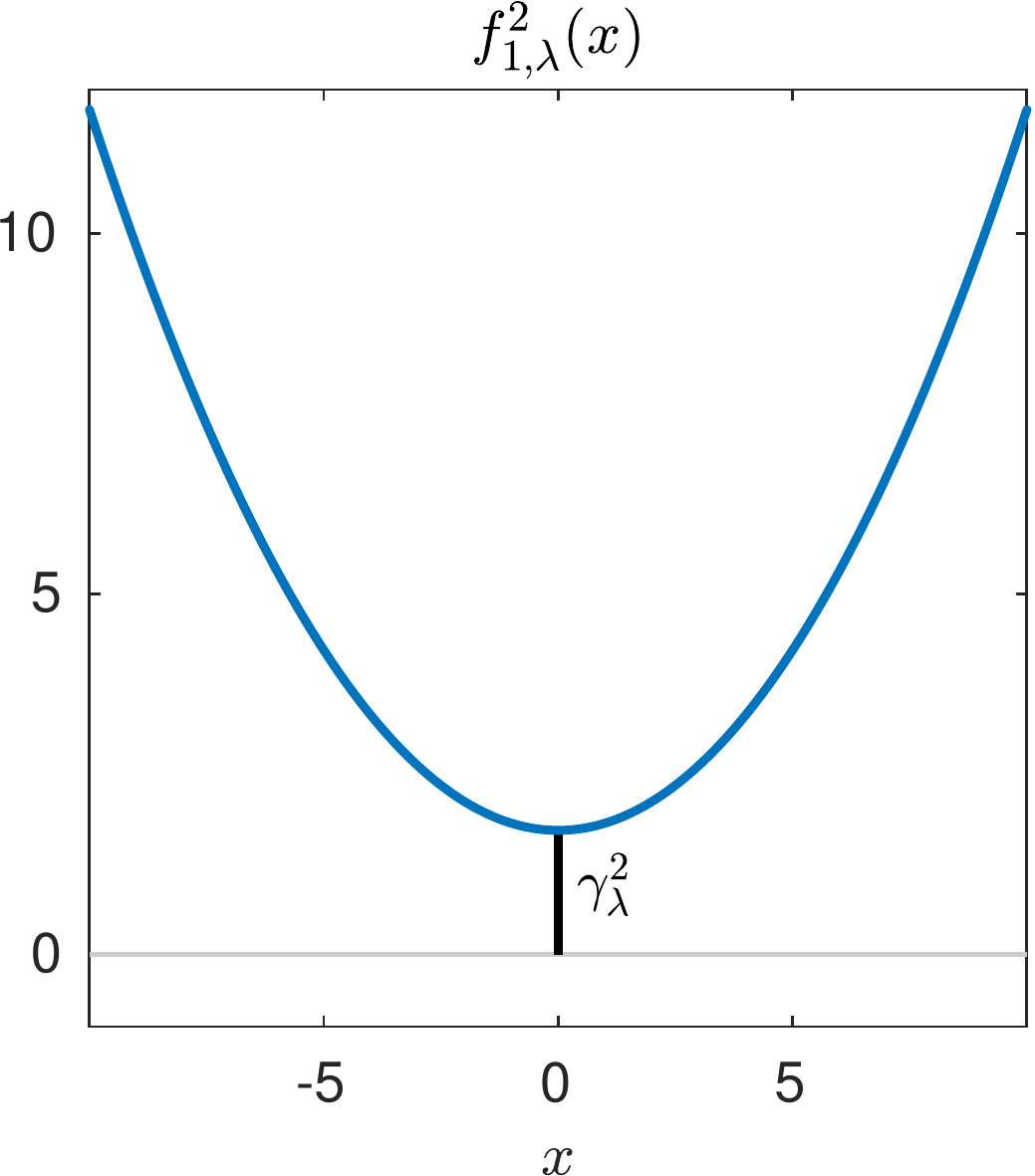}\hfill%
  \includegraphics[width=.32\linewidth,viewport=0 0 300 349,clip]{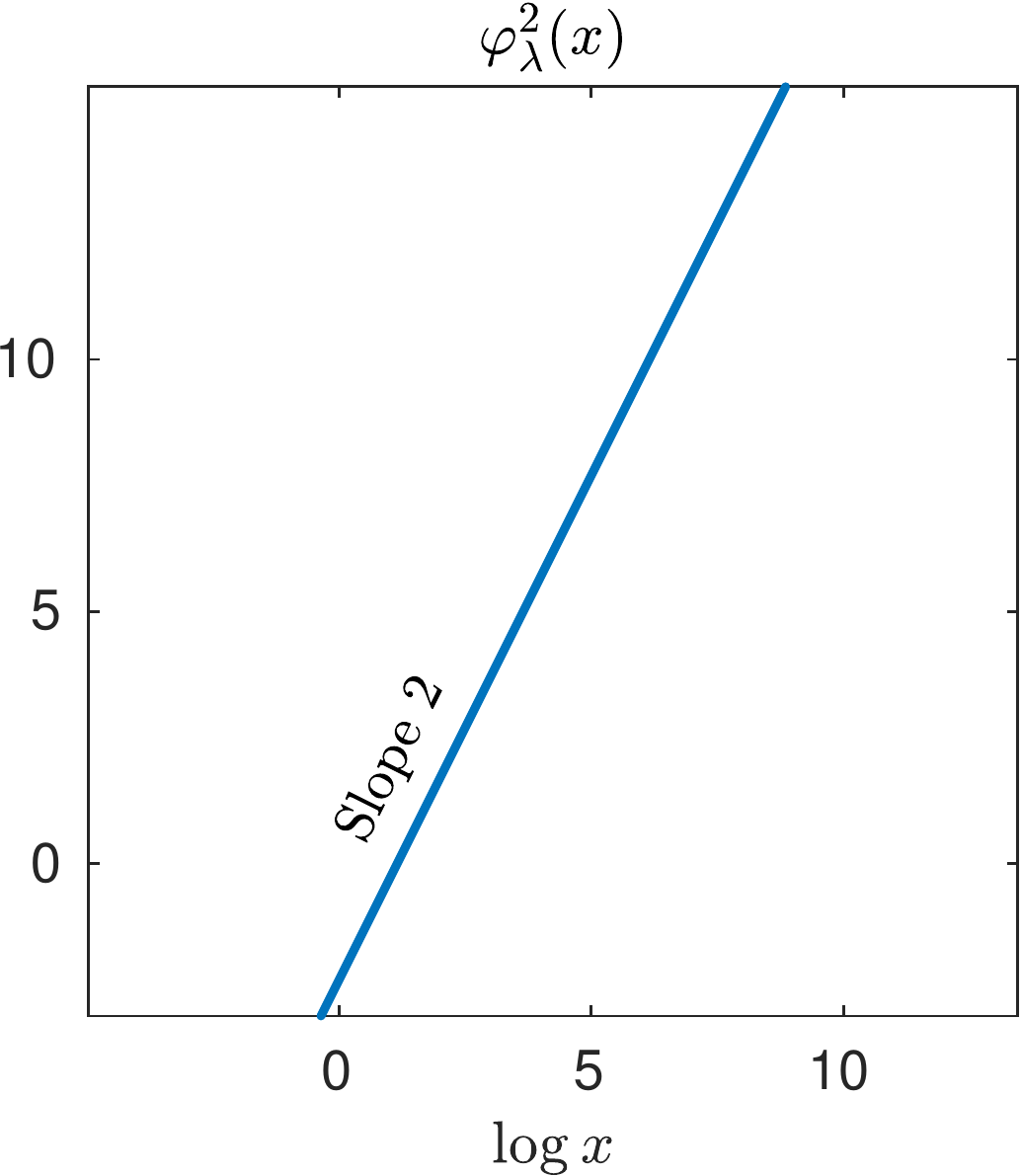}%
  \caption{From left to right: the convolution of a Gaussian distribution
    with standard deviation $\lambda=2$ with a Gaussian distribution
    with standard deviation $\sigma=1$, the corresponding discrepancy
    function and log-discrepancy function.}
  \label{fig:disc_nu2}
  \vskip1em
\end{figure}

At this point, let us consider an instructive toy example for the case when $\nu=2$. In this case, from
eq.~\eqref{eq:discrepancy_l2}, we can deduce that
the log-discrepancy function is a log-linear function
(\ie a linear function of $\log x$)
\begin{align}
  &\varphi^2_\lambda(x)
  =
  \alpha \log x + \beta~,\\
  \qandq
  &\gamma^2_\lambda =
  \tfrac12\left[
    \log 2\pi
    +
    \log(1 + \lambda^2)
    \right]~,
  \\
  \qwhereq
  &\alpha = 2
  \qandq
  \beta = -\log 2 - \log (1 + \lambda^2) ~.
\end{align}
Here, the slope $\alpha = 2$ reveals the quadratic behavior
of the discrepancy function.
\Cref{fig:disc_nu2} gives an illustration
of the resulting convolution (a Gaussian distribution),
the discrepancy function (a quadratic function) and
the log-discrepancy (a linear function with slope $2$).
Note that quadratic metrics are well-known to
be non-robust to outliers, which is in complete
agreement with the fact that Gaussian priors have thin tails.

Another example is the case of $\nu=1$.
From eq.~\eqref{eq:discrepancy_l1},
the log-discrepancy is given by 
\begin{align}
  \label{eq:log_discrepancy_l1}
  \varphi^{1}_\lambda(x)
  &=
    \log \left[
    \log
    \left[
      2
      \erfc{\left(
        \textstyle \frac{1}{\lambda}
      \right)}
      \right]
    -
    \log
    \left[
      e^{\frac{\sqrt{2} x}{\lambda}}
      \erfc{\left(
      \textstyle \frac{x}{\sqrt{2}}
      +
      \frac{1}{\lambda}
      \right)}
      +
      e^{-\frac{\sqrt{2} x}{\lambda}}
      \erfc{\left(
        \textstyle
      -\frac{x}{\sqrt{2}}
      +
      \frac{1}{\lambda}
      \right)}
      \right]
    \right]~,\\
  \qandq
  &\gamma_\lambda^1
  =
  \frac12 \log 2
  +\log \lambda
  -
  \frac1{\lambda^2}
  -
  \log
  \left[
    \erfc{\left(
      \textstyle
      \frac{1}{\lambda}
    \right)}
    \right]~.
\end{align}
Unlike for $\nu = 2$, this function is not log-linear
and thus $f_{\sigma,\lambda}^1$ is not a power function.
Nevertheless, as shown by the next two theorems,
it is also asymptotically log-linear for small and large values of $x$.
\begin{thm}\label{thm:assignment_nu1_left}
  The function $\varphi_\lambda^1$ is
  asymptotically log-linear in the vicinity of $0$
  \begin{align}
    & \varphi_\lambda^1(x)
    \usim{0}
    \alpha_1 \log x + \beta_1~,
    \\
    \qwhereq
    &\alpha_1 = 2
    \qandq
    \beta_1 =
    -\log \lambda
    +
    \log\left[
      \frac{1}{\sqrt{\pi}}
      \frac{
        \exp\left(
        -
        \frac1{\lambda^{2}}
        \right)
      }{
        \erfc\left( \frac{1}{\lambda} \right)
      }
      -
      \frac1{\lambda}
      \right]~.
  \end{align}
\end{thm}
The proof can be found in \Cref{proof:nu1_left}.

\begin{figure}%
  \includegraphics[width=.32\linewidth,viewport=0 -4 305 345,clip]{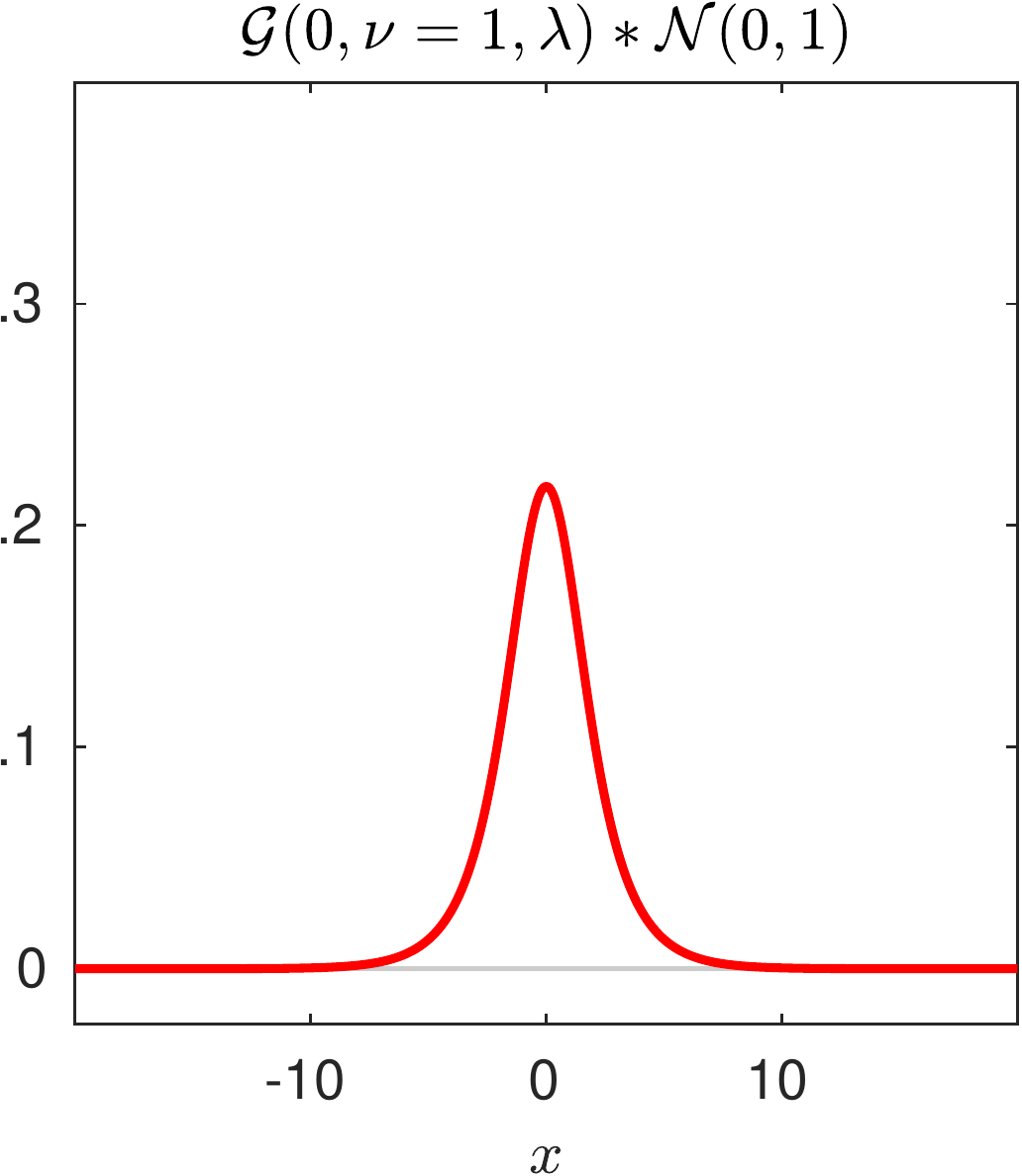}\hfill%
  \includegraphics[width=.32\linewidth,viewport=0 -4 300 345,clip]{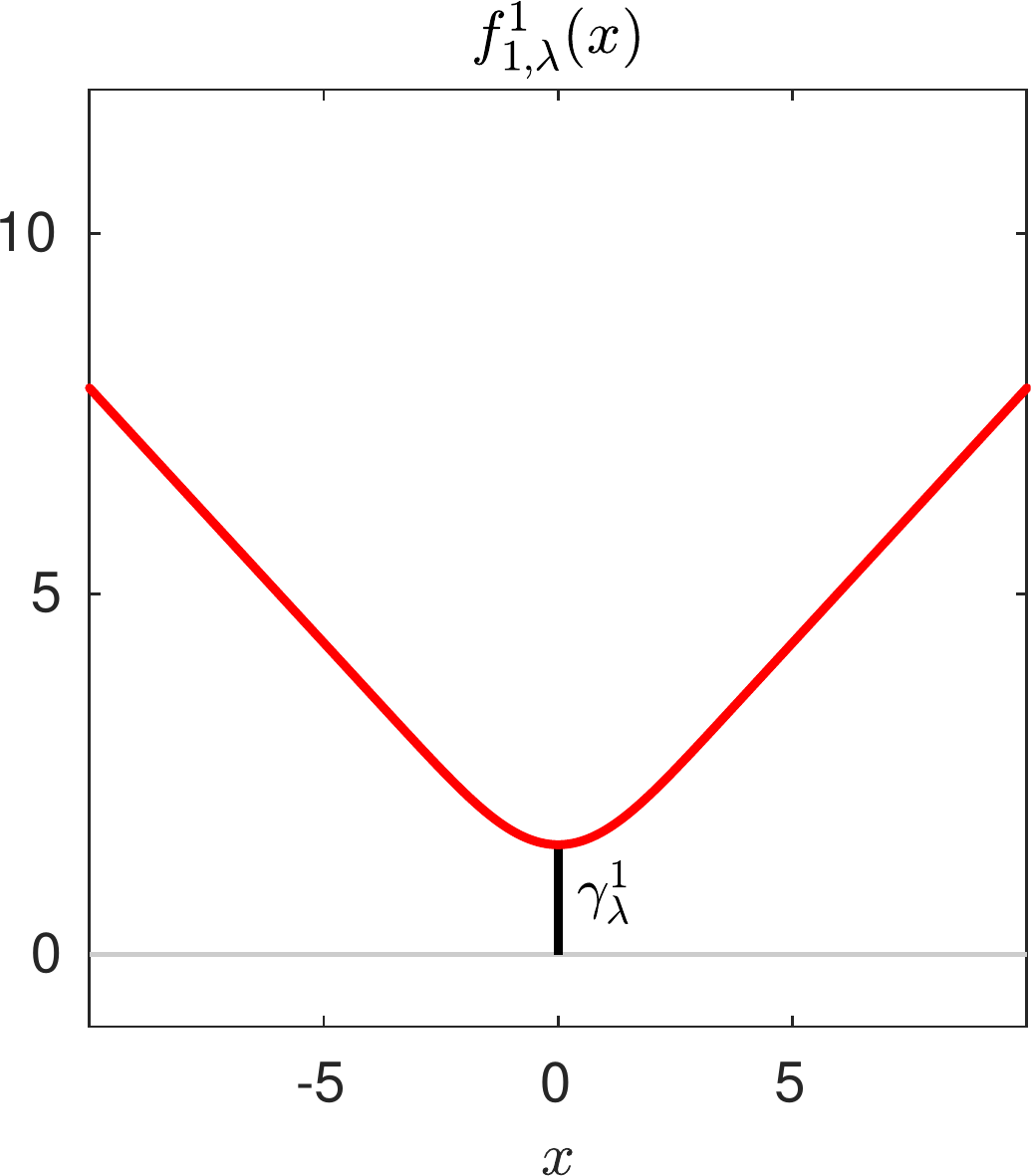}\hfill%
  \includegraphics[width=.317\linewidth,viewport=0 0 300 349,clip]{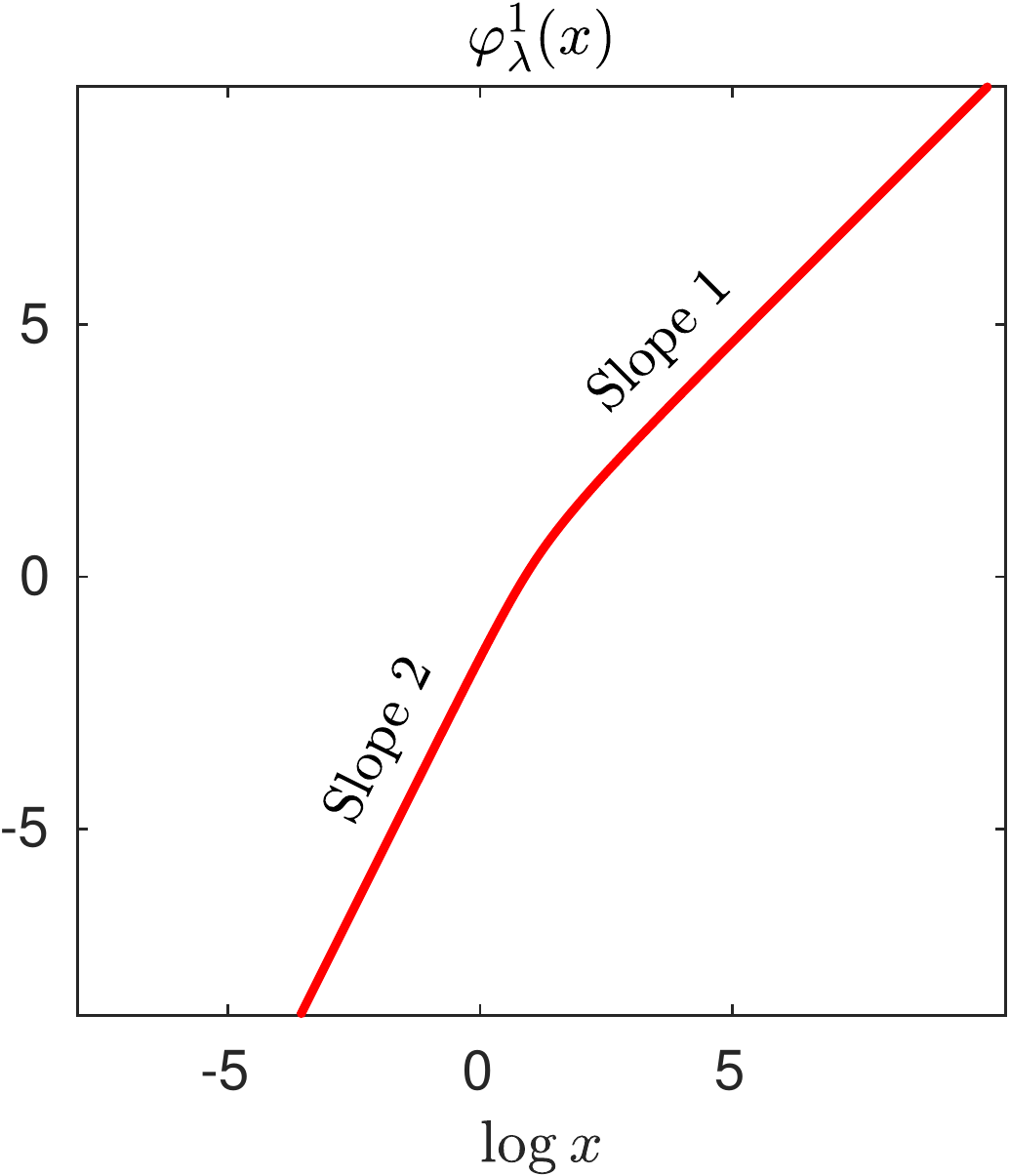}%
  \caption{From left to right: the convolution of a Laplacian distribution
    with standard deviation $\lambda=2$ with a Gaussian distribution
    with standard deviation $\sigma=1$, the corresponding discrepancy
    function and log-discrepancy function.}
  \label{fig:disc_nu1}
  \vskip1em
\end{figure}

\begin{thm}\label{thm:assignment_nu1_right}
  The function $\varphi_\lambda^1$ is
  asymptotically log-linear in the vicinity of $+\infty$
  \begin{align}
    & \varphi_\lambda^1(x)
    \usim{\infty}
    \alpha_2 \log x + \beta_2~,
    \\
    \qwhereq
    &\alpha_2 = 1
    \qandq
    \beta_2 = \frac12 \log 2 - \log \lambda~.
  \end{align}
\end{thm}
The proof can be found in \Cref{proof:nu1_right}.

\Cref{thm:assignment_nu1_left} and
\Cref{thm:assignment_nu1_right} show that
$\varphi_\lambda^1$ has two different asymptotics that can
be approximated by a log-linear function.
Interestingly, the exponent $\alpha_1 = 2$
in the vicinity of $0$ shows that
the Gaussian distribution involved in the convolution
prevails over the Laplacian distribution
and thus, the behavior of $f_{\sigma,\lambda}^1$ is quadratic.
Similarly, the exponent $\alpha_2 = 1$
in the vicinity of $+\infty$ shows that
the Laplacian distribution involved in the convolution
prevails over the Gaussian distribution
and the behavior of $f_{\sigma,\lambda}^1$ is then linear.
These results are supported by \Cref{fig:disc_nu1}
which illustrates the resulting convolution,
the discrepancy function (eq.~\eqref{eq:discrepancy_l1}) and
the log-discrepancy function (eq.~\eqref{eq:log_discrepancy_l1}).
Furthermore, 
the discrepancy function
$f_{\sigma,\lambda}^1$ shares a similar behavior with the well-known
Huber loss function \cite{huber1964robust} (also called smoothed $\ell_1$), known to be
more robust to outliers.
This is again in complete
agreement with the fact that Laplacian priors have heavier tails.

In the case $\frac23 < \nu<2$,
even though $\varphi_\lambda^\nu$ has no simple closed form expression,
the similar conclusions can be made
as a result of the next two theorems.

\begin{thm}\label{thm:assignment_general_left}
  Let $\nu >0$. The function
  $\varphi^\nu_\lambda$ is
  asymptotically log-linear in the vicinity of $0$
  \begin{align*}
    & \varphi^\nu_\lambda(x)
    \usim{0}
    \alpha_1 \log x + \beta_1~,
    \\
    \qwhereq
    &\alpha_1 = 2
    \qandq
    \beta_1 =
    - \log 2
    +
    \log\left(
      1-
    \frac{
      \displaystyle
      \int_{-\infty}^\infty
      t^2
      e^{-\frac{t^2}{2}}
      \exp\left[ -\left(\frac{|t|}{\lambda_\nu}\right)^\nu \right] \; \d t
    }{
      \displaystyle
      \int_{-\infty}^\infty
      e^{-\frac{t^2}{2}}
      \exp\left[ -\left(\frac{|t|}{\lambda_\nu}\right)^\nu \right] \; \d t
    }
    \right)~.
    \end{align*}
\end{thm}
The proof can be found in \Cref{proof:nugen_left}.

\begin{figure}%
  \includegraphics[width=.32\linewidth,viewport=0 20 305 341,clip]{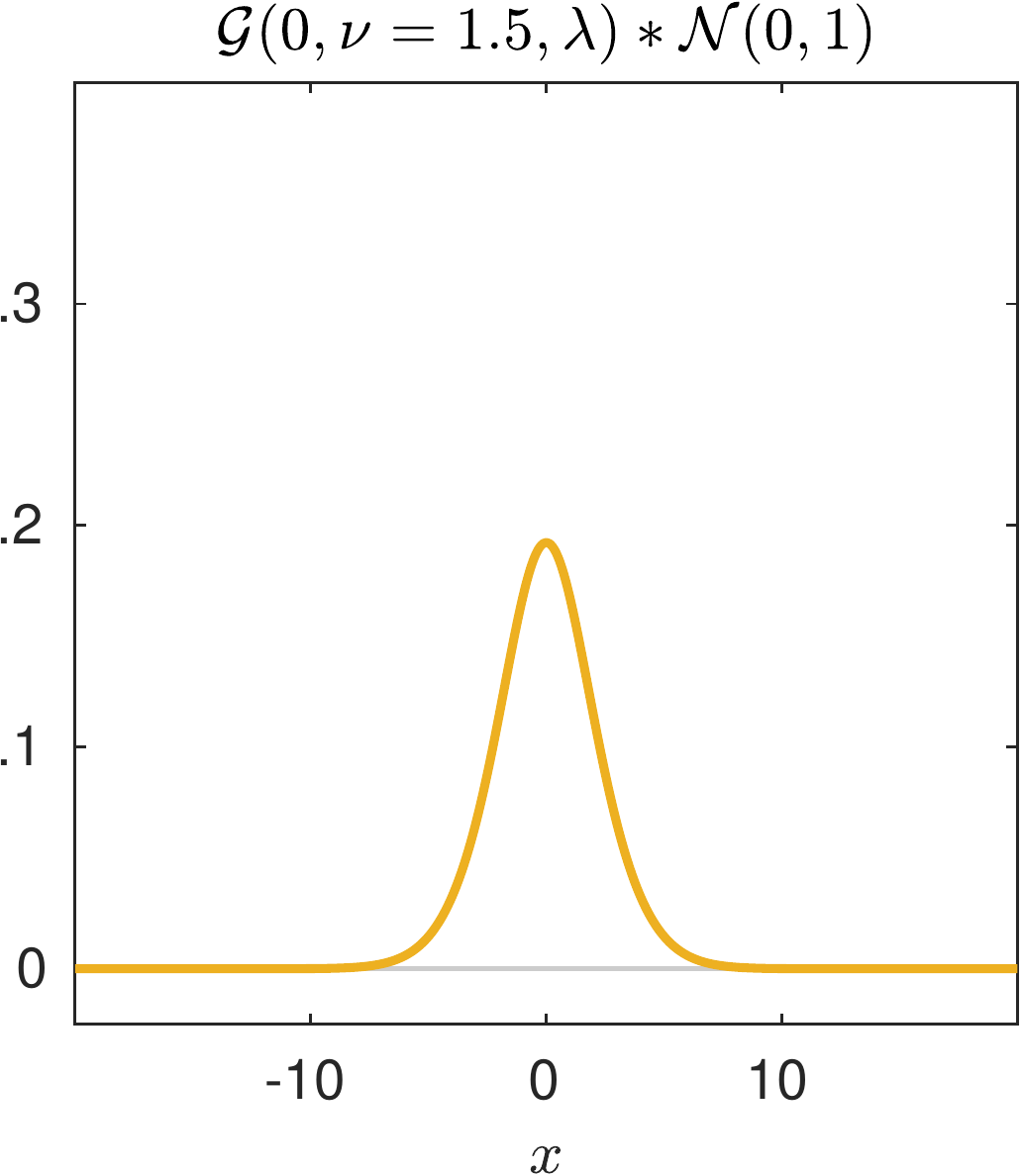}\hfill
  \includegraphics[width=.32\linewidth,viewport=0 20 300 341,clip]{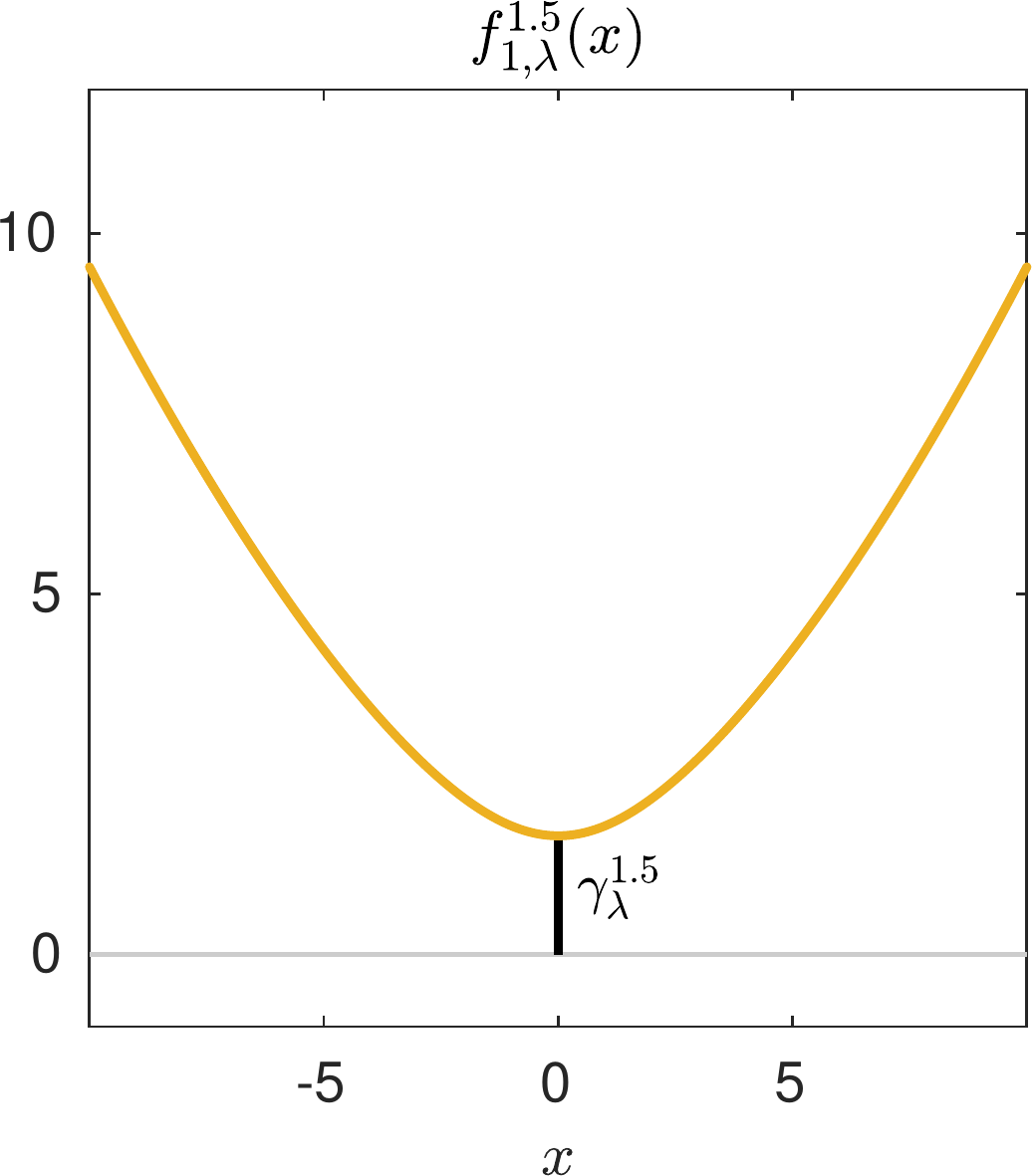}\hfill%
  \includegraphics[width=.317\linewidth,viewport=0 24  300 347,clip]{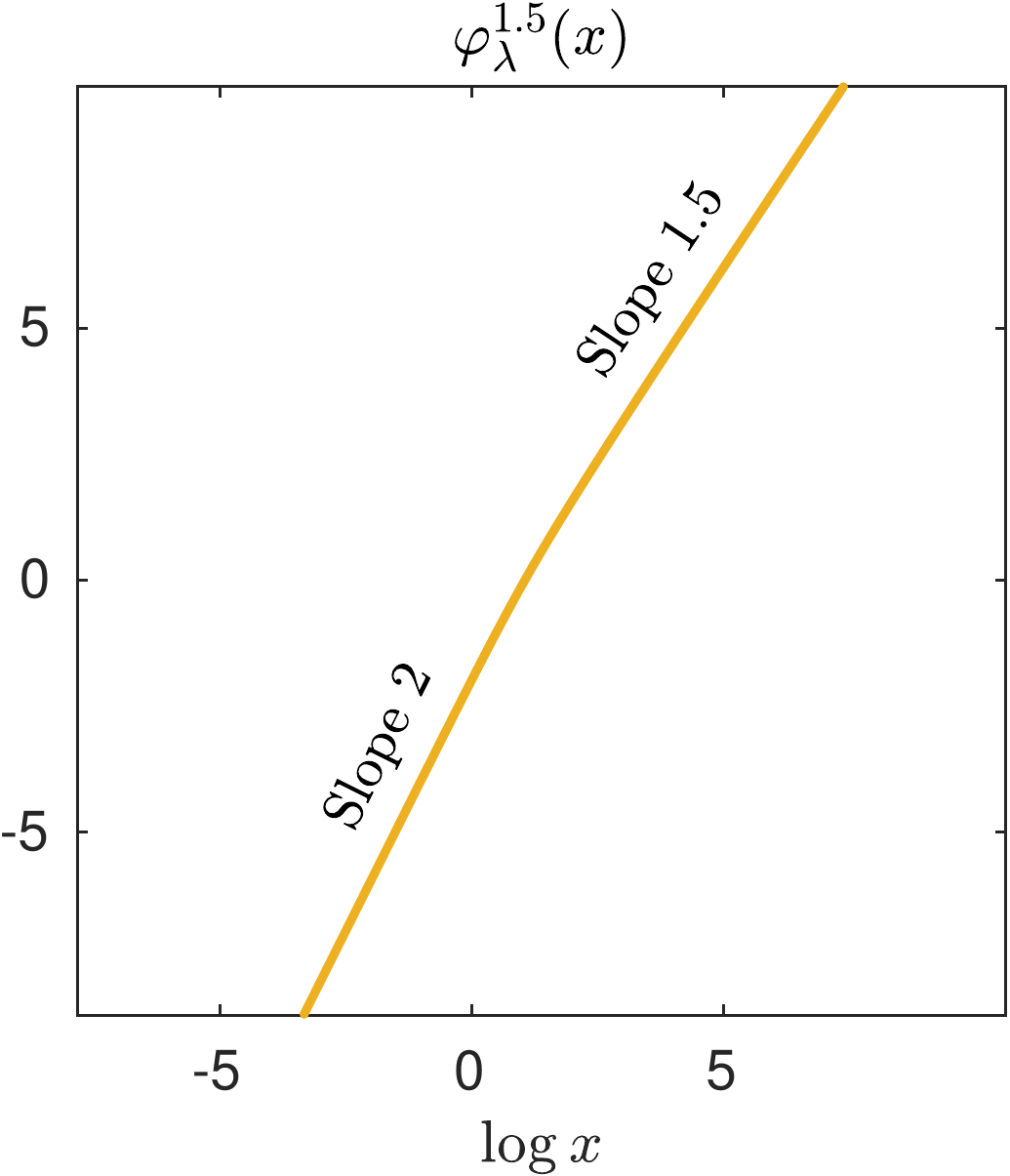}\\[1em]%
\includegraphics[width=.32\linewidth,viewport=0 -4 305 341,clip]{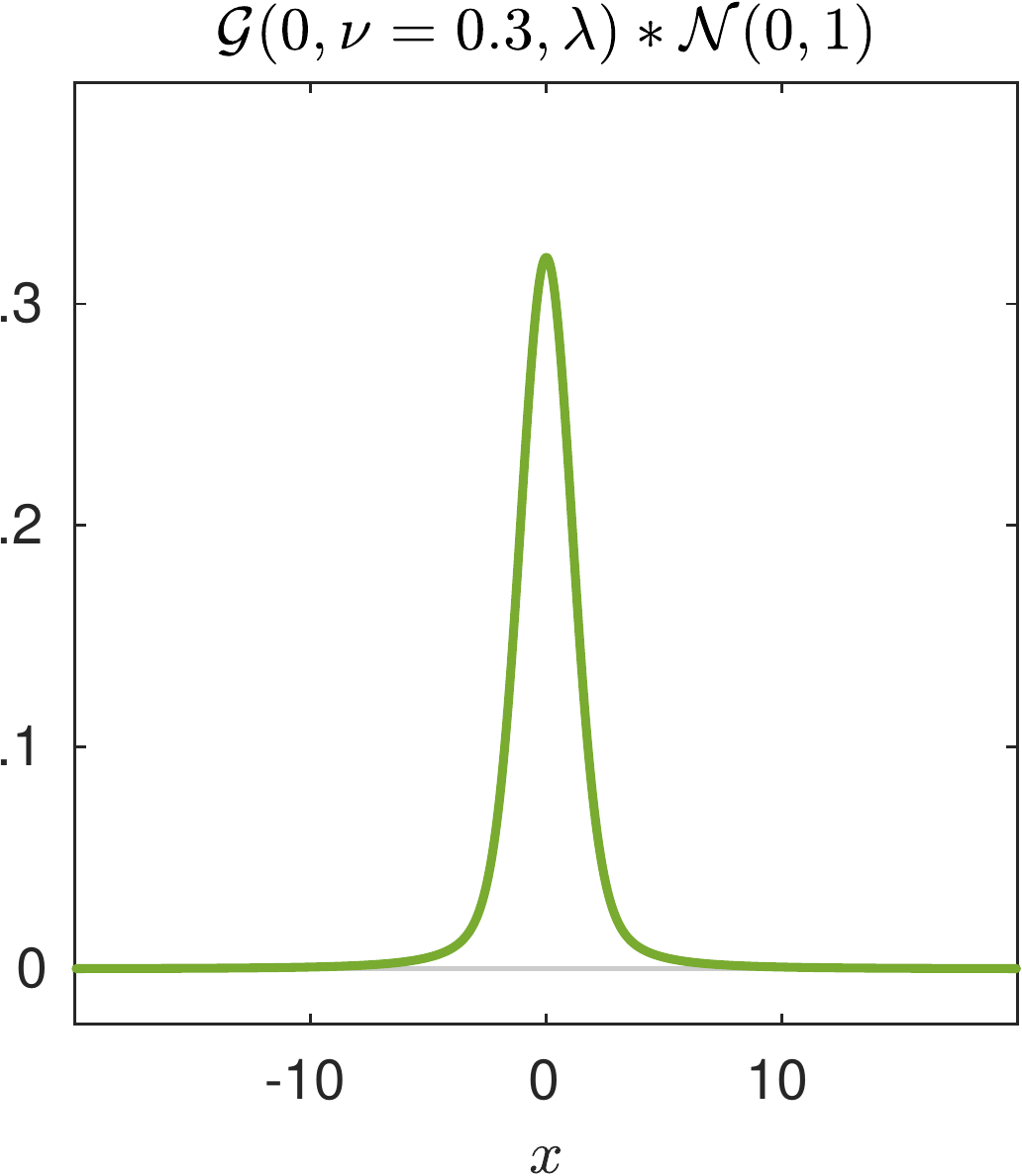}\hfill%
  \includegraphics[width=.32\linewidth,viewport=0 -4 300 341,clip]{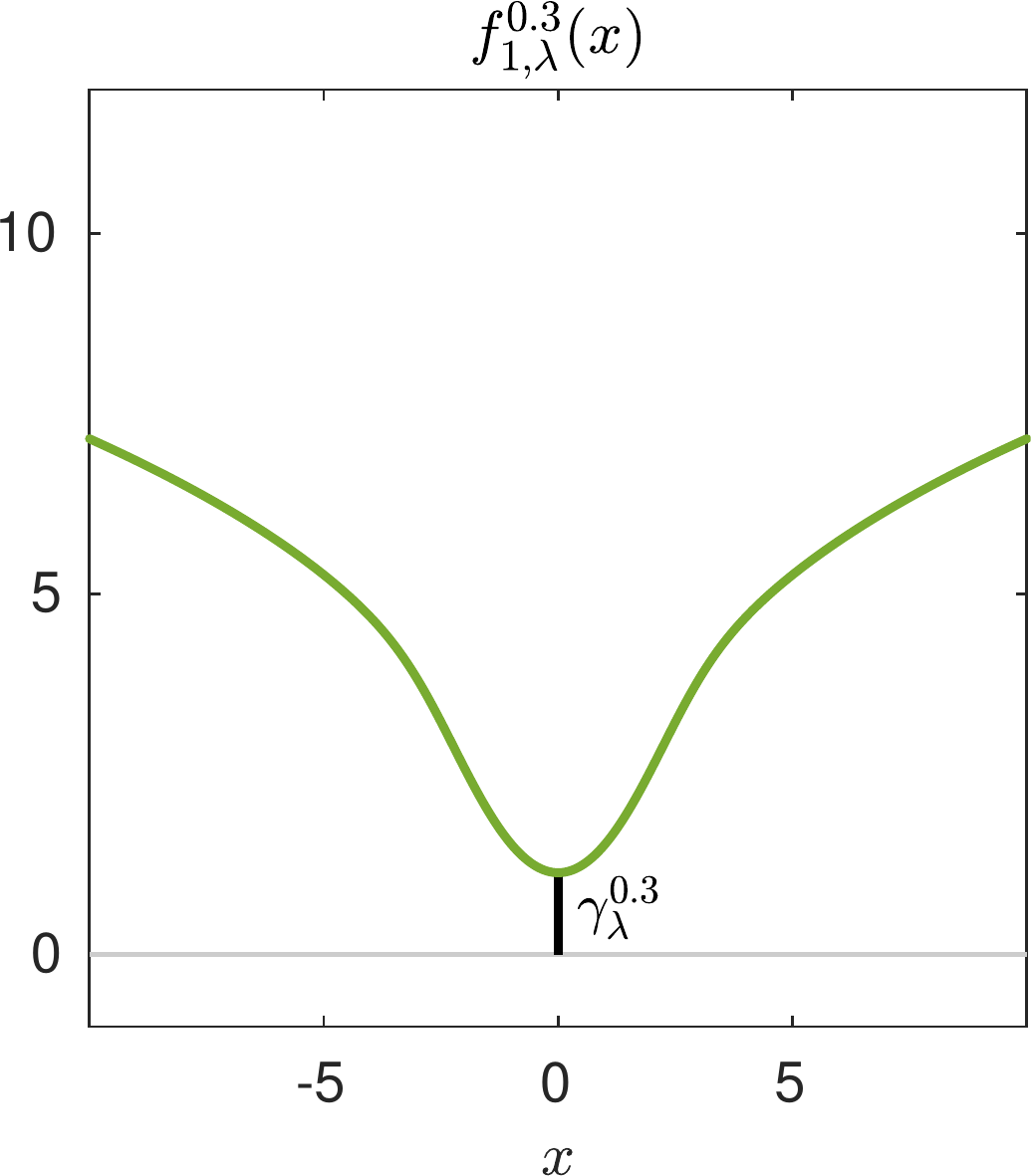}\hfill%
  \includegraphics[width=.317\linewidth,viewport=0 0 300 347,clip]{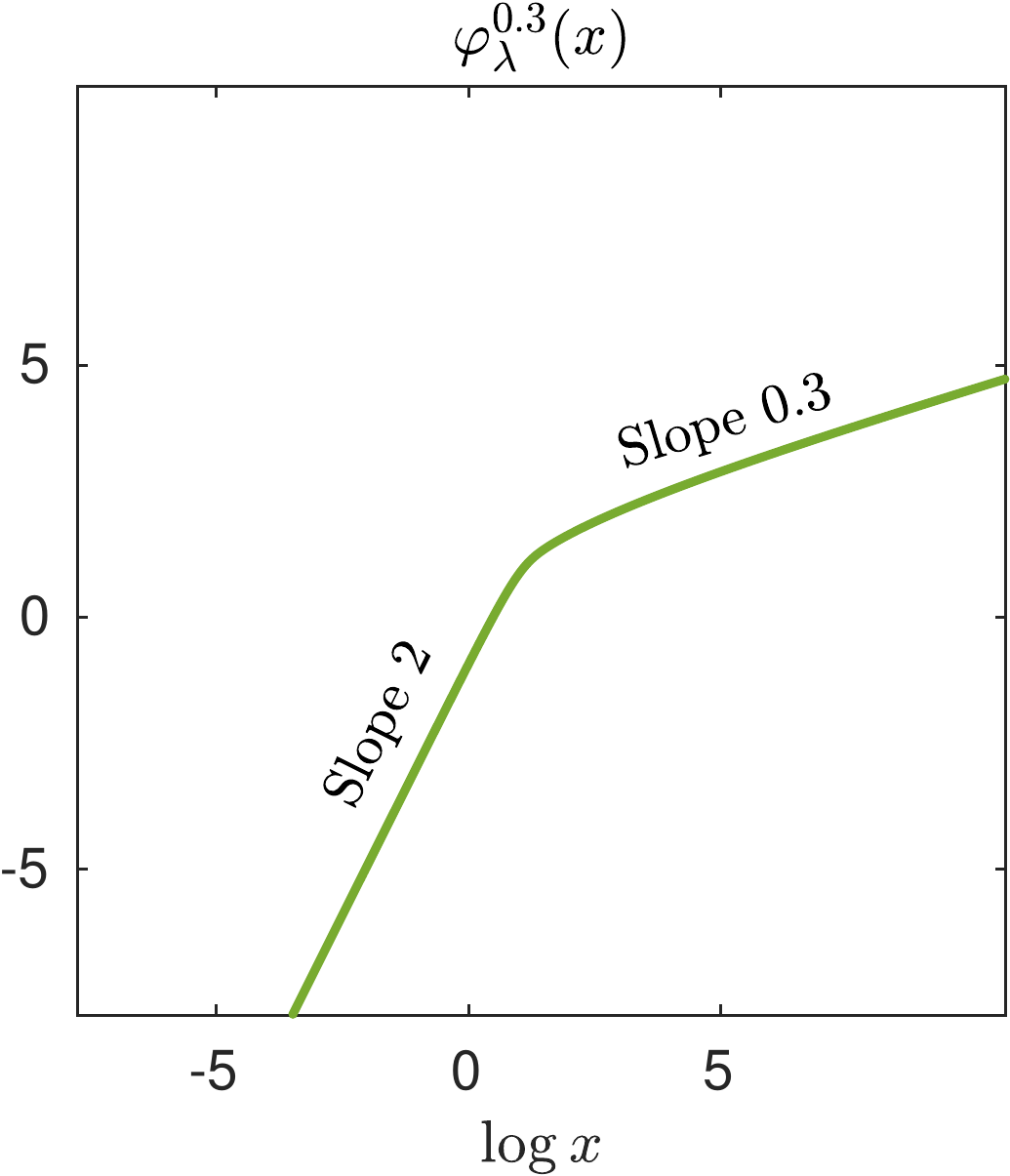}%
  \caption{From left to right: the convolution of a generalized Gaussian distribution
    with standard deviation $\lambda=2$ with a Gaussian distribution
    with standard deviation $\sigma=1$, the corresponding discrepancy
    function and log-discrepancy function.
    From top to bottom: the GGD has a shape parameter $\nu=1.5$ and $.3$, respectively.}
  \label{fig:disc_nuother}
  \vskip1em
\end{figure}

\begin{thm}\label{thm:assignment_general_right}
  Let $\frac23 < \nu < 2$, then
  $\varphi^\nu_\lambda$ is
  asymptotically log-linear in the vicinity of $+\infty$
  \begin{align*}
    & \varphi^\nu_\lambda(x)
    \usim{\infty}
    \alpha_2 \log x + \beta_2~,
    \\
    \qwhereq
    &\alpha_2 = \nu
    \qandq
    \beta_2 =
    -\nu \log \lambda
    -
    \frac{\nu}{2} \log \frac{\Gamma(1/\nu)}{\Gamma(3/\nu)}~.
    \end{align*}
\end{thm}
The proof relies on a result of Berman (1992) \cite{berman1992tail}
and is detailed in \Cref{proof:nugen_right}.

\begin{rem}
  For $\nu > 2$, an asymptotic log-linear
  behavior with $\alpha_2 = 2$ and $\beta_2 = -\log 2$
  can be obtained using exactly the same sketch of proof as the one of
  \Cref{thm:assignment_general_right}.
\end{rem}

\begin{rem}
  For $\nu =2$, we have $\varphi_\lambda^2$ is linear,
  $\beta_1 = -\log 2 - \log \left( 1 + \lambda^2 \right)$
  and $\beta_2 = -\log 2 - \log \lambda^2$,
  which shows that \Cref{thm:assignment_general_right}
  cannot hold true for $\nu=2$.
\end{rem}

\begin{rem}
  For $\nu = 1$, \Cref{thm:assignment_nu1_left} and
  \Cref{thm:assignment_nu1_right} coincide with
  \Cref{thm:assignment_general_left} and
  \Cref{thm:assignment_general_right}.
\end{rem}

\begin{rem}\label{rem:conjecture}
  For $0 < \nu \leq \frac23$, though we did not succeed in proving it,
  our numerical simulations also revealed
  a log-linear asymptotic behavior for $x \to \infty$
  in perfect agreement with the expression of $\alpha_2$ and $\beta_2$
  given in \Cref{thm:assignment_general_right}.
\end{rem}



Again, the exponent $\alpha_1 = 2$
in the vicinity of $0$ shows that
the Gaussian distribution involved in the convolution
prevails over the generalized Gaussian distribution
and the behavior of $f_{\sigma,\lambda}^\nu$ is then quadratic.
Similarly, the exponent $\alpha_2 = \nu$
in the vicinity of $+\infty$ shows that
the generalized Gaussian distribution involved in the convolution
prevails over the Gaussian distribution
and the behavior of $f_{\sigma,\lambda}^\nu$ is then
a power function of the form $x^\nu$.
These results are supported by \Cref{fig:disc_nuother}
that illustrates the resulting convolution,
the discrepancy function and
the log-discrepancy function for $\nu=1.5$ and $\nu=.3$.
Moreover, 
the discrepancy function
$f_{\sigma,\lambda}^\nu$ with $\nu \leq 1$ shares a similar behavior with well-known
robust M-estimator loss functions \cite{huber2011robust}.
In particular, the asymptotic case for $\nu \to 0$ resembles the
Tukey's bisquare loss, known to be insensitive to outliers.
This is again in complete
agreement with GGD priors having larger tails
as $\nu$ goes to $0$.

\begin{figure}
  {\includegraphics[height=.39\linewidth]{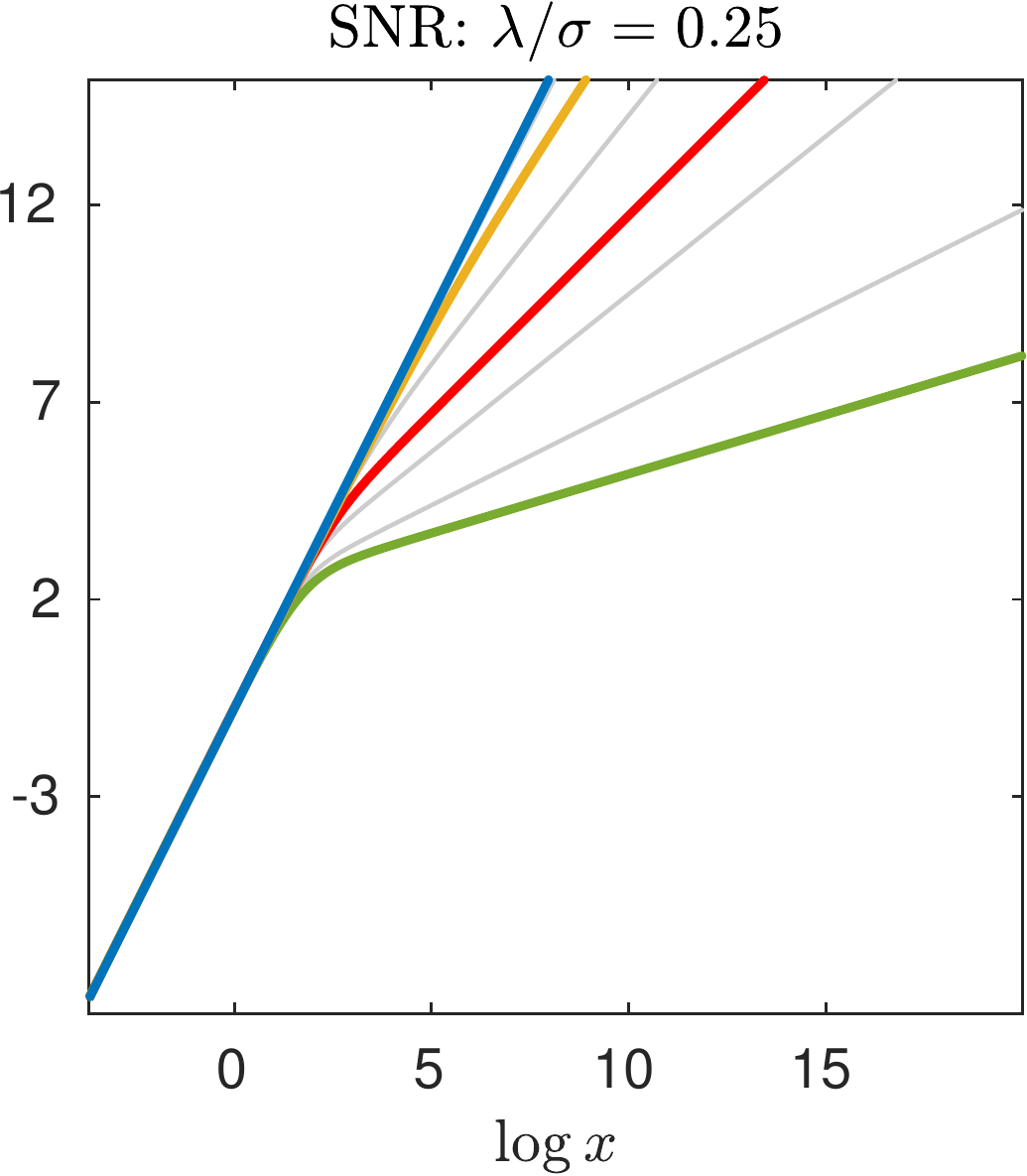}}\hfill%
  {\includegraphics[height=.39\linewidth,viewport=25 0 298 341,clip]{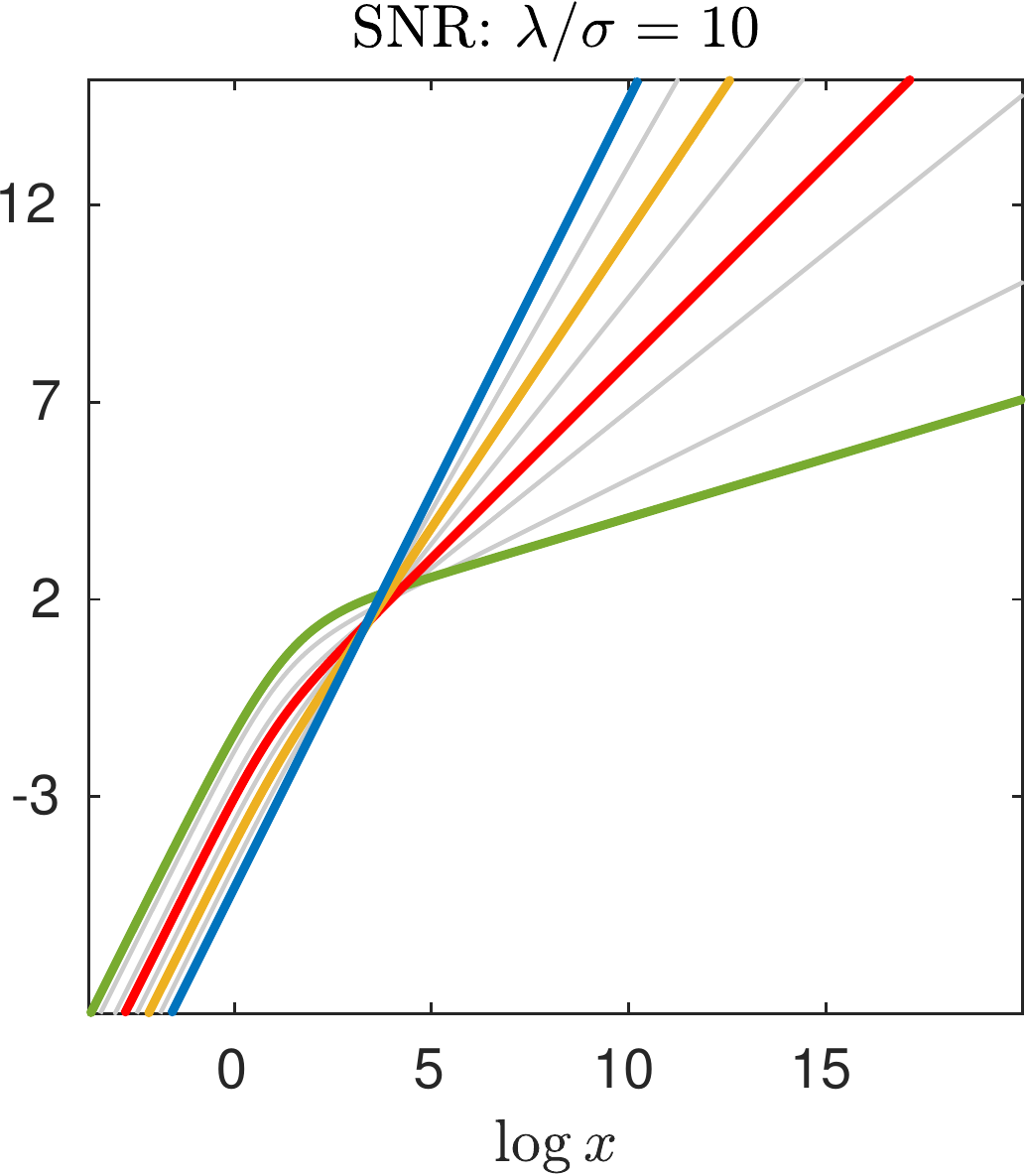}}\hfill%
  {\includegraphics[height=.39\linewidth,viewport=25 0 298 341,clip]{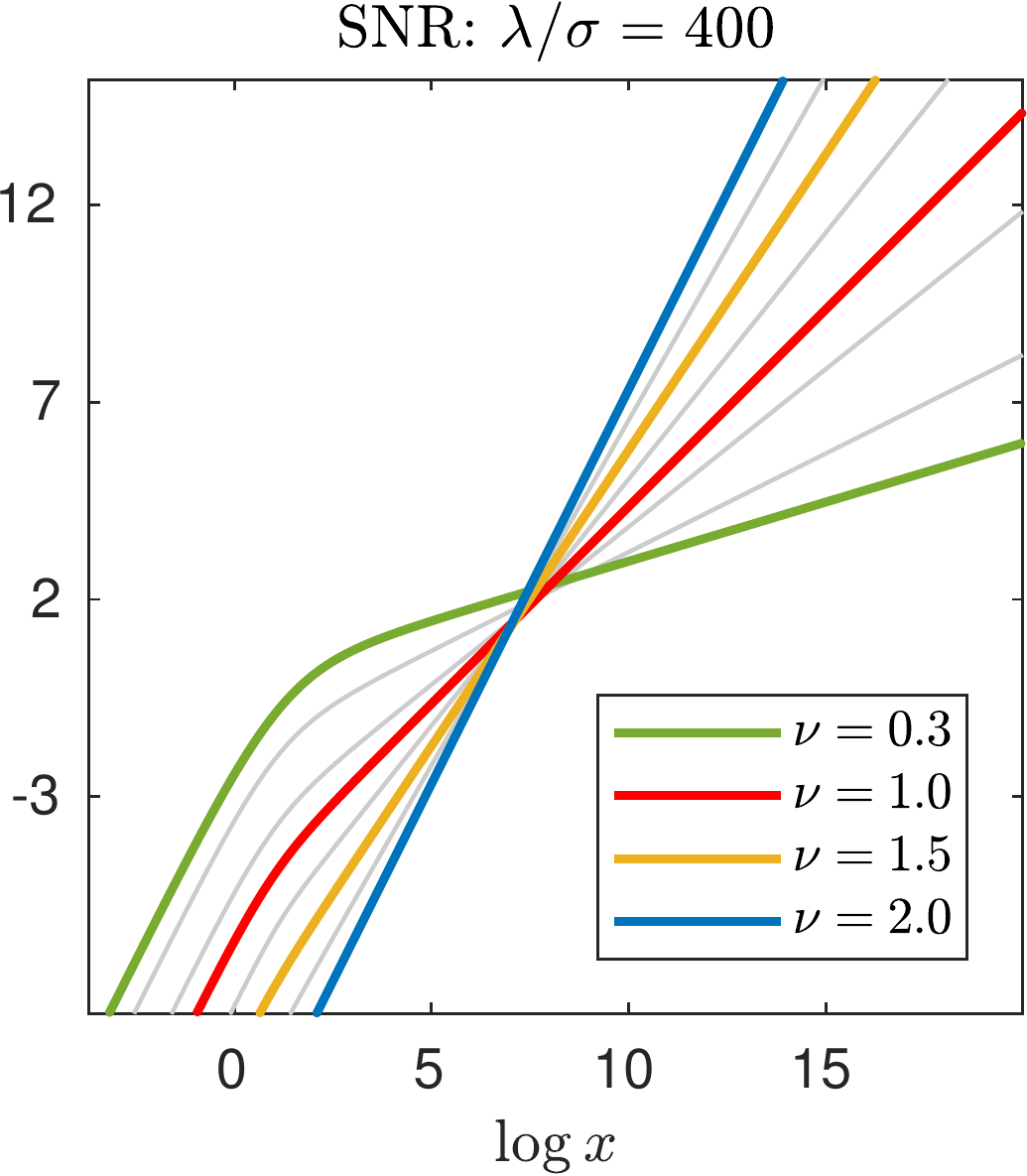}}%
  \caption{Illustrations of the log-discrepancy function
    for various $0.3 \leq \nu \leq 2$ and SNR $\lambda/\sigma$.
  }
  \label{fig:disc_snr}
  \vskip1em
\end{figure}

\Cref{fig:disc_snr} shows the evolution of
the log-discrepancy function for various values of $\nu$
in the context of three different signal-to-noise ratios $\lambda/\sigma$ (SNR).
One can observe that as the SNR decreases (resp., increases), the left (resp., right)
asymptotic behavior starts dominating over the right (resp., left) asymptotes.
In other words, for $\nu < 2$, the intersection of the two asymptotes goes to $+\infty$
(resp., $-\infty$). Last but not least, for $0 < \nu \leq 2$,
the log-discrepancy function $\varphi_\lambda^\nu$ is always concave and since
$\alpha_2 \leq \alpha_1$ it is thus upper-bounded by its left and right asymptotes.

\bigskip

From \Cref{thm:assignment_general_left},
\Cref{thm:assignment_general_right} and \Cref{rem:conjecture},
we can now build two asymptotic
log-linear approximations for $\varphi_\lambda^\nu$, with $0 < \nu \leq 2$,
and subsequently an asymptotic power approximation
for  $f_{\sigma,\lambda}^\nu$ by using 
the relation \eqref{eq:f_vs_varphi}.
Next, we explain the approximation process of the in-between behavior,
as well as its efficient evaluation.

\bigskip

\subsection{Numerical approximation}

We now describe the proposed approximation of the discrepancy function
$f_{1,\lambda}^\nu$ through 
an approximation
$\hat \varphi_\lambda^\nu$ of the log-discrepancy
function as
\begin{align}
  & \hat f_{1,\lambda}^\nu(x)
  =
  \gamma_\lambda^\nu
  + \exp \hat \varphi_\lambda^\nu(x)
  \qwhereq
  \gamma_\lambda^\nu
  = f_{1,\lambda}^\nu(0)~.
\end{align}
Based on our previous theoretical analysis,
a solution preserving the asymptotic, increasing and concave behaviors
of $\varphi_{1,\lambda}^\nu$ can be defined
by making use of the following approximations
\begin{align}
  \hat \varphi_\lambda^\nu(x)
  =
  \alpha_1 \log |x| + \beta_1 -
  \texttt{rec}(\alpha_1 \log |x| + \beta_1 - \alpha_2 \log |x| - \beta_2)~,
\end{align}
where $\texttt{rec}$ is a so-called rectifier function
that is positive, increasing, convex and satisfies 
\begin{align}
  \ulim{x \to -\infty} \texttt{rec}(x) = 0~ \qandq
  \texttt{rec}(x) \usim{x \to \infty} x~.
\end{align}
In this paper, we consider the two following rectifying functions
\begin{align}
  &
  \texttt{relu}(x)
  =
  \max(0, x)
  \qandq
  \texttt{softplus}(x)
  =
  h\log\left[1 + \exp\left(\frac{x}{h}\right)\right],~
  h > 0~,
\end{align}
as coined respectively in \cite{nair2010rectified}
and \cite{dugas2001incorporating}.
Using the function \texttt{relu}
(Rectified linear unit) leads to an approximation
$\hat\varphi_\lambda^\nu$ that is exactly equal to the asymptotes
of $\varphi_\lambda^\nu$ with a singularity at their crossing point.
In this paper, we will instead use the function \texttt{softplus}
as it allows the approximation of $\varphi_\lambda^\nu$ to converge smoothly to the asymptotes
without singularity. Its behavior is controlled by the parameter $h > 0$.
The smaller the value of $h$ is, the faster the convergence speed to the asymptotes.

\begin{figure}
  \centering
  \includegraphics[height=.5\linewidth]{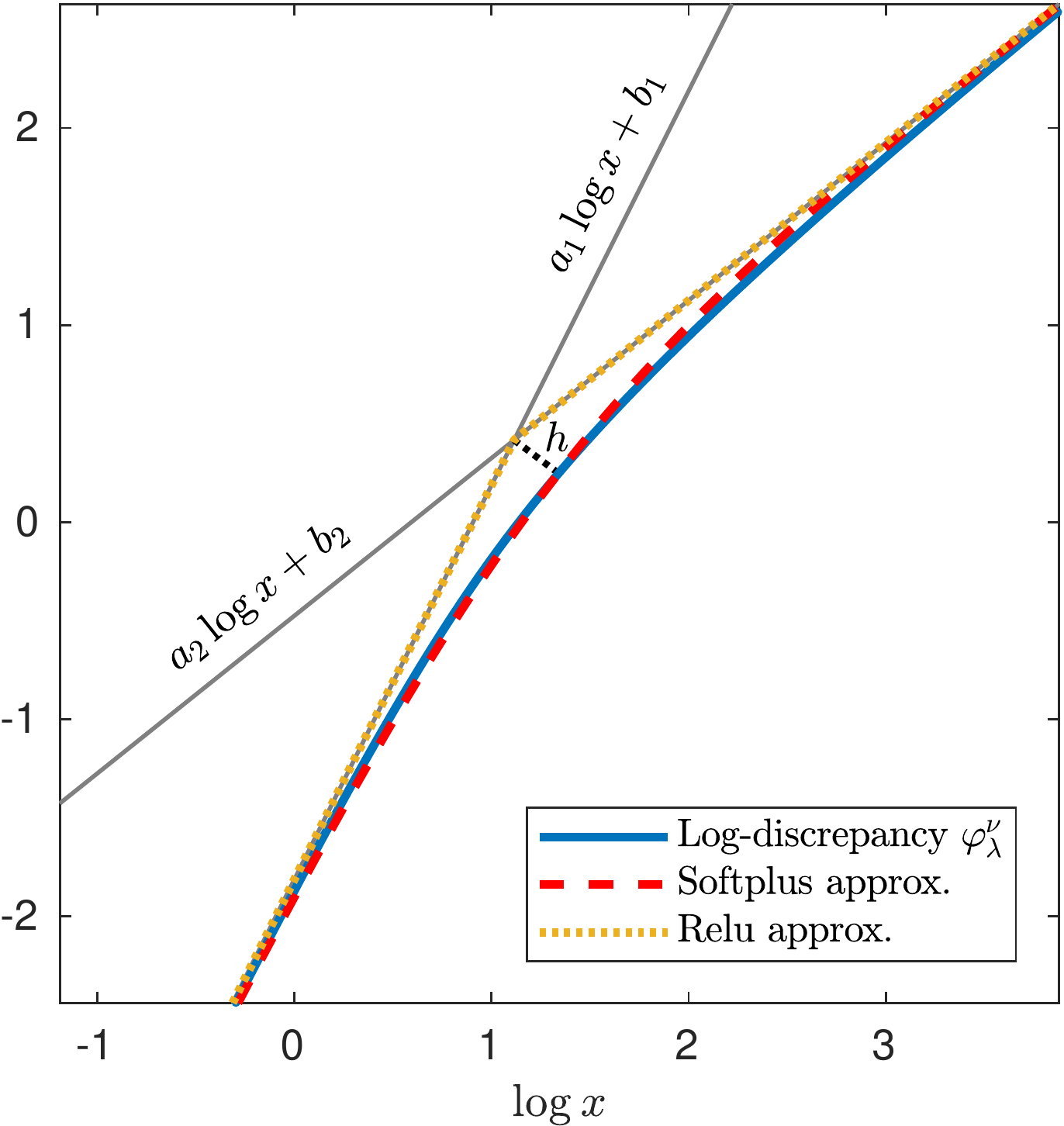}\hfill
  \includegraphics[height=.5\linewidth]{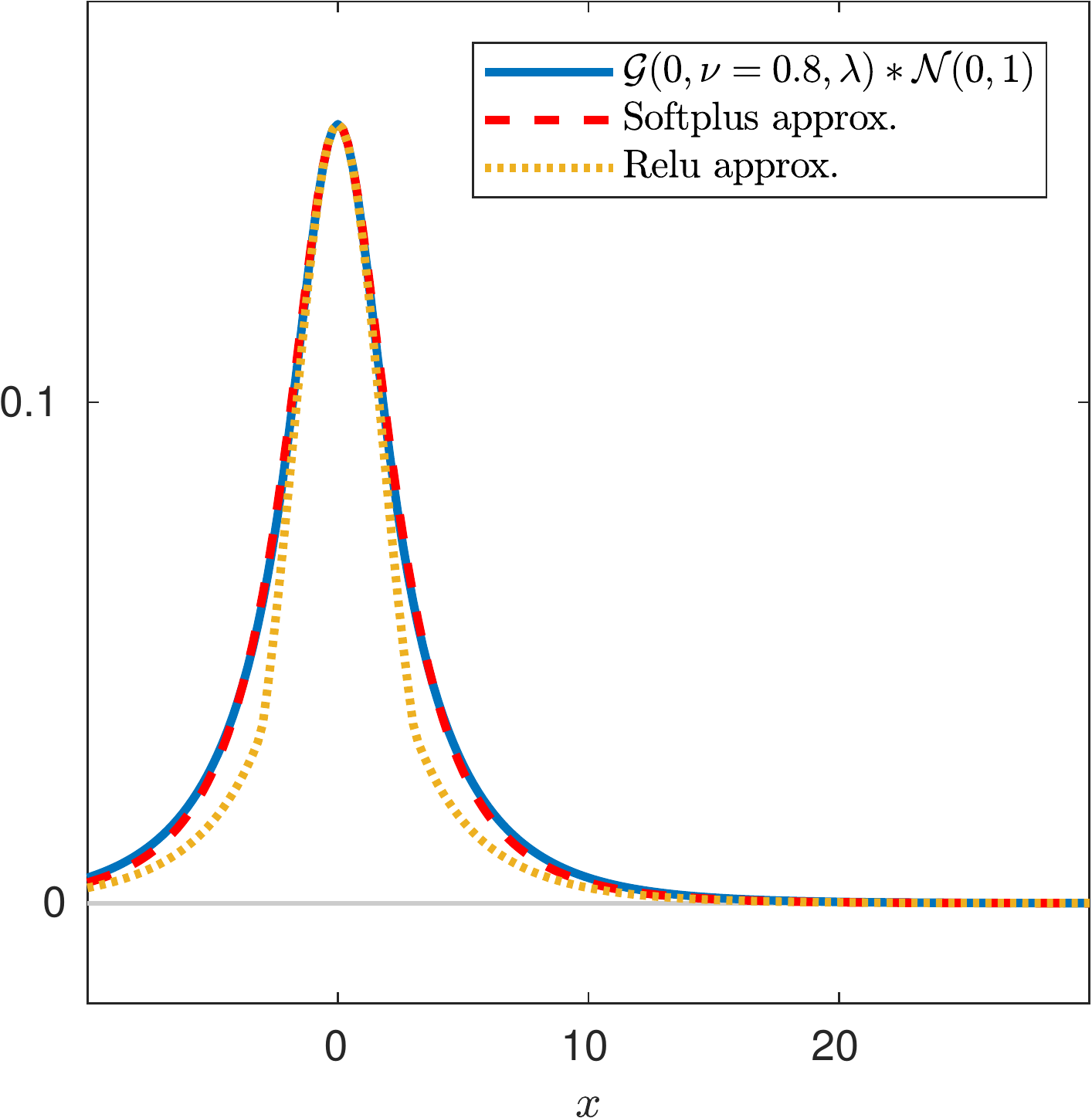}%
  \caption{Illustrations of our approximations of $\varphi_\lambda^\nu$
    and the corresponding underlying posterior distribution
    $\Nn(0, \nu, \lambda) * \Gg(0, 1)$
    (where $\nu=.8$ and $\lambda=4$).
    The blue curves have been obtained by evaluating the convolution
    using numerical integration techniques for all $x$. The dashed
    curves are obtained using the proposed
    \texttt{relu}- and \texttt{softplus}-based approximations
    that have closed-form expressions.
  }
  \label{fig:relu_vs_sp}
  \vskip1em
\end{figure}

\begin{figure}
  {\includegraphics[height=.399\linewidth,viewport=0 42 598 507,clip]{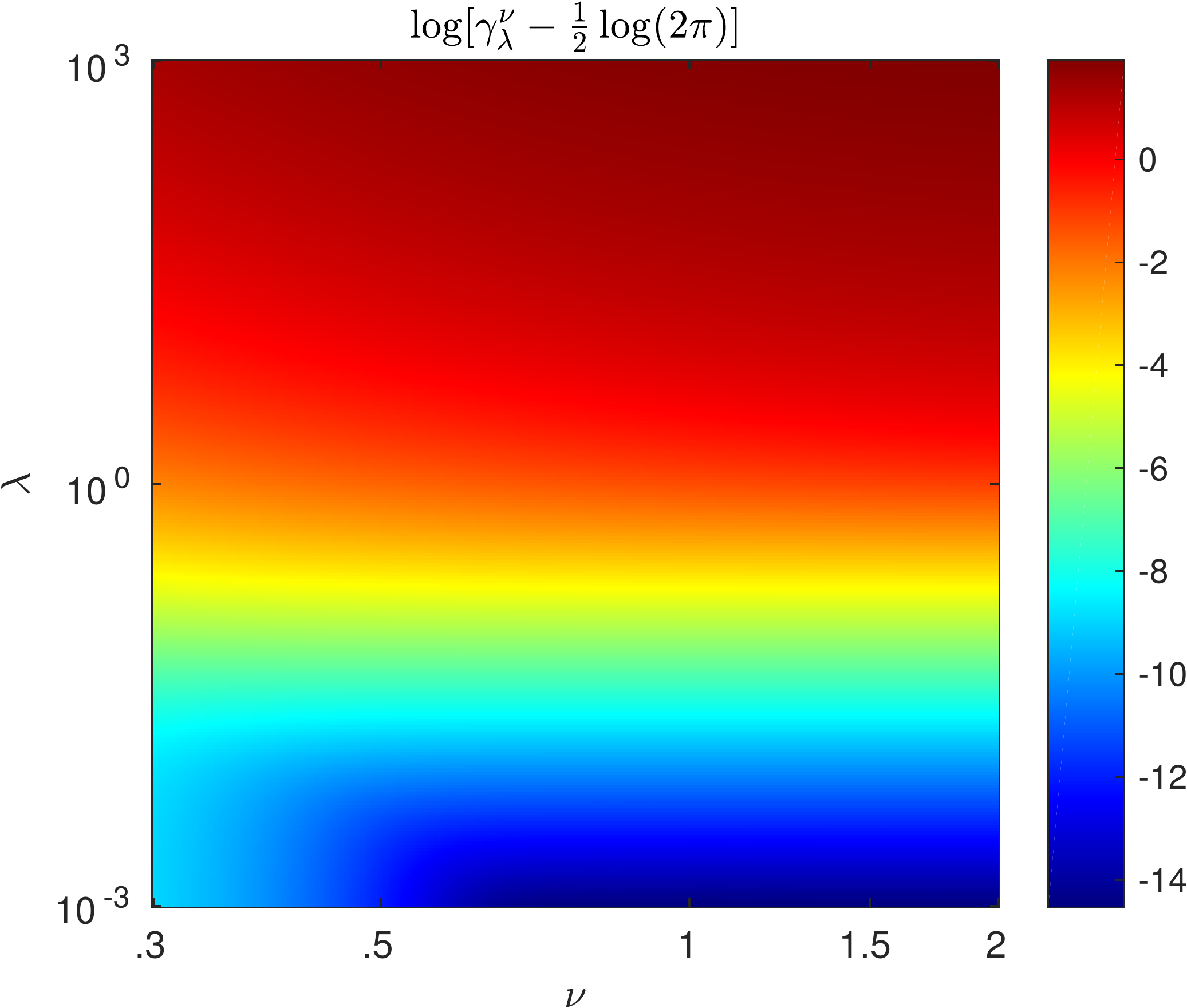}}\hfill%
  {\includegraphics[height=.399\linewidth,viewport=68 42 598 507,clip]{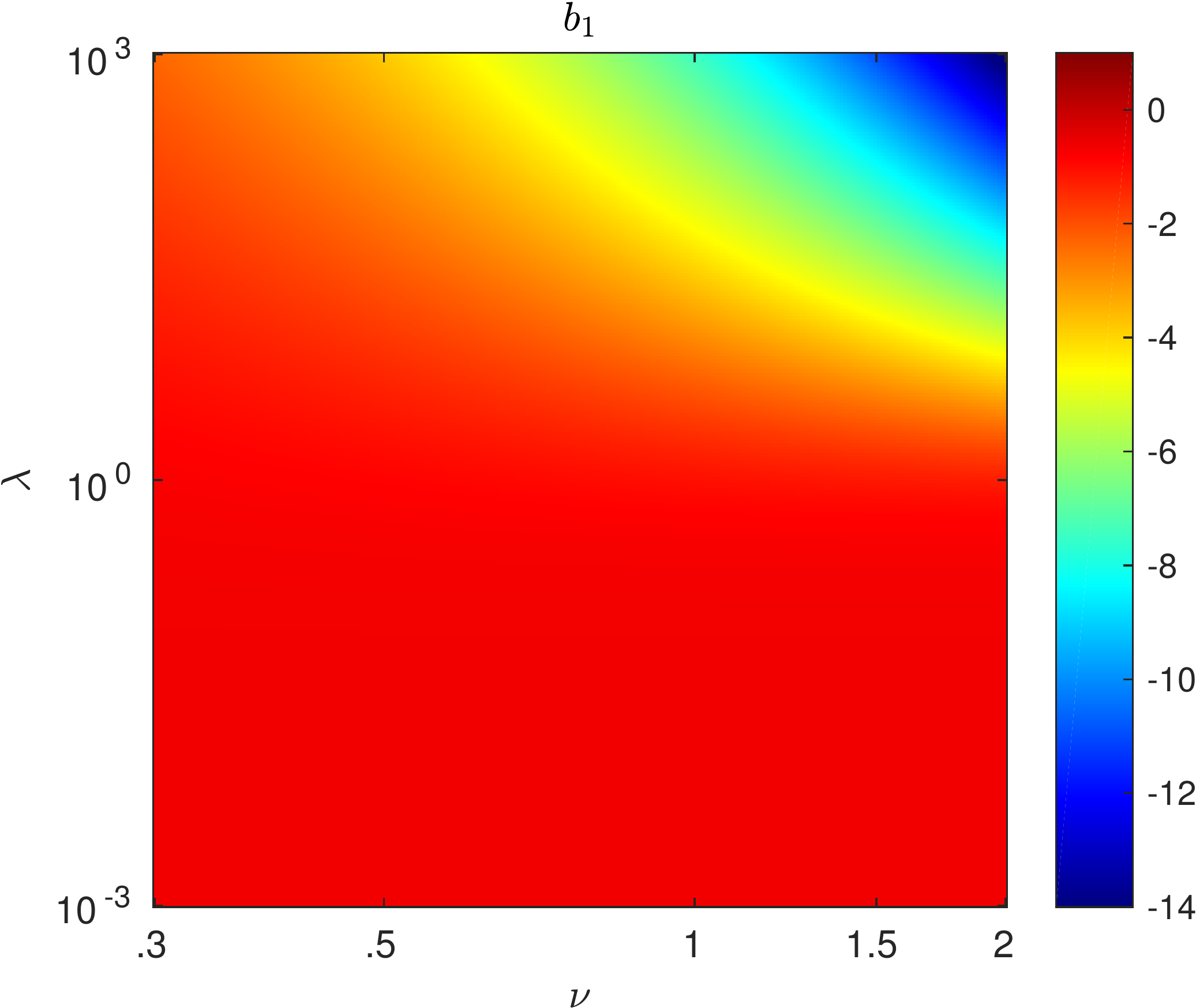}}\\[1em]
  {\includegraphics[height=.43\linewidth]{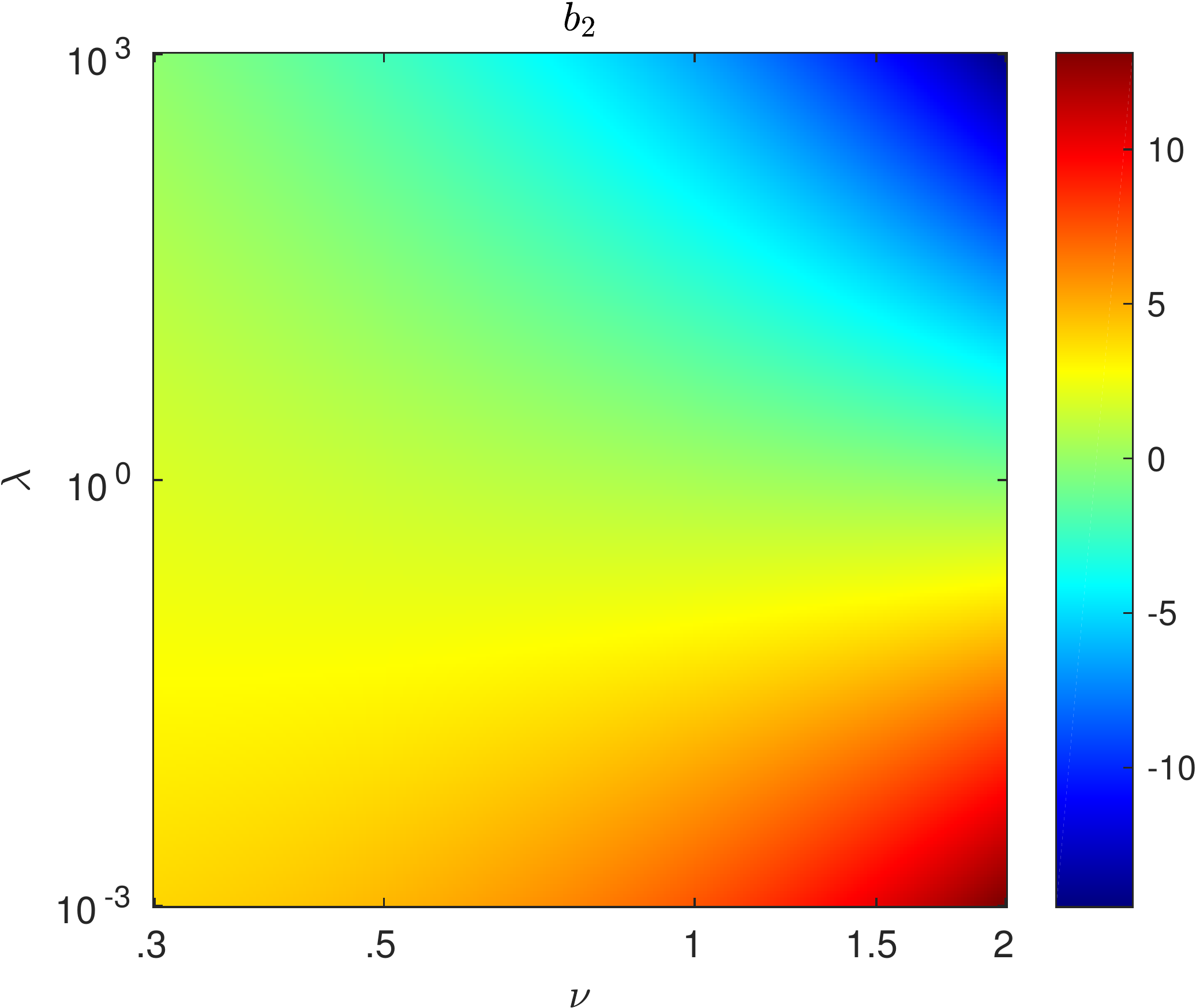}}\hfill%
  {\includegraphics[height=.43\linewidth,viewport=68 0 601 507,clip]{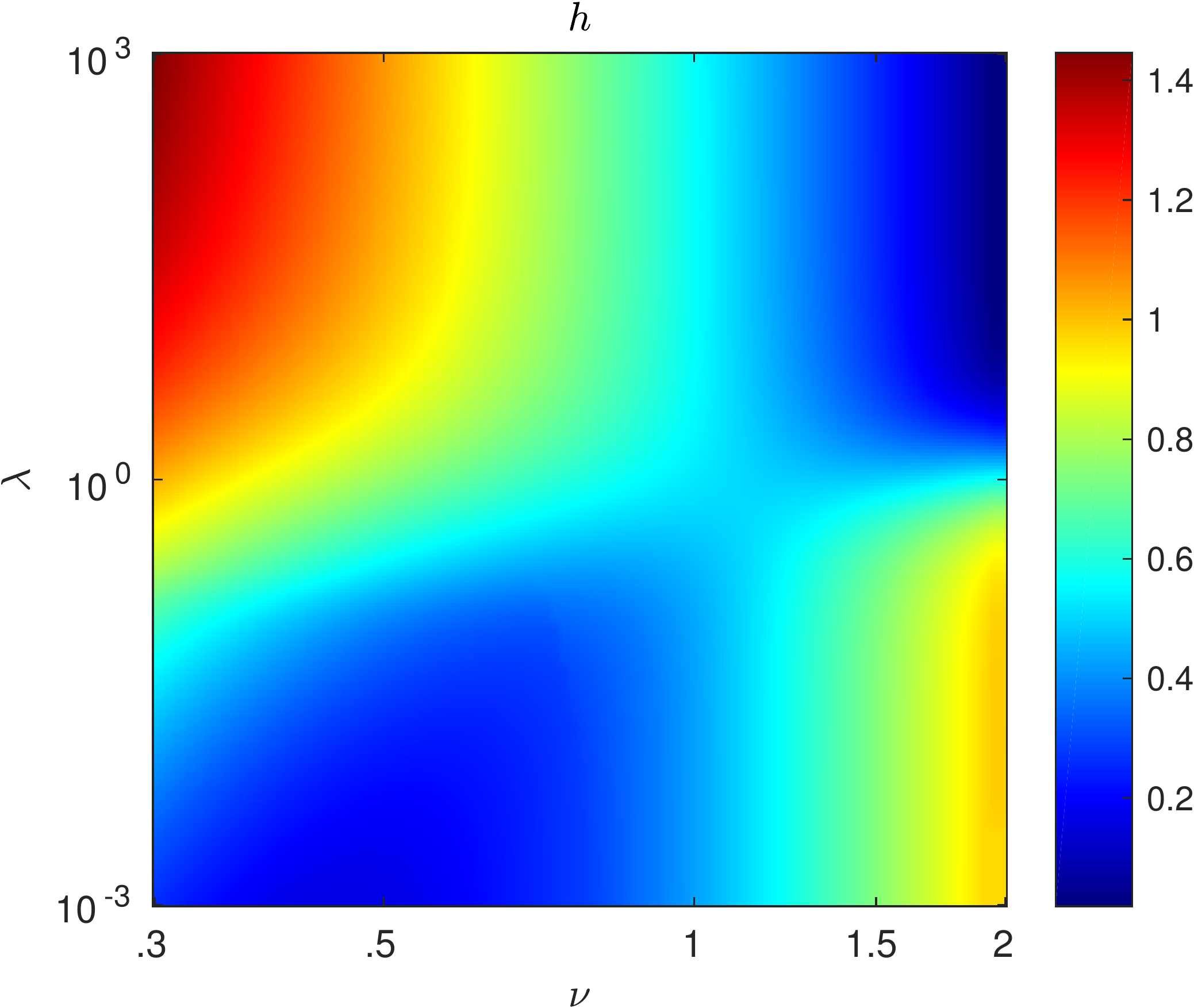}}%
  \caption{Lookup tables used to store the values of the parameters
    $\gamma_\lambda^\nu$, $\beta_1$, $\beta_2$ and $h$
    for various $.3 \leq \nu \leq 2$ and $10^{-3} \leq \lambda \leq 10^3$.
    A regular grid of $100$ values has been used for $\nu$ and
    a logarithmic grid of $100$ values has been used for $\lambda$.
    This leads to a total of $10, 000$ combinations for each of the
    four lookup tables.
  }
  \label{fig:lookuptables}
  \vskip1em
\end{figure}

The parameter $h$ should be chosen such that the approximation error
between $\hat \varphi_\lambda^\nu(x)$ and $\varphi_\lambda^\nu(x)$ is as small as possible.
This can be done numerically by first evaluating
$\varphi_\lambda^\nu(x)$ with integration techniques
for a large range of values $x$,
and then selecting the parameter $h$ by least square.
Of course, the optimal value for $h$ depends on the parameter $\lambda$ and $\nu$.

\Cref{fig:relu_vs_sp} gives an illustration of our approximations
of the log-discrepancy and the corresponding distribution
obtained with \texttt{relu} and \texttt{softplus}.
On this figure the underlying functions have been obtained by numerical
integration for a large range of value of $x$.
One can observe that using \texttt{softplus} provides a better approximation
than \texttt{relu}.
\\

Our approximation for $\hat f_{1,\lambda}^\nu(x)$
is parameterized by six scalar values: $\gamma_\nu^\lambda$, $\alpha_1$, $\beta_1$,
$\alpha_2$, $\beta_2$ and $h$ that depend only
on the original parameters $\lambda$ and $\nu$.
From our previous analysis, we have that $\alpha_1=2$ and $\alpha_2=\nu$.
The other parameters are non-linear functions of $\lambda$ and $\nu$.
The parameters $\gamma_\nu^\lambda$, $\beta_1$
and $\beta_2$ require either performing numerical integration or
evaluating the special function $\Gamma$. As discussed, the parameter $h$
requires numerical integration for various $x$ and then optimization.
For these reasons, these values cannot be computed during runtime.
Instead, we pre-compute these four parameters offline
for $10,000$ different combinations of $\lambda$ and $\nu$
values in the intervals $[10^{-3},10^3]$ and $[0.3, 2]$, respectively
(the choice for this range was motivated in \Cref{sec:learning_ggmm}).
The resulting values are then stored in four corresponding lookup tables.
During runtime, these parameters are retrieved online
by bi-linear extrapolation and interpolation.
The four lookup tables are given in \Cref{fig:lookuptables}.
We will see in \Cref{sec:exp_eval} that using the approximation
$\hat f_{1,\lambda}^\nu$ results in substantial
acceleration
without significant loss of performance as compared to
computing $f_{1,\lambda}^\nu$ directly by numerical integration during runtime.

\begin{figure}
  {\includegraphics[height=.37\linewidth]{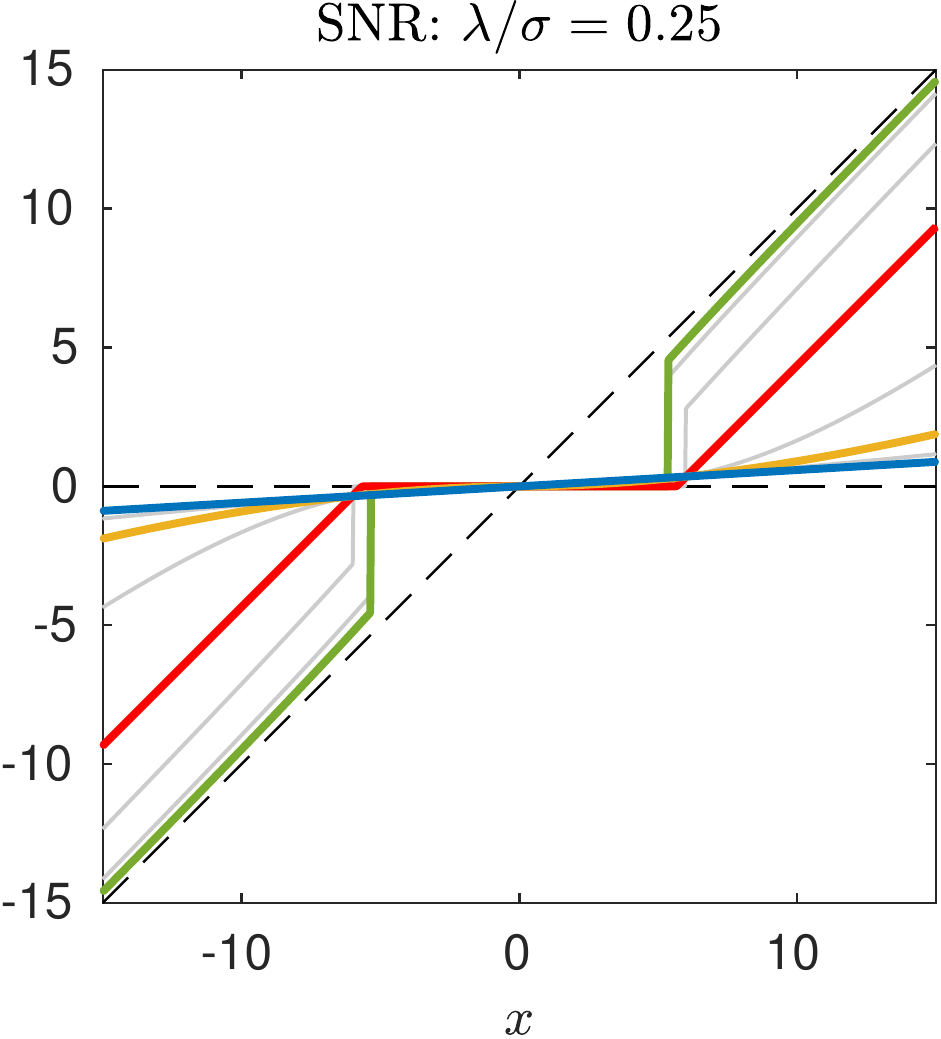}}\hfill%
  {\includegraphics[height=.37\linewidth,viewport=25 0 271 298,clip]{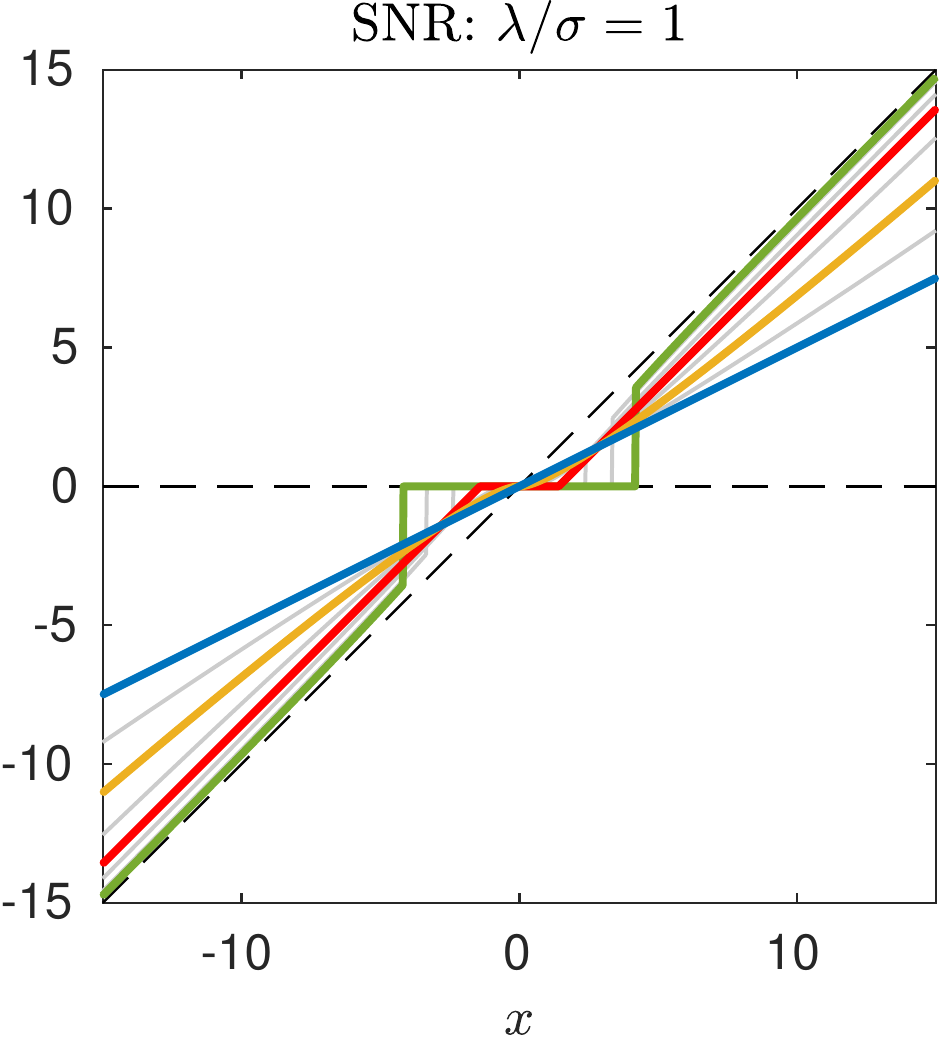}}\hfill%
  {\includegraphics[height=.37\linewidth,viewport=25 0 271 297,clip]{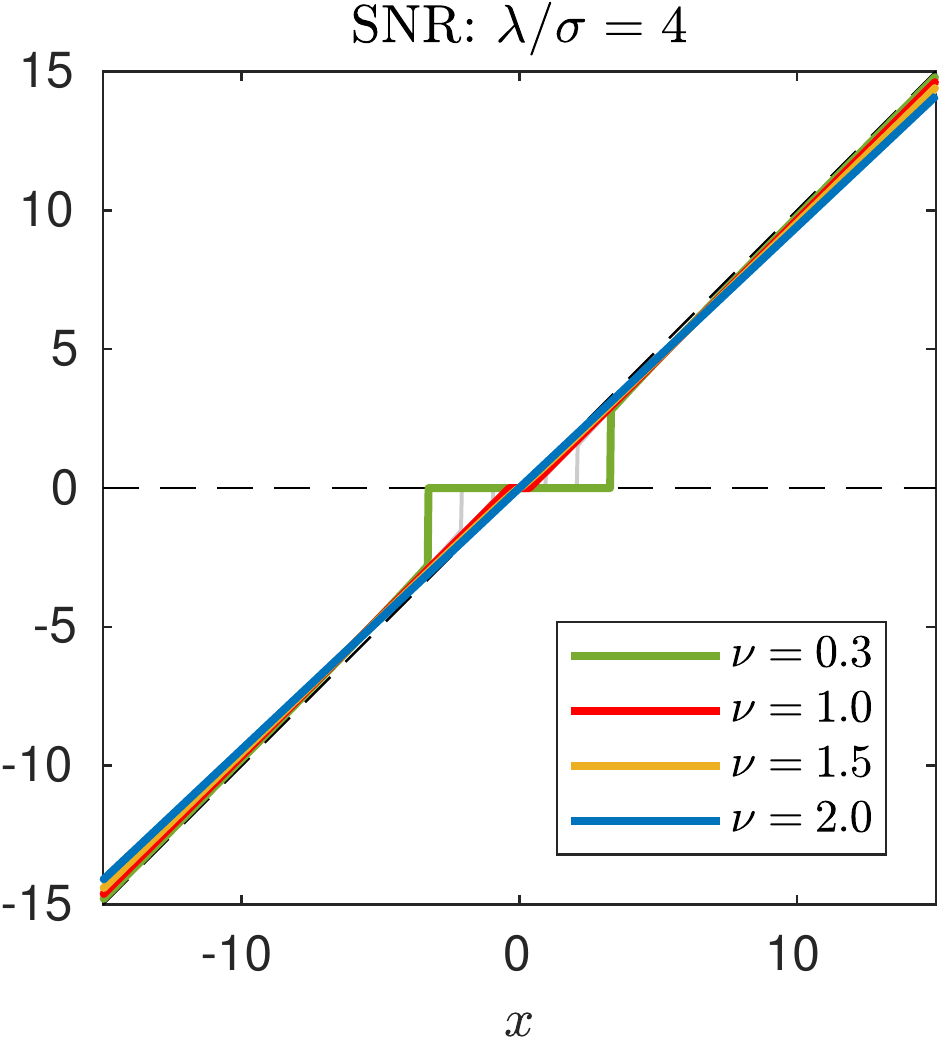}}%
  \caption{Illustrations of the shrinkage function
    for various $0.3 < \nu \leq 2$ and SNR $\lambda/\sigma$.
  }
  \label{fig:shrinkage}
  \vskip1em
\end{figure}

\section{Shrinkage functions: analysis and approximations}
\label{sec:shrinkage}

Recall that from its definition given in eq.~\eqref{eq:shrinkage_func},
the shrinkage function is defined for $\nu > 0$, $\sigma >0$
and $\lambda >0$, as
\begin{align}\label{eq:shrinkage_func_bis}
  s_{\sigma,\lambda}^\nu(x)
  &\in
  \uargmin{t\in \RR}
  \frac{(x - t)^2}{2 \sigma^2}
  +
  \lambda_\nu^{-\nu}|t|^\nu~.
\end{align}

\subsection{Theoretical analysis}

Except for some particular values of $\nu$ (see, \Cref{sec:shrinn_num}),
Problem \eqref{eq:shrinkage_func_bis} does not have explicit solutions.
Nevertheless, as shown in \cite{moulin1999analysis}, Problem \eqref{eq:shrinkage_func_bis}
admits two (not necessarily distinct) solutions. One of them is implicitly characterized as
\begin{align}\label{eq:map}
  & s_{\sigma,\lambda}^\nu(x) =
  \choice{
    0 & \ifq 0 < \nu \leq 1 \qandq |x| \leq \tau_\lambda^\nu~,\\
    t^\star & \text{otherwise}~,
  }\\
  \qwhereq
  &
  t^\star = x - \sign(t^\star) \nu \sigma^2 \lambda_\nu^{-\nu} |t^\star|^{\nu - 1}~,
  \nonumber
  \\
  \qandq &
  \tau_\lambda^\nu =
  \choice{
    (2 - \nu) (2 - 2 \nu)^{-\frac{1-\nu}{2-\nu}}
    (\sigma^2 \lambda_\nu^{-\nu})^{\frac{1}{2-\nu}}
    & \ifq \nu < 1~,\\
    \sigma^2 \lambda^{-1}
    & \otherwise \;\text{($\nu = 1$)}~.
  }
  \nonumber
\end{align}
The other one is obtained by changing $|x| \leq \tau_\lambda^\nu$ to $|x| < \tau_\lambda^\nu$
in \eqref{eq:map}, and so they coincide for almost every $(x, \lambda, \sigma, \nu)$.
As discussed in \cite{moulin1999analysis}, for $\nu > 1$,
$s_{\sigma,\lambda}^\nu(x)$ is differentiable, and for $\nu \leq 1$,
the shrinkage exhibits a threshold $\tau_\lambda^\nu$ that produces
sparse solutions.
\Cref{prop:shrink_basic_prop} summarizes a few
important properties.

\begin{prop}\label{prop:shrink_basic_prop}
  Let $\nu >0$, $\sigma>0$, $\lambda >0$
  and $s_{\sigma,\lambda}^\nu$ as defined in eq.~\eqref{eq:shrinkage_func}.
  The following relations hold true
  \begin{align}
    \label{eq:shrink_reduction}
  \tag{reduction}
  s_{\sigma,\lambda}^\nu(x)
  &=
  \sigma
  s_{1,\frac{\lambda}{\sigma}}^\nu\left(\frac{x}{\sigma}\right)
  ~,
  \\[.8em]
    \label{eq:shrink_anti_symmetry}
  \tag{odd}
  s_{\sigma,\lambda}^\nu(x)
  &=
  -s_{\sigma,\lambda}^\nu(-x)
  ~,
  \\[.8em]
  \label{eq:shrinkage}
  \tag{shrinkage}
  s_{\sigma,\lambda}^\nu(x)
  &\in
  \choice{
    {[}0, x{]} & \ifq x \geq 0\\
    {[}x, 0{]} & \otherwise\\
  }
  ~,
  \\[.8em]
  \label{eq:shrink_increasing}
  \tag{increasing with $x$}
  x_1 \geq x_2 & \Leftrightarrow
  s_{\sigma,\lambda}^\nu(x_1)
  \geq
  s_{\sigma,\lambda}^\nu(x_2)
  ~,
  \\[.8em]
  \label{eq:shrink_increasing2}
  \tag{increasing with $\lambda$}
  \lambda_1 \geq \lambda_2 & \Leftrightarrow
  s_{\sigma,\lambda_1}^\nu(x)
  \geq
  s_{\sigma,\lambda_2}^\nu(x)
  ~,
  \\[.8em]
  \tag{kill low SNR}
  \label{eq:shrink_kill}
  \ulim{\frac\lambda\sigma \to 0}
  &s_{\sigma,\lambda}^\nu(x)
  =
  0
  ~,
  \\
  \tag{keep high SNR}
  \label{eq:shrink_keep}
  \ulim{\frac\lambda\sigma \to +\infty}
  &s_{\sigma,\lambda}^\nu(x)
  =
  x
  ~.
\end{align}
\end{prop}
The proofs can be found in \Cref{proof:shrink_basic_prop}.
These properties show that $s_{\sigma,\lambda}^\nu$ is indeed
a shrinkage function (non-expansive). It shrinks the input coefficient $x$
according to the model $\nu$ and the
modeled signal to noise ratio $\frac\lambda\sigma$ (SNR).
When $x$ is small in comparison to the SNR, it is likely that
its noise component dominates the underlying signal, and
is, therefore, shrunk towards $0$. Similarly, when $x$ is large,
it will likely be preserved. This is even more likely when $\nu$ is small, since
in this case large coefficients are favored by the {\it prior}.
Illustrations of shrinkage functions for various SNR and $\nu$
are given in \Cref{fig:shrinkage}.

\begin{table}[t]
  \centering
  \caption{Shrinkage function under generalized Gaussian priors
  }
  \label{tab:shrinkage}
  \begin{tabular*}{\textwidth}{c@{\quad\quad}l@{\extracolsep{\fill}}r}
    \hline
    $\nu$ & Shrinkage $s_{\sigma,\lambda}^\nu(x)$ & Remark\\
    \hline
    \\[-.5em]
    $< 1$
    &
    $\displaystyle \choice{\displaystyle
      x - \gamma x^{\nu - 1} + O(x^{2(\nu - 1)})
      &
      \displaystyle
      \ifq |x| \geq \tau_\lambda^\nu\\
      0 & \otherwise\\
      }$
    &
    $\begin{array}{r}
      \approx \text{Hard-thresholding}\\
              \text{\cite{moulin1999analysis}}\\
    \end{array}$\hspace*{-.5em}\\[2em]
    $1$ &
    $\displaystyle \sign(x) \max\left( |x| - \frac{\sqrt{2} \sigma^2}{\lambda}, 0\right)$
    &
    $\begin{array}{r}
      \text{Soft-thresholding}\\
      \text{\cite{donoho1994ideal}}\\
    \end{array}$\hspace*{-.5em}\\[2em]
    $4/3$ &
    $\displaystyle x + \gamma \left(\sqrt[3]{\frac{\zeta - x}{2}} - \sqrt[3]{\frac{\zeta + x}2} \right)
    $
    &
    \cite{chaux2007variational}
    \\[2em]
    $3/2$ &
    $\displaystyle \sign(x)
    \frac{\left(\sqrt{\gamma^2 + 4 |x|} - \gamma\right)^2}{4}$
    &
    \cite{chaux2007variational}
    \\[2em]
    $2$ &
    $\displaystyle \frac{\lambda^2}{\lambda^2 + \sigma^2} \cdot
    x$
    &
    Wiener (LMMSE)
    \\[1em]
    \hline
  \end{tabular*}
  \begin{align*}
    \text{with} \quad
    \gamma = \nu \sigma^2 \lambda_\nu^{-\nu}
    \qandq
    \zeta = \sqrt{x^2 + 4 \left(\frac{\gamma}{3}\right)^3}~.
  \end{align*}
  \hrule
  \vskip1em
\end{table}

\subsection{Numerical approximations}
\label{sec:shrinn_num}

The shrinkage function $s_{\sigma,\lambda}^\nu$,
implicitly defined in \eqref{eq:map} does not have
a closed form expression in general. Nevertheless,
for fixed values of $x$, $\sigma$, $\lambda$, $\nu$,
$s_{\sigma,\lambda}^\nu(x)$ can be estimated using iterative solvers
such as Newton descent or Halley’s root-finding method.
These approaches converge quite fast and, in practice, reach a satisfying
solution within ten iterations.
However, since in our application of interest we need to evaluate this function
a large number times, we will follow a different path in order to reduce computation time 
(even though we have implemented this strategy).

As discussed earlier, $s_{\sigma,\lambda}^\nu$ is known in closed form
for some values of $\nu$, more precisely: $\nu = \{1, 4/3, 3/2, 2\}$ (as well as $\nu = 3$
but this is out of the scope of this study),
see for instance \cite{chaux2007variational}.
When $\nu = 2$, we retrieve the linear minimum mean square estimator
(known in signal processing as Wiener filtering) and related to
Tikhonov regularization and ridge regression.
This shrinkage is linear and the slope of the shrinkage goes from $0$
to $1$ as the SNR increases (see \Cref{fig:shrinkage}).
When $\nu = 1$, the shrinkage is the well-known soft-thresholding
\cite{donoho1994ideal},
corresponding to the maximum a posteriori estimator under a Laplacian prior.
When $\nu < 1$, the authors of \cite{moulin1999analysis} have shown that
(i) the shrinkage admits a threshold with a closed-form expression (given
in eq.~\eqref{eq:map}), and (ii) the shrinkage is approximately equal to
hard-thresholding with an error term that vanishes when $|x| \to \infty$.
All these expressions are summarized in \Cref{tab:shrinkage}.

In order to keep our algorithm as fast as possible, we propose to
use the approximation of the shrinkage given for $\nu < 1$ in \cite{moulin1999analysis}.
Otherwise, we pick one of the four shrinkage functions
corresponding to $\nu = \{1, 4/3, 3/2, 2\}$ by nearest neighbor
on the actual value of $\nu \in [1, 2]$. Though this approximation may seem
coarse compared to the one based on iterative solvers,
we did not observe any significant loss of quality
in our numerical experiments (see \Cref{sec:exp_eval}).
Nonetheless, this alternative leads
to 6 times speed-up while evaluating shrinkage.

\section{Experimental evaluation}
\label{sec:exp_eval}

In this section we explain the methodology used to evaluate the GGMM model,
and present numerical experiments to compare the performance of the proposed GGMM
model over existing GMM-based image denoising algorithms.
To demonstrate the advantage of allowing for
a flexible GGMM model, we also present results using GGMM models
with fixed shape parameters,
$\nu = 1$ (Laplacian mixture model) and $\nu = 0.5$ (Hyper-Laplacian mixture model).
For learning Laplacian mixture model (LMM) and
hyper-Laplacian mixture model (HLMM), we use the same procedure
as described in \Cref{sec:learning_ggmm}
but force all shape parameters to be equal to $1$ or $0.5$, respectively.

\begin{figure}
  \centering
  \subfigure[GMM ($\nu=2$)]{\includegraphics[width=.245\linewidth]{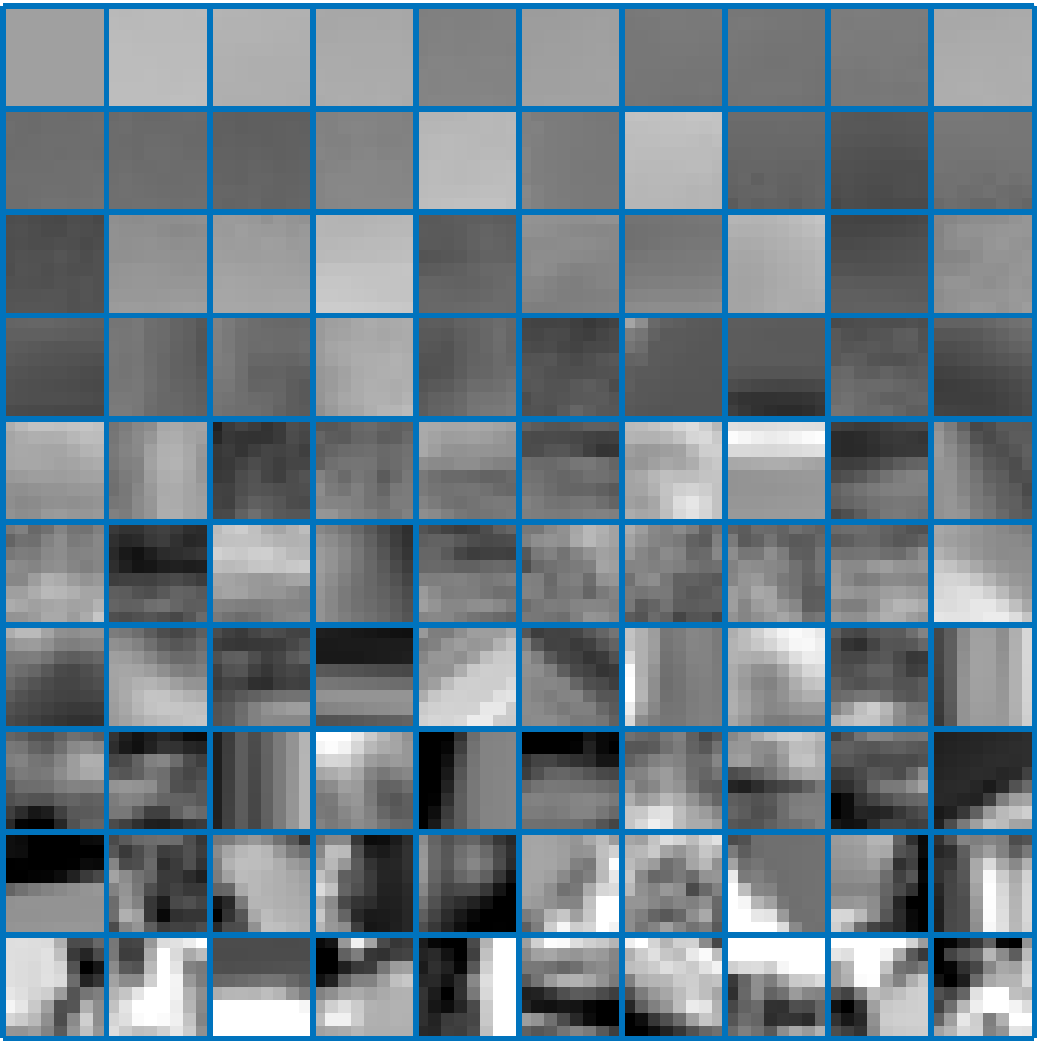}}\hfill
  \subfigure[GGMM ($.3 \leq \nu \leq 2$)]{\includegraphics[width=.245\linewidth]{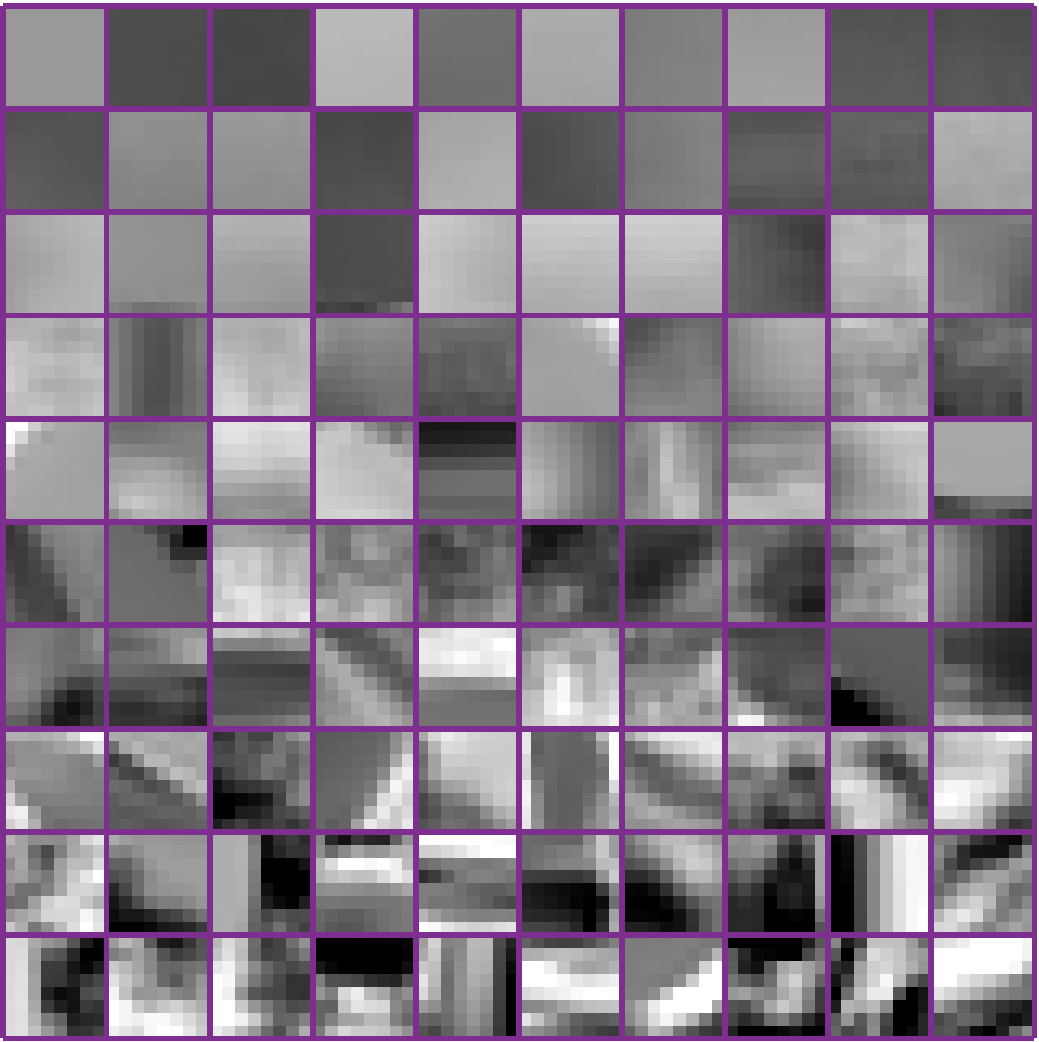}}\hfill
  \subfigure[LMM ($\nu=1$)]{\includegraphics[width=.245\linewidth]{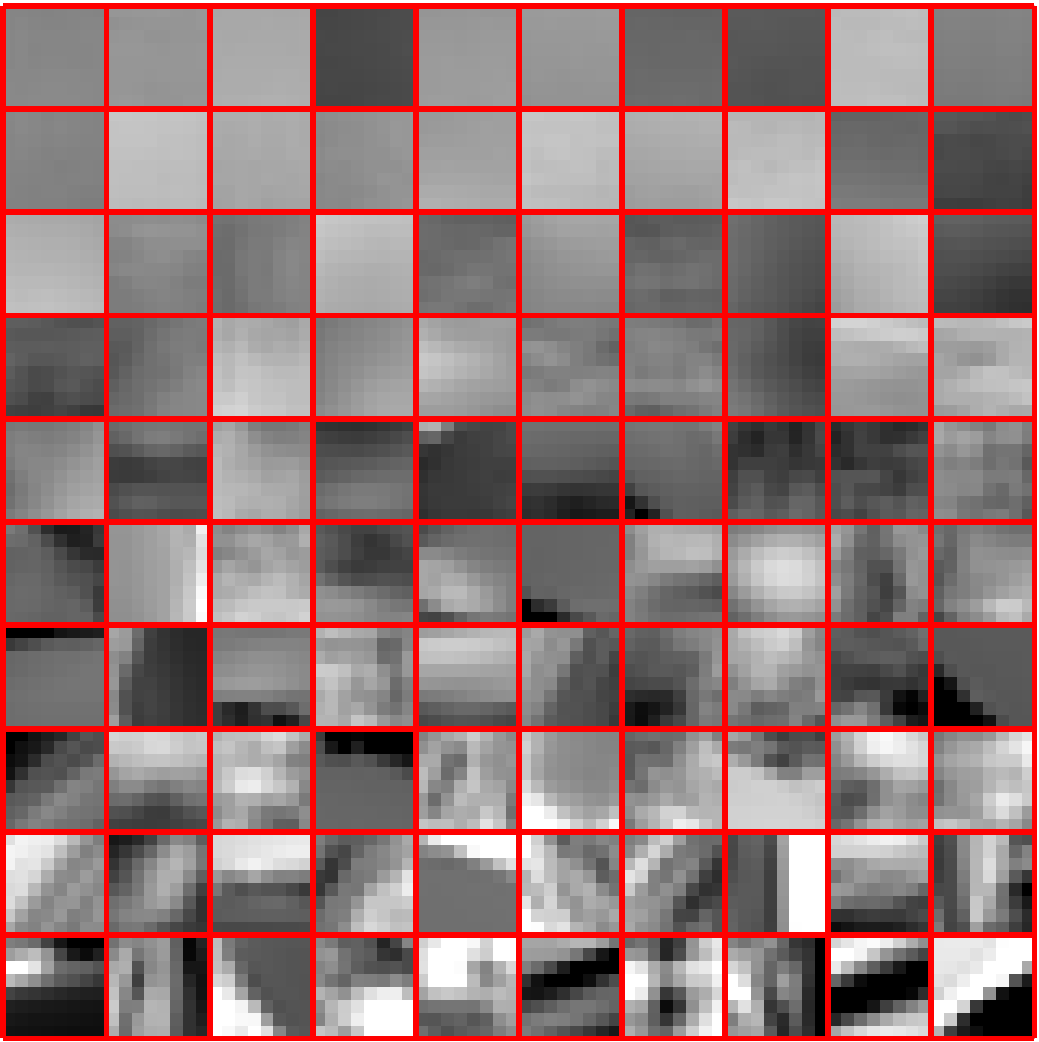}}\hfill
  \subfigure[HLMM ($\nu=.5$)]{\includegraphics[width=.245\linewidth]{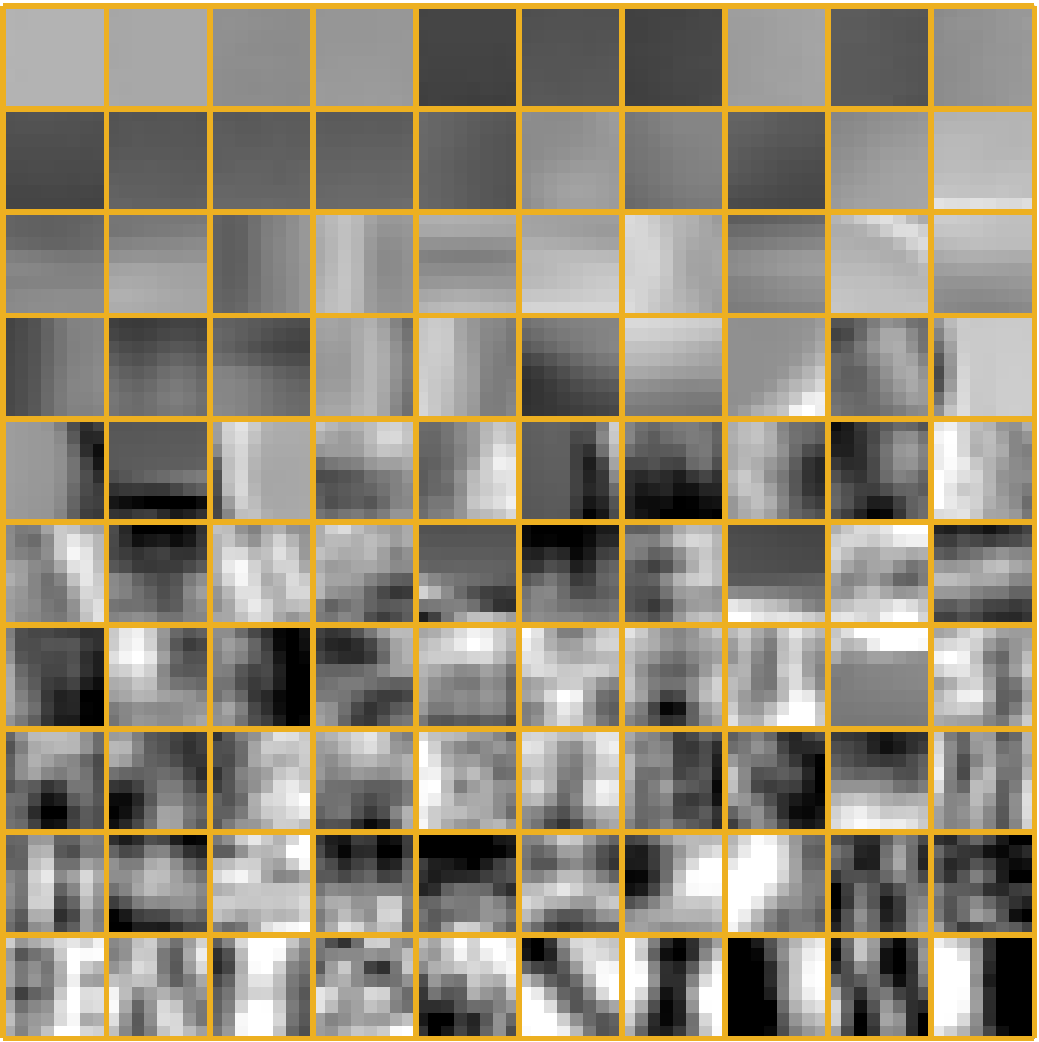}}\\
  \caption{Set of 100 patches, sorted by the norm of their gradient,
    and generated to be independently distributed according to
    (from left to right) a GMM, GGMM, LMM and HLMM.
    For ease of visualization, only the top eigendirections corresponding to $80\%$ of the
    variance have been chosen. Near-constant patches with variance smaller
    than $\tfrac2P \norm{\bSigma_k}_F^2$ have also been discarded.
  }
  \label{fig:random_patches}
  \vskip1em
\end{figure}

\paragraph{Model validation}
As discussed in \Cref{sec:ggmm},
\Cref{fig:ggmm_vs_gmm} illustrates the validity of our model choices with
histograms of different dimensions of a single patch cluster. It clearly
shows the importance of allowing the shape and scale parameter
to vary across dimensions for capturing underlying patch distributions.
Since GGMM (and obviously, GMM) falls into the class of \textit{generative} models,
one can also assess the \textit{expressivity} of a model by
analyzing the variability of generated patches 
and its
ability to generate relevant image features (edges, texture elements etc.).
This can be tested by selecting a component $k$ of
the GGMM (or GMM) with probability $w_k$ and sampling patches from
the GGD (or GD) as described in \cite{nardon2009simulation}.
\Cref{fig:random_patches} presents a collage of 100 patches
independently generated by this procedure using GMM, GGMM, LMM and HLMM.
As observed, patches generated from GGMM show greater balance
between strong/faint edges, constant patches and subtle textures than
the models that use constant shape parameters such as GGM, LMM and HLMM.


The superiority of our GGMM model over GMM, LMM or HLMM models can also be
illustrated by comparing the log-likelihood (LL) achieved by these
models over a set of clean patches from natural images.
Note that, to maintain objectivity, the models have to be tested
on data that is different than the dataset used during training.
To this end, we compute the LL of the four above-mentioned models on
all non-overlapping patches of $40$ randomly selected images
extracted from BSDS testing set \cite{martin2001database},
which is a different set than the training images used in the
EM algorithm (parameter estimation/model learning).
One can observe that not only GGMM is a better fit than GMM, LMM and HLMM
on average for the $40$ images,
but it is also a better fit on each single image.
Given that GGMM have a larger degree of freedom than GMM,
this study proves that our learning procedure did not
fall prey to over-fitting, and that the
extra flexibility provided by GGMM was used to capture
relevant and accurate image patterns.

\begin{figure}
  \centering
  \includegraphics[width=1\linewidth]{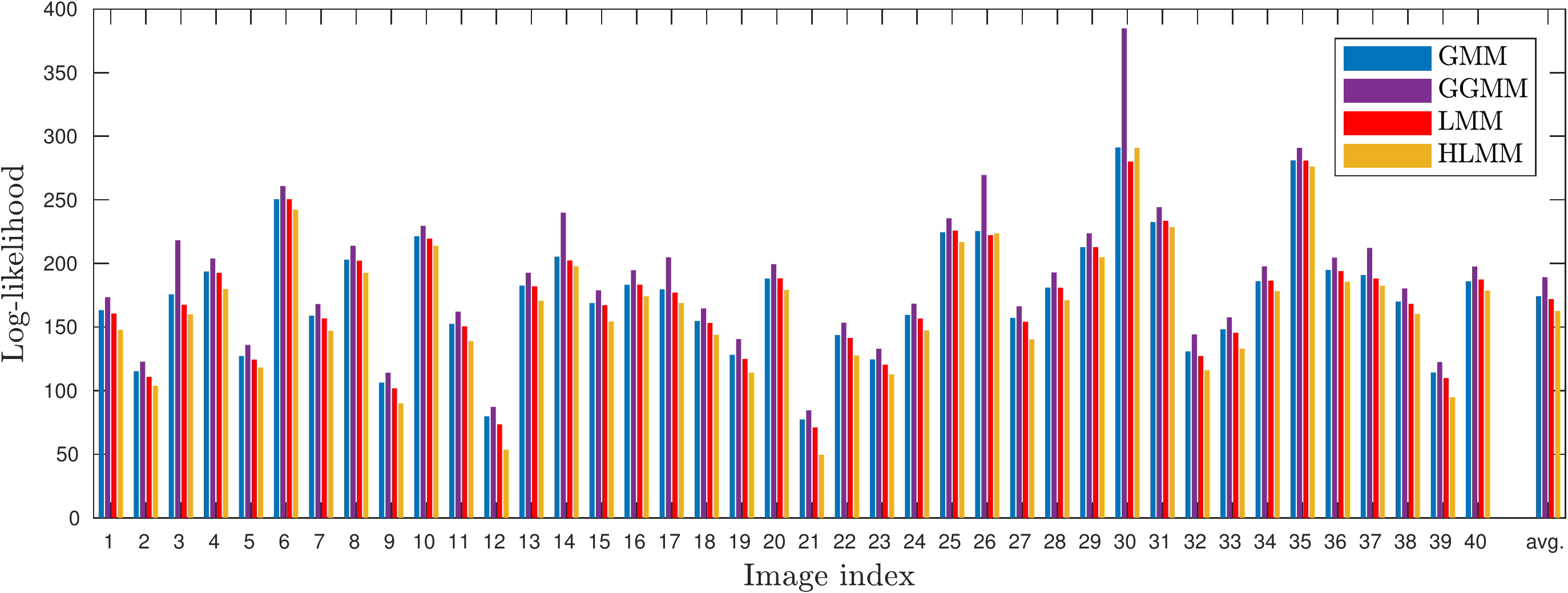}\\
  \caption{Average log-likelihood of all non-overlapping patches
    (with subtracted mean) of each of the $40$ images of our validation subset
    of the testing BSDS dataset for the GMM, GGMM, LMM and HLMM.
    The total average over the $40$ images is shown in the last column.
  }
  \label{fig:loglikelihood}
  \vskip1em
\end{figure}

\afterpage{
\input{table_3algos_main}

\clearpage
\input{table_nu_main}

\input{visual_results}
}

\paragraph{Denoising evaluation}
Following the verification of the model, 
we provide a thorough evaluation of our GGMM prior in denoising task by
comparing its performance against EPLL that uses a GMM prior \cite{Zoran11} and
with our LMM and HLMM models explained above. For the ease of comparison,
we utilize the pipeline and settings that
was prescribed for the original EPLL \cite{Zoran11} algorithm
(see Section \ref{sec:image_restoration}).
To reduce the computation time of \textit{all} EPLL-based
algorithms, we utilize the random patch overlap procedure introduced by
\cite{parameswaran2017accelerating}. That is, instead of extracting
all patches at each iteration, a randomly selected but different subset of
overlapping patches consisting of only $3\%$ of all possible patches is processed
in each iteration.
For the sake of reproducibility of our results, we have made
our MATLAB/MEX-C implementation available online at
\url{\ourcodeurl}.

The evaluation is carried out on classical images such as
{\it Barbara}, {\it Cameraman}, {\it Hill}, {\it House}, {\it Lena} and {\it Mandrill},
and on 60 images taken from BSDS testing set \cite{martin2001database} (the original
BSDS test data contains 100 images, the other 40 were used for model
validation experiments). All image have been corrupted independently by
ten independent realizations of additive white Gaussian noise with standard deviation
$\sigma = 5, 10, 20, 40$ and $60$ (with pixel values between $[0, 255]$).
The EPLL algorithm using mixture of Gaussian, generalized Gaussian
priors are indicated as GMM and GGMM
in \Cref{tab:psnr_denoising}.
Results obtained with BM3D algorithm \cite{dabov2007image} are also included
for reference purposes. To stay with the focus of this paper,
\ie on the effect of image priors
on EPLL-based algorithms, BM3D will be excluded from
our performance comparison discussions.
The denoising performance of the algorithms
are measured in terms of Peak Signal to Noise
Ratio (PSNR) and Structural SIMilarity (SSIM) \cite{wang2004image}.
As can be observed in \Cref{tab:psnr_denoising}, in general,
GGMM prior provides better PSNR performance than the GMM prior.
In terms of SSIM values (shown in the bottom part of \Cref{tab:psnr_denoising}),
GGMM prior is comparable to GMM.
In order to demonstrate the effect of
fixed $\nu$ values compared to the more flexible GGMM prior, we
compare the results of GGMM against GMM ($\nu = 2$),
Laplacian Mixture Model (GGMM with $\nu=1$) and
hyper-Laplacian mixture model (GGMM with $\nu=0.5$) priors
in the same scenarios for $\sigma = 5, 20$ and $60$. These results
are shown in \Cref{tab:psnr_denoising_nu}.
GGMM prior provides better PSNR performance on average than the fixed-shape priors.
The differences in
denoising performance can also be verified visually in \Cref{fig:castle},
\Cref{fig:cameraman} and \Cref{fig:barbara}. The denoised images obtained
using GGMM prior show much fewer artifacts as compared to GMM-EPLL results,
in particular in homogeneous regions. On the
other hand, GGMM prior is also able to better preserve textures than LMM and HLMM.

\begin{figure}
  \centering%
  \includegraphics[width=.7\linewidth]{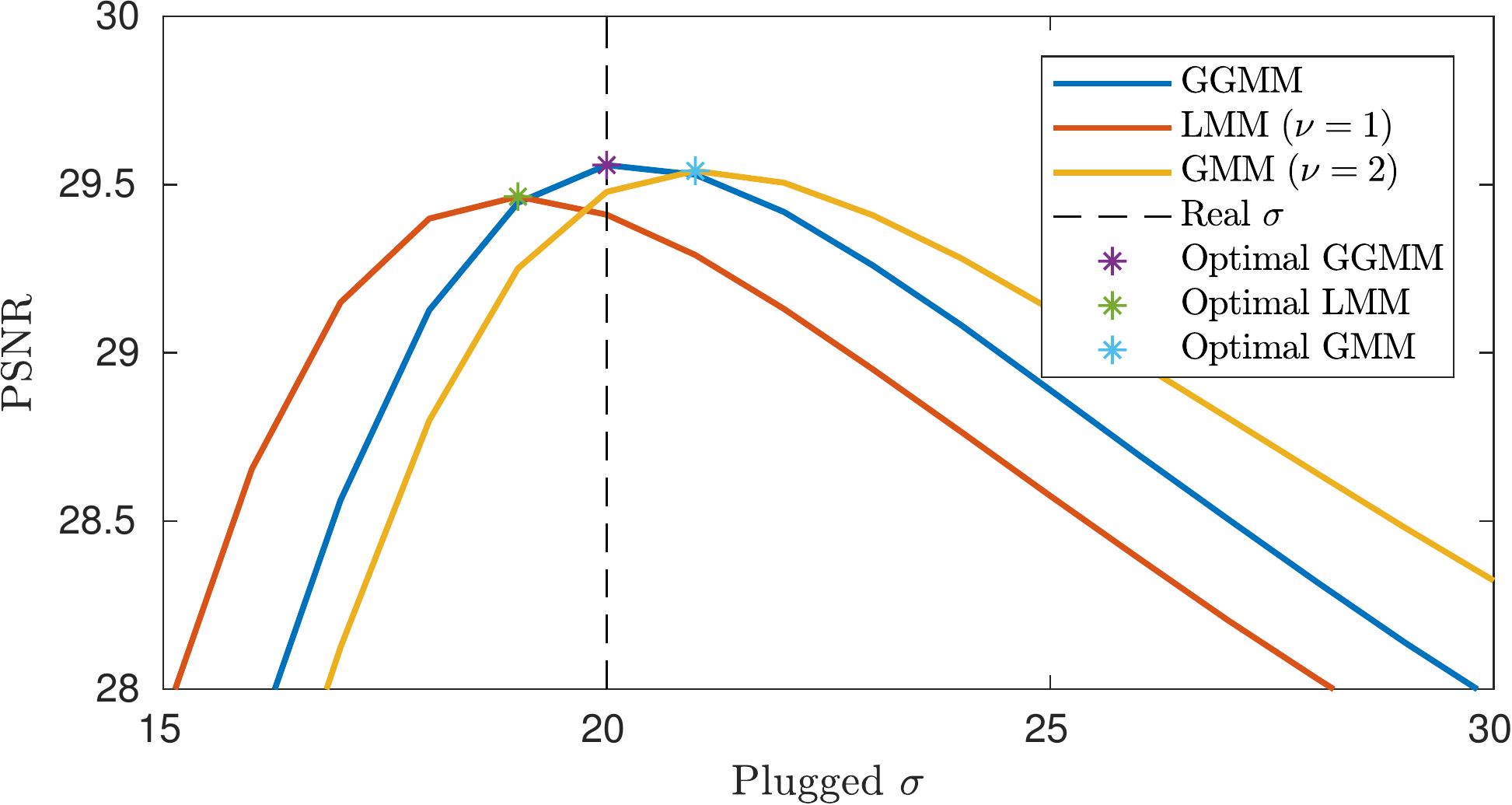}
  \caption{Evolution of performance of EPLL with a GGMM, LMM ($\nu = 1$)
    and a GMM ($\nu = 2$) under misspecification of the noise standard
    deviation $\sigma$. Performances are measured in terms of PSNR
    on the BSDS dataset corrupted by a Gaussian noise with standard
    deviation $\sigma=20$. For each of the three priors,
    EPLL has been run assuming $\sigma$ was ranging from $15$
    to $30$.}
    \label{fig:sig_evolve}
  \vskip1em
\end{figure}

\paragraph{Prior fitness for image denoising}
In this work, we have considered non-blind image denoising.
That is, the noise standard deviation is assumed to be known. In this setting,
if the restoration model is accurate, one should
ideally achieve optimal restoration performance when using the true
degradation. To verify this, we conducted a denoising task with image
corrupted with noise with standard deviation $\sigma=20$. We used GMM, LMM and GGMM priors
in the restoration framework with assumed $\sigma$ values ranging from 15 to 30.
\Cref{fig:sig_evolve} shows
the evolution of average restoration performance over 40 images
from BSDS testing set (kept aside
for validation, as mentioned above) with varying noise variances.
GGMM prior attains its best
performance when the noise variance used in the restoration model
matches with the ground truth $\sigma=20$. In contrast to GGMM,
GMM (resp., LMM) reaches its best performance at a larger
(resp., lower) value of $\sigma$ than the correct noise used during
degradation. This is because GMM tends to under-smooth clean patches
(resp., over-smooth) so that a larger (resp., lower) value of $\sigma$
is required to compensate the mismatch between the assumed prior
and the actual distribution in the restoration model.
 This indicates that GGMM is a better option to model
distribution of image patches than GMM or LMM.

\afterpage{
\begin{figure}[t!]
  \centering
  \subfigure[No approx.~(10h 29m)]{%
    {\includegraphicspsnr{.32\linewidth}{}{comp_direct/subimg2}{29.49}{0.877}}%
  }\hfill
  \subfigure[Approx.~discrepancy (2s54)]{%
    {\includegraphicspsnr{.32\linewidth}{}{comp_direct/subimg2}{29.47}{0.875}}%
  }\hfill
  \subfigure[Approx.~disc.~\& shrink.~(1s63)]{%
    {\includegraphicspsnr{.32\linewidth}{}{comp_direct/subimg3}{29.48}{0.875}}%
  }\\
  \caption{Results obtained by GGMM-EPLL on a
    $128 \times 128$ cropped image of the BSDS testing dataset damaged
    by additive white Gaussian noise with $\sigma=20$.
    These are obtained respectively by
    (a) evaluating the classification and shrinkage problem
    with numerical solvers (numerical integration and Halley's root-finding method),
    (b) approximating the classification problem only, and
    (c) approximating both problems.
    PSNR and SSIM are given in the bottom-left corner.
    Running time (averaged on ten runs)
    of the overall GGMM-EPLL are indicated on the captions:
    our accelerations lead to a speed-up of
    $\times 15,000$ and $\times 1.5$ respectively.}
    \label{fig:direct_vs_approx}
  \vskip1em
\end{figure}
\begin{table}[t!]
  \caption{
    Runtime profiles (averaged over ten runs) of GGMM-EPLL corresponding to
    the denoising experiment shown in \Cref{fig:direct_vs_approx}.
    These profiles are obtained respectively by
    evaluating the classification and shrinkage problem
    with numerical solvers (numerical integration and Halley's root-finding method),
    approximating the classification problem only,
    and approximating both problems.
    Profiles are split into discrepancy, shrinkage and
    patch extration/reprojection. Speed-up
    with respect to the no-approximation (first column) are indicated in parenthesis
    and major accelerations in green.
    Percentage of time taken for each step with respect to the overall execution time (first row),
    are indicated below each time reading and bottlenecks
    are indicated in red.
  }
  \label{tab:profile}
  \centering
  \begin{tabular}{@{}l@{\;}ccc@{}}
    \hline
    &
    \multicolumn{1}{c}{No approximations} &
    \multicolumn{1}{c}{Approx. discrepancy} &
    \multicolumn{1}{c}{Approx. disc. \& shrink.}\\
    \hline
    \hline
    \\[-.5em]
    Total &
    10h 29m 15s
    &
    2.54s \tiny ({\color{green!60!black} $\bf \times 15,000$})
    &
    1.63s \tiny ({\color{green!60!black} $\bf \times 23,000$})
    \\[-.5em]
    &\multicolumn{1}{c}{\tiny  100\%}
    &\multicolumn{1}{c}{\tiny  100\%}
    &\multicolumn{1}{c}{\hspace{-2.3cm}\tiny $\xrightarrow{{\color{green!60!black} \bf \times 1.5}}$
      \hspace{1.4cm}  100\%}
    \\[1.em]
    Discrepancy &
    10h 29m 14s
    &
    1.44s \tiny ({\bf \color{green!60!black} $\bf \times 26,000$})
    &
    1.44s \tiny ($\times 26,000$)
    \\[-.5em]
    &\multicolumn{1}{c}{\tiny \color{red} \bf $>$99.99\%}
    &\multicolumn{1}{c}{\tiny \color{red} \bf 56.69\%}
    &\multicolumn{1}{c}{\tiny \color{red} \bf 88.34\%}
    \\[1.em]
    Shrinkage &
    1.08s
    &
    1.08s \tiny ($\times 1$)
    &
    0.17s \tiny ({\bf \color{green!60!black} $\bf \times 6.3$})
    \\[-.5em]
    &\multicolumn{1}{c}{\tiny  $<$0.001\%}
    &\multicolumn{1}{c}{\tiny  \color{red} \bf 42.52\%}
    &\multicolumn{1}{c}{\tiny  10.43\%}
    \\[.7em]
    $\substack{
      \displaystyle \text{Patch extraction}\\[.5em]
      \displaystyle \text{and reprojection}
    }$
    &
    0.02s
    &
    0.02s \tiny ($\times 1$)
    &
    0.02s \tiny ($\times 1$)
    \\[-0.9em]
    &\multicolumn{1}{c}{\tiny  $<$0.001\%}
    &\multicolumn{1}{c}{\tiny  0.79\%}
    &\multicolumn{1}{c}{\tiny  1.23\%}
    \\[1em]
    \hline
  \end{tabular}
  \vspace{1em}
\end{table}
}

\paragraph{Influence of our approximations}
All previous experiments using GGMM patch priors
were conducted based on the two proposed
approximations introduced in \Cref{sec:discrepancy} and \Cref{sec:shrinkage}.
In \Cref{fig:direct_vs_approx} and \Cref{tab:profile},
we provide a quantitative illustration of the speed-ups provided by these approximations and
their effect on denoising performance.
Timings were carried out with Matlab 2018a
on an Intel(R) Core(TM) i7-7600U CPU @ 2.80GHz
(neither multi-core paralellization nor GPUs acceleration were used).
\Cref{fig:direct_vs_approx}a shows the result obtained
by calculating original discrepancy function via numerical integration and the shrinkage function
via Halley's root-finding method. This makes the denoising
process extremely slow and takes 10 hours and 29 minutes for denoising an image of size 128$\times$128 pixels.
The approximated discrepancy function provides 4 orders of magnitude speed-up with no perceivable
drop in performance (\Cref{fig:direct_vs_approx}b). In addition, incorporating the shrinkage
approximation provides further acceleration that allows the denoising to complete in less
than 2 seconds with a very minor drop in PSNR/SSIM.
As indicated in the detailed profiles on \Cref{tab:profile},
the shrinkage approximation provides an acceleration of six-fold to the
shrinkage calculation step itself and leads to an overall speed-up of 1.5
due to the larger bottleneck caused by discrepancy function calculation.
The approximately 23,000$\times$ speed-up realized 
without any perceivable drop in denoising performance underscores the efficacy of our proposed approximations.

\section{Conclusions and Discussion}
\label{sec:conclusions}
In this work, we suggest using a mixture of generalized Gaussians
for modeling the patch distribution of clean images. We provide a
detailed study of the challenges that one encounters
when using a highly flexible GGMM prior for image restoration
in place of a more common GMM prior. We identify the two main
bottlenecks in the restoration procedure when using EPLL and GGMM --
namely, the patch classification step and the shrinkage step. One
of the main contributions of this paper, is the
thorough theoretical analysis of the classification problem
allowing us to introduce an asymptotically accurate approximation that is
computationally efficient. In order to tackle the shrinkage step,
we collate and extend the existing solutions under GGMM prior.

Our numerical experiments indicate that our flexible GGMM patch prior is a better fit
for modeling natural images 
than GMM and other mixture distributions with constant shape parameters such as LMM or HLMM. In image denoising tasks, we have shown that using
GGMM priors, often, outperforms GMM when used in the EPLL framework.

Nevertheless, we believe the performance of GGMM prior in these scenarios
falls short of its expected potential. Given that GGMM is persistently a better prior than GMM
(in terms of log-likelihood),
one would expect the GGMM-EPLL to outperform GMM-EPLL consistently.
We postulate that this under-performance is caused by the EPLL strategy that we use for optimization.
That is, even though the GGMM prior may be improving the quality of
the global solution,
the half quadratic splitting strategy used
in EPLL is not guaranteed to return a better solution due to the non-convexity
of the underlying problem. For this reason,
we will focus our future work in designing specific optimization strategies
for GGMM-EPLL leveraging the better expressivity of the proposed prior model
for denoising and other general restoration applications.

Another direction of future research will focus on extending this work
to employ GGMMs/GGDs in other model-based signal processing tasks.
Of these tasks, estimating the parameters of GGMM
directly on noisy observations is a problem of particular interest,
that could benefit from our approximations.
Learning GMM priors on noisy patches has been shown to be
useful in patch-based image restoration
when clean patches are not available \textit{a priori},
or to further adapt the model to the specificities of a given noisy image
\cite{yu2012solving,teodoro2015single,houdard2017}.
Another open problem is to analyze
the asymptotic behavior of the minimum mean square estimator (MMSE)
shrinkage with GGD prior, as an alternative to MAP shrinkage.
This could be useful to design accurate
approximations for other general inference frameworks.
Last but not least, characterizing the exact asymptotic behaviors
of the convolution of two arbitrary GGDs, as investigated in \cite{soury2015new},
is still an open question.
To the best of our knowledge, our study is the first attempt towards this goal
but in ours one of the GGD is always Gaussian (noise).
Extending our study to the general GGD case (or even specific cases such as Laplacian)
is a challenging problem that is of major interest in signal processing tasks
where noise is not Gaussian but instead follows another GGD.



\section*{Acknowledgments}

The authors would like to thank Charles Dossal and Nicolas Papadakis for fruitful discussions.

Part of the experiments
presented in this paper were carried out using the PlaFRIM experimental testbed, supported by Inria, CNRS (LABRI and IMB), Universit\'e de Bordeaux, Bordeaux INP and Conseil R\'egional d'Aquitaine (see \url{https://www.plafrim.fr/}).

\appendix
\section{Proof of equation \eqref{eq:discrepancy_l1}}
\label{proof:discrepancy_l1}

\begin{proof}
  For $\nu=1$, using $\Gamma(1) = 1$ and $\Gamma(3) = 2$, we obtain
  \begin{align}
    f_{\sigma, \lambda}^1(x)
    &=
    \log (2 \sqrt{\pi} \sigma \lambda)
    -\log
    \int_{-\infty}^\infty
    e^{ -\frac{t^2}{2 \sigma^2}
      -\frac{\sqrt{2} |x-t|}{\lambda} } \; \d t~,
    \\
    \label{eq:proof_discrepancy_l1_start}
    &=
    \log (2 \sqrt{\pi} \sigma \lambda)
    -\log
    \left[
    e^{-\frac{\sqrt{2} x}{\lambda}}
    \int_{-\infty}^x
    e^{ -\frac{t^2}{2 \sigma^2} + \frac{\sqrt{2} t}{\lambda} } \; \d t
    +
    e^{\frac{\sqrt{2} x}{\lambda}}
    \int_{x}^\infty
    e^{ -\frac{t^2}{2 \sigma^2} -\frac{\sqrt{2} t}{\lambda} } \; \d t
    \right]
    ~.
  \end{align}
  Since $\erf'(t) = \frac{2 e^{-t^2}}{\sqrt{\pi}}$, it follows that
  for $a > 0$ and $b > 0$
  \begin{align}
  \pd{}{t}
  \left[
    -
    \frac{\sqrt{\pi a}}{2}
    e^{ \frac{a}{4 b^2} } \erf\left( -\frac{t}{\sqrt{a}} - \frac{\sqrt{a}}{2 b} \right)
    \right]
  =
  e^{
  -\frac{t^2}a - \frac{t}b
  }~.
  \end{align}
  Therefore we have with $a = 2\sigma^2$ and $b=\lambda/\sqrt{2}$
  \begin{align}
    \int_{-\infty}^x
    e^{ -\frac{t^2}{2 \sigma^2} + \frac{\sqrt{2} t}{\lambda} } \; \d t
    =
    \frac{
    -\sqrt{\pi} \sigma e^{ \frac{\sigma^2}{\lambda^2} }
    }{
      \sqrt{2}
    }
    \left[
      \erf\left( -\frac{t}{\sqrt{2}\sigma} + \frac{\sigma }{\lambda} \right)
    \right]_{-\infty}^x\!\!\!\!
    =
    \frac{
      \sqrt{\pi} \sigma e^{ \frac{\sigma^2}{\lambda^2} }
    }{
      \sqrt{2}
    }
    \erfc\left( -\frac{x}{\sqrt{2}\sigma} + \frac{\sigma }{\lambda} \right)
    ~,
  \end{align}
  since $\ulim{t \to \infty} \erf(t) = 1$ and $\erfc(t) = 1 - \erf(t)$.
  Similarly, we get
  \begin{align}
    &
    \int_x^\infty
    e^{ -\frac{t^2}{2 \sigma^2}  -\frac{\sqrt{2} t}{\lambda} } \; \d t
    =
    \frac{
      -\sqrt{\pi} \sigma e^{ \frac{\sigma^2}{\lambda^2} }
    }{
      \sqrt{2}
    }
    \erfc\left( \frac{x}{\sqrt{2}\sigma} + \frac{\sigma }{\lambda} \right)
    ~.
  \end{align}
  Plugging these two last expressions in \eqref{eq:proof_discrepancy_l1_start}
  and rearranging the terms
  conclude the proof.
\end{proof}

\section{Proof of \Cref{prop:disc_basic_prop}}
\label{proof:disc_basic_prop}

\begin{proof}
  Starting from the definition of $f_{\sigma,\lambda}^{\nu}$
  and using the change of variable $t \to \sigma t$,
  eq.~\eqref{eq:disc_reduction} follows as
  \begin{multline}
    f_{\sigma,\lambda}^{\nu}(x)
    =
    -\log
    \int_{\RR}
    \sigma
    \frac{\kappa}{2 \lambda_\nu}
    \exp\left[- \left(\frac{\sigma |t|}{\lambda_\nu}\right)^\nu \right]
    \frac{1}{\sqrt{2\pi}\sigma}
    \exp\left[- \frac{(x-\sigma t)^2}{2\sigma^2} \right]
    \;\d t~,
    \\
    =
    \log \sigma
    -\log
    \int_{\RR}
    \frac{\kappa}{2 \lambda_\nu/\sigma}
    \exp\left[- \left(\frac{|t|}{\lambda_\nu/\sigma}\right)^\nu \right]
    \frac{1}{\sqrt{2\pi}}
    \exp\left[- \frac{(x/\sigma-t)^2}{2} \right]
    \;\d t
    =
    \log \sigma + f_{1,\lambda/\sigma}^\nu\left(\frac{x}{\sigma}\right)~.
  \end{multline}
  Properties \eqref{eq:disc_symmetry} and \eqref{eq:disc_unimodal} hold since the
  convolution of two real even unimodal distributions
  is even unimodal \cite{wintner1938asymptotic,purkayastha1998simple}.
  Property \eqref{eq:disc_bound} follows from
  \eqref{eq:disc_symmetry}, \eqref{eq:disc_unimodal} and the fact that the convolution
  of continuous and bounded real functions are continuous and bounded.
\end{proof}

\section{Proof of \Cref{thm:assignment_nu1_left}}
\label{proof:nu1_left}

\begin{lemma}\label{lem:erfc}
  Let $a > 0$ and $b > 0$. For $x$ in the vicinity of $0$, we have
  \begin{align}
    \frac1{2abx}
    \log\left[
      \frac{\erfc(ax + b)
        +
        e^{-4 a b x}
        \erfc(-ax + b)
      }{2\erfc(b)}
      \right]
    &=
    - 1
    +
    \left(a b -
    \frac{a e^{-b^2}}{\erfc(b) \sqrt{\pi}}
    \right) x
    +
    o(x)~.
  \end{align}
\end{lemma}

\begin{proof}
  Since $\erfc'(x) = - \frac{2 e^{-x^2}}{\sqrt{\pi}}$
  and $\erfc''(x) = \frac{2 x e^{-x^2}}{\sqrt{\pi}}$,
  using second order Taylor's expansion for $x$ in the vicinity of $0$,
  it follows that
  \begin{align}
    \erfc(ax+b) &\usim{0}
    \erfc(b)
    - \frac{2 a e^{-b^2}}{\sqrt{\pi}} x
    + \frac{2 a^2 b e^{-b^2}}{\sqrt{\pi}} x^2~,
    \\
    \qandq
    \erfc(-ax+b) &\usim{0}
    \erfc(b)
    + \frac{2 a e^{-b^2}}{\sqrt{\pi}} x
    + \frac{2 a^2 b e^{-b^2}}{\sqrt{\pi}} x^2~.
  \end{align}
  We next make the following deductions
  \begin{gather}
    e^{-4 a b x}\erfc(-ax+b)) \usim{0}\!
    e^{-4 a b x}\erfc(b) + e^{-4 a b x}\frac{2 a e^{-b^2}}{\sqrt{\pi}} x
    + e^{-4 a b x} \frac{2 a^2 b e^{-b^2}}{\sqrt{\pi}} x^2~,
    \\
    \erfc(ax\!+\!b)
    \!+\!
    e^{-4 a b x}\erfc(-ax\!+\!b)) \usim{0}\!
    (1 \!+\! e^{-4 a b x})\!\left(\! \erfc(b) \!+\! \frac{2 a^2 b e^{-b^2}}{\sqrt{\pi}} x^2 \!\right)
    \!-\!
    (1 \!-\! e^{-4 a b x})
    \frac{2 a e^{-b^2}}{\sqrt{\pi}} x~,
    \nonumber
    \\
    \underbrace{\frac{\erfc(ax\!+\!b)
    \!+\!
    e^{-4 a b x}\erfc(-ax\!+\!b))
    }{2 \erfc{b}}}_{A(x)}
    \usim{0}\!
    \frac{1 \!+\! e^{-4 a b x}}{2}
    \left(\! 1 \!+\! \frac{2 a^2 b e^{-b^2}}{\erfc(b) \sqrt{\pi}} x^2 \!\right)
    \!-\!
    (1 \!-\! e^{-4 a b x})
    \frac{a e^{-b^2}}{\erfc(b) \sqrt{\pi}} x~.
    \nonumber
  \end{gather}
  The left-hand side $A(x)$ of this last equation
  is then, in the vicinity of $x=0$, equals to
  \begin{align}
    A(x) =
    \frac{1 + e^{-4 a b x}}{2}
    \left( 1 + \frac{2 a^2 b e^{-b^2}}{\erfc(b) \sqrt{\pi}} x^2 \right)
    -
    (1 - e^{-4 a b x})
    \frac{a e^{-b^2}}{\erfc(b) \sqrt{\pi}} x
    +
    o(x^2)~.
  \end{align}
  We next use second-order Taylor's expansion for $e^{-4 a b x}$,
  leading to
  \begin{align}
    A(x) =\;
    &(1 - 2 a b x + 4 a^2 b^2 x^2 + o(x^2))
    \left( 1 + \frac{2 a^2 b e^{-b^2}}{\erfc(b) \sqrt{\pi}} x^2 \right)
    \\
    &
    \hspace{5cm}
    -
    (4 a b x - 8 a^2 b^2 x^2 + o(x^2))
    \frac{a e^{-b^2}}{\erfc(b) \sqrt{\pi}} x
    +
    o(x^2)~,
    \nonumber
    \\
    =\;&
    1 - 2 a b x +
    \left(4 a^2 b^2 -
    \frac{2 a^2 b e^{-b^2}}{\erfc(b) \sqrt{\pi}}
    \right) x^2
    +
    o(x^2)~.
  \end{align}
  By using the second-order Taylor's expansion of $\log(1+x)$,
  it follows that
  \begin{gather}
    \log\left[
      A(x)
    \right]
    =
    - 2 a b x +
    \left(4 a^2 b^2 -
    \frac{2 a^2 b e^{-b^2}}{\erfc(b) \sqrt{\pi}}
    \right) x^2
    -
    2 a^2 b^2 x^2
    +
    o(x^2)~.
  \end{gather}
  Dividing both sides by $2 a b x$ then concludes the proof,
  \begin{gather}
    \frac1{2 a b x}
    \log\left[
    \frac{\erfc(ax+b)
    +
    e^{-4 a b x}\erfc(-ax+b))
    }{2 \erfc{b}}
    \right]
    =
    - 1 +
    \left(a b -
    \frac{a e^{-b^2}}{\erfc(b) \sqrt{\pi}}
    \right) x
    +
    o(x)~.
  \end{gather}
\end{proof}

\begin{proof}[Proof of \Cref{thm:assignment_nu1_left}]
  We first rewrite $\varphi^{1}_\lambda(x)$ as
  \begin{align}\label{eq:phi1_rewritten}
    \varphi^{1}_\lambda(x)
    &=
    \log \left[
    \log
    \left[
      2
      \erfc\left(
      \frac{1}{\lambda}
      \right)
      \right]
    -
    \log
    \left[
      e^{\frac{\sqrt{2} x}{\lambda}}
      \erfc\left(
      \frac{x}{\sqrt{2}}
      +
      \frac{1}{\lambda}
      \right)
      +
      e^{-\frac{\sqrt{2} x}{\lambda}}
      \erfc\left(
      -\frac{x}{\sqrt{2}}
      +
      \frac{1}{\lambda}
      \right)
      \right]
    \right]~,
    \nonumber
    \\
    &=
    \log \left[ \frac{\sqrt{2} x}{\lambda} \right]
    +
    \log \left[
      -
      1
      -
      \frac{\lambda}{\sqrt{2} x}
    \log
    \left[
      \frac{
        \erfc\left(
        \frac{x}{\sqrt{2}}
        +
        \frac{1}{\lambda}
        \right)
        +
        e^{-\frac{2\sqrt{2} x}{\lambda}}
        \erfc\left(
        -\frac{x}{\sqrt{2}}
        +
        \frac{1}{\lambda}
        \right)
      }{
        2
        \erfc\left(
        \frac{1}{\lambda}
        \right)
      }
      \right]
    \right]~.
  \end{align}
  Next, using \Cref{lem:erfc} with $a=1/\sqrt{2}$ and $b=1/\lambda$, it follows that
  \begin{align}
    \varphi^1_\lambda(x)
    &=
    \log \left[ \frac{\sqrt{2} x}{\lambda} \right]
    +
    \log \left[
      -
      \left(
      \frac{1}{\sqrt{2}\lambda} -
      \frac{e^{-\frac{1}{\lambda^2}}}{\sqrt{2} \erfc(\frac{1}{\lambda}) \sqrt{\pi}}
      \right) x
      +
      o(x)
      \right]~,
    \\
    &=
    \log \left[ \frac{x^2}{\lambda} \right]
    +
    \log \left[
      \frac{e^{-\frac{1}{\lambda^2}}}{\erfc(\frac{1}{\lambda}) \sqrt{\pi}}
      -
      \frac{1}{\lambda}
      +
      o(1)
      \right]~,
    \\
    &=
    \log \left[ \frac{x^2}{\lambda} \right]
    +
    \log \left[
      \frac{1}{\sqrt{\pi}}
      \frac{e^{-\frac{1}{\lambda^2}}}{\erfc(\frac{1}{\lambda})}
      -
      \frac{1}{\lambda}
      \right]
    +
    o(1)~,
  \end{align}
  where the last equation follows from the first-order Taylor expansion
  of $\log(a + x)$.
\end{proof}

\section{Proof of \Cref{thm:assignment_nu1_right}}
\label{proof:nu1_right}

\begin{lemma}\label{lem:erfc2}
  Let $a > 0$  and $b > 0$. For $x$ in the vicinity of $+\infty$, we have
  \begin{align}
    \frac1{2abx}
    \log\left[
      \frac{\erfc(ax + b)
        +
        e^{-4 a b x}
        \erfc(-ax + b)
      }2
      \right]
    &=
    -2 + o\left(1\right)~.
  \end{align}
\end{lemma}

\begin{proof}
  We have
  $\ulim{x \to +\infty} \erfc(x) = 0$ and
  $\ulim{x \to +\infty} \erfc(-x) = 2$,
  it follows that
  \begin{align}
    \erfc(-ax+b) &\usim{+\infty} 2~,
    \\
    e^{-4abx} \erfc(-ax+b) &\usim{+\infty} 2 e^{-4abx}~,
    \\
    \frac{\erfc(ax+b) + e^{-4abx} \erfc(-ax+b)}2 &\usim{+\infty} e^{-4abx}~,
    \\
    \log\left[\frac{\erfc(ax+b) + e^{-4abx} \erfc(-ax+b)}2 \right]
    &\usim{+\infty} -4abx~,
    \\
    \frac1{2abx} \log\left[\erfc(ax+b) + e^{-4abx} \erfc(-ax+b)\right]
    &\usim{+\infty} -2~,
  \end{align}
  where we have used the knowledge that $f \sim g$ implies that
  $\log f \sim \log g$.
\end{proof}

\begin{proof}[Proof of \Cref{thm:assignment_nu1_right}]
  By writing $\varphi^1_\lambda$ as in eq.~\eqref{eq:phi1_rewritten} and
  using \Cref{lem:erfc2} with $a=1/\sqrt{2}$ and $b=1/\lambda$, it follows that
  \begin{align*}
    \varphi^1_\lambda(x)
    &=
    \log \left[ \frac{\sqrt{2} x}{\lambda} \right]
    +
    \log \left[
      1
      +
      \frac{\lambda}{\sqrt{2} x}
      \log \erfc\left(\frac{1}{\lambda}\right)
      +
      o\left(1\right)
      \right]
    =
    \log \left[ \frac{\sqrt{2} x}{\lambda} \right]
    +
    o\left(1\right)~.
\end{align*}
\end{proof}

\section{Proof of \Cref{thm:assignment_general_left}}
\label{proof:nugen_left}
\begin{proof}
  We first decompose the term $\exp\left( -\frac{(x - t)^2}{2} \right)$
  involved in the definition of $f_{1,\lambda}^\nu$ and rewrite $e^{xt}$
  using its power series
  \begin{align}
    f_{1,\lambda}^\nu(x)
    =
    &
    -\log \frac{1}{\sqrt{2 \pi} } \frac{\nu}{2 \lambda_\nu \Gamma(1/\nu)}
    -\log
    \int_{-\infty}^\infty
    \exp\left( -\frac{(x - t)^2}{2} \right)
    \exp\left[ -\left(\frac{|t|}{\lambda_\nu}\right)^\nu \right] \; \d t~,
    \\
    =&
    -\log \frac{1}{\sqrt{2 \pi} } \frac{\nu}{2 \lambda_\nu \Gamma(1/\nu)}
    -\log
    \int_{-\infty}^\infty
    e^{-\frac{x^2}{2}}
    e^{xt}
    e^{-\frac{t^2}{2}}
    \exp\left[ -\left(\frac{|t|}{\lambda_\nu}\right)^\nu \right] \; \d t~,
    \\
    =&
    -\log \frac{1}{\sqrt{2 \pi} } \frac{\nu}{2 \lambda_\nu \Gamma(1/\nu)}
    +\frac{x^2}{2}
    -\log
    \int_{-\infty}^\infty
    \sum_{k=0}^{\infty}
    \frac{(xt)^k}{k!}
    e^{-\frac{t^2}{2}}
    \exp\left[ -\left(\frac{|t|}{\lambda_\nu}\right)^\nu \right] \; \d t~.
  \end{align}
  For $x$ in the vicinity of $0$, we can consider $|x| \leq 1$, and then
  \begin{align}
    \int_{-\infty}^\infty
    \sum_{k=0}^{\infty}
    \left|
    \frac{(xt)^k}{k!}
    e^{-\frac{t^2}{2}}
    \exp\left[ -\left(\frac{|t|}{\lambda_\nu}\right)^\nu \right]
    \right|
    \; \d t
    &
    \leq
    \int_{-\infty}^\infty
    \sum_{k=0}^{\infty}
    \frac{|t|^k}{k!}
    e^{-\frac{t^2}{2}}
    \exp\left[ -\left(\frac{|t|}{\lambda_\nu}\right)^\nu \right] \; \d t~,
    \\
    &
    \leq
    \int_{-\infty}^\infty
    e^{-\frac{t^2}{2} + |t| -\left(\frac{|t|}{\lambda_\nu}\right)^\nu } \; \d t
    \;
    < \infty~.
  \end{align}
  Then, Fubini's theorem applies and we get
  \begin{align}
    f_{1,\lambda}^\nu(x)
    =&
    -\log \frac{1}{\sqrt{2 \pi} } \frac{\nu}{2 \lambda_\nu \Gamma(1/\nu)}
    +\frac{x^2}{2}
    -\log
    \sum_{k=0}^{\infty}
    \int_{-\infty}^\infty
    \frac{(xt)^k}{k!}
    e^{-\frac{t^2}{2}}
    \exp\left[ -\left(\frac{|t|}{\lambda_\nu}\right)^\nu \right] \; \d t~,
    \\
    =&
    -\log \frac{1}{\sqrt{2 \pi} } \frac{\nu}{2 \lambda_\nu \Gamma(1/\nu)}
    +\frac{x^2}{2}
    -\log
    \sum_{k=0}^{\infty}
    \frac{x^k}{k!}
    \int_{-\infty}^\infty
    t^k
    e^{-\frac{t^2}{2}}
    \exp\left[ -\left(\frac{|t|}{\lambda_\nu}\right)^\nu \right] \; \d t~.
  \end{align}
  By definition, we have
  \begin{align}
    \gamma_\lambda^\nu \triangleq f_{1,\lambda}^\nu(0) =
  -\log \frac{1}{\sqrt{2 \pi} } \frac{\nu}{2 \lambda_\nu \Gamma(1/\nu)}
  -\log
  \int_{-\infty}^\infty
  \exp\left( -\frac{t^2}{2} \right)
  \exp\left[ -\left(\frac{|t|}{\lambda_\nu}\right)^\nu \right]
  \; \d t~.
  \end{align}
  Moreover, when $k$ is odd, we have
  \begin{align}
  \int_{-\infty}^\infty
  t^k
  e^{-\frac{t^2}{2}}
  \exp\left[ -\left(\frac{|t|}{\lambda_\nu}\right)^\nu \right] \; \d t
  =
  0~.
  \end{align}
  Using third-order Taylor expansion of $\log(1+x)$ for $x$ in the vicinity of $0$,
  it follows that
  \begin{align}
    f_{1,\lambda}^\nu(x)
    =\;&
    \gamma_\lambda^\nu
    +\frac{x^2}{2}
    -\log
    \left(
    1+
    \frac{x^2}{2}
    \frac{
      \int_{-\infty}^\infty
      t^2
      e^{-\frac{t^2}{2}}
      \exp\left[ -\left(\frac{|t|}{\lambda_\nu}\right)^\nu \right] \; \d t
    }{
      \int_{-\infty}^\infty
      e^{-\frac{t^2}{2}}
      \exp\left[ -\left(\frac{|t|}{\lambda_\nu}\right)^\nu \right] \; \d t
    }
    +
    o(x^3)
    \right)~,
    \\
    =\;&
    \gamma_\lambda^\nu
    +
    \frac{x^2}{2}
    \left(
    1
    -
    \frac{
      \int_{-\infty}^\infty
      t^2
      e^{-\frac{t^2}{2}}
      \exp\left[ -\left(\frac{|t|}{\lambda_\nu}\right)^\nu \right] \; \d t
    }{
      \int_{-\infty}^\infty
      e^{-\frac{t^2}{2}}
      \exp\left[ -\left(\frac{|t|}{\lambda_\nu}\right)^\nu \right] \; \d t
    }
    \right)
    +
    o(x^3)~.
  \end{align}
  Finally, using first-order Taylor's expansion for $\log(1+x)$,
  we conclude the proof as
  \begin{align}
    \varphi^\nu_\lambda(x) &=
    \log\left[
      \frac{x^2}{2}
      \left(
      1
      -
      \frac{
        \int_{-\infty}^\infty
        t^2
        e^{-\frac{t^2}{2}}
        \exp\left[ -\left(\frac{|t|}{\lambda_\nu}\right)^\nu \right] \; \d t
      }{
        \int_{-\infty}^\infty
        e^{-\frac{t^2}{2}}
        \exp\left[ -\left(\frac{|t|}{\lambda_\nu}\right)^\nu \right] \; \d t
      }
      +
      o(x)
      \right)
      \right]~,
    \\
    &=
    2 \log x
    - \log 2
    +
    \log\left(
    1
    -
    \frac{
      \int_{-\infty}^\infty
      t^2
      e^{-\frac{t^2}{2}}
      \exp\left[ -\left(\frac{|t|}{\lambda_\nu}\right)^\nu \right] \; \d t
    }{
      \int_{-\infty}^\infty
      e^{-\frac{t^2}{2}}
      \exp\left[ -\left(\frac{|t|}{\lambda_\nu}\right)^\nu \right] \; \d t
    }
    \right)
    +
    o(x)~.
  \end{align}
\end{proof}

\section{Proof of \Cref{thm:assignment_general_right}}
\label{proof:nugen_right}

We first recall in a lemma, a result extracted
from Corollary 3.3 in \cite{berman1992tail}.

\begin{lemma}[Berman]\label{lem:simeon}
Let $p$ and $q$ be differentiable real probability density functions.
Define for $x$ large enough
$
  u(x) = p^{-1}(q(x))
$
and define $v$ and $w$ as
\begin{align}
  v(x) = - \pd{}{x} \log p(x) \qandq
  w(x) = - \pd{}{x} \log q(x)~.
\end{align}
Assume $v$ and $w$ are positive continuous function and regularly oscillating, {\it i.e.}:
\begin{align}
  \ulim{\substack{x, x' \to \infty\\x/x' \to 1}}
  \frac{v(x)}{v(x')} = 1
  \qandq
  \ulim{\substack{x, x' \to \infty\\x/x' \to 1}}
  \frac{w(x)}{w(x')} = 1~.
\end{align}
Suppose that we have
\begin{align}
  \ulim{x \to \infty}
  \frac{w(x)}{v(x)} = 0
  \qandq
  \ulim{x \to \infty} u(x) w(x) = +\infty~,
\end{align}
then, for $x \to \infty$, we have
\begin{align}
  \log \int_{-\infty}^{+\infty}
  p(x-t) q(t)
  \; \d t
  \sim
  \log q(x)~.
\end{align}
\end{lemma}

\begin{proof}[Proof of \Cref{thm:assignment_general_right}]
  Using the definition of the discrepancy function,
  \begin{align}
    f_{1,\lambda}^{\nu}(x)
    =
    -\log
    \int_{\RR}
    \Gg(t; 0, \lambda, \nu) \cdot \Nn(x-t; 0, 1)
    \;\d t
    =
    -\log
    \int_{-\infty}^{+\infty}
    p(x - t) q(t)
    \;\d t~.
  \end{align}
  with $q(x) = \Gg(x; 0, \lambda, \nu)$
  and $p(x) = \Nn(x; 0, 1)$.
  We have
\begin{align}
  v(x) = - \pd{}{x} \log p(x) = 2 x \qandq
  w(x) = - \pd{}{x} \log q(x) = \frac{\nu x^{\nu -1}}{\lambda_\nu^\nu}~.
\end{align}
Remark that $v$ and $w$ are positive continuous, and
\begin{align}
  \ulim{\substack{x, x' \to \infty\\x/x' \to 1}}
  \frac{v(x)}{v(x')} =
  \ulim{\substack{x, x' \to \infty\\x/x' \to 1}}
  \frac{x}{x'} = 1
  \qandq
  \ulim{\substack{x, x' \to \infty\\x/x' \to 1}}
  \frac{w(x)}{w(x')} =
  \ulim{\substack{x, x' \to \infty\\x/x' \to 1}}
  \left(\frac{x}{x'}\right)^{\nu - 1} = 1
  ~.
\end{align}
  For $x > 0$ large enough and $y > 0$ small enough
  \begin{align}
    y = p(x)
    \Leftrightarrow
    y = \frac{1}{\sqrt{2\pi}} \exp\left( -\frac{x^2}{2} \right)
    \Leftrightarrow
    x = \sqrt{-\log(2\pi) - 2 \log y}~.
  \end{align}
  Then, from \Cref{prop:sep},
  we have 
  \begin{align}
    u(x) w(x) &=
    \frac{\nu x^{\nu - 1}}{\lambda_\nu^\nu} \sqrt{-\log(2\pi) - 2 \log q(x)}\\
    &=
    \frac{\nu x^{\nu - 1}}{\lambda_\nu^\nu}
    \sqrt{-\log(2\pi) - 2 \log\left[
        \frac{\kappa}{2 \lambda_\nu}
        \right]
      +
      2 \left(\frac{x}{\lambda_\nu}\right)^\nu
    }
    \sim
    \frac{
      \nu \sqrt{2} x^{\frac32\nu - 1}
    }{
      \lambda_\nu^{\frac32\nu}
    }
    ~,
  \end{align}
  and thus, as $\nu > \frac23$,  $\ulim{x \to \infty} u(x) w(x) = \infty$.
Moreover, for $\nu < 2$, we have
\begin{align}
  \ulim{x \to \infty}
  \frac{w(x)}{v(x)} =
  \ulim{x \to \infty}
  \frac{\nu}{2 \lambda_\nu} x^{\nu - 2} = 0~.
\end{align}
It follows that \Cref{lem:simeon} applies, and then
for large $x$
\begin{align}
  f_{1,\lambda}^{\nu}(x)
    \sim
    -\log q(x)
    \sim
    \left(\frac{x}{\lambda_\nu}\right)^\nu~.
\end{align}
Using that $\lambda_\nu = \lambda \sqrt{\frac{\Gamma(1/\nu)}{\Gamma(3/\nu)}}$,
we conclude the proof since
\begin{align}
  \varphi_{\lambda}^{\nu}(x)
  =
  \log\left[ f_{1,\lambda}^{\nu}(x) - \gamma_\lambda^\nu \right]
  \sim
  \nu \log x
  -
  \nu \log \lambda_\nu~.
\end{align}
\end{proof}


\section{Proof of \Cref{prop:shrink_basic_prop}}
\label{proof:shrink_basic_prop}

\begin{proof}
  Starting from the definition of $s_{\sigma,\lambda}^{\nu}$
  and using the change of variable $t \to \sigma t$,
  eq.~\eqref{eq:shrink_reduction} follows as
  \begin{align}
  s_{\sigma,\lambda}^\nu(x)
  &=
  \uargmin{t\in \RR}
  \frac{(x - t)^2}{2 \sigma^2}
  +
  \lambda_\nu^{-\nu}|t|^\nu
  =
  \sigma
  \uargmin{t\in \RR}
  \frac{(x - \sigma t)^2}{2 \sigma^2}
  +
  \lambda_\nu^{-\nu}|\sigma t|^\nu~,
  \\
  &=
  \sigma
  \uargmin{t\in \RR}
  \frac{(x/\sigma - t)^2}{2}
  +
  (\lambda_\nu/\sigma)^{-\nu} |t|^\nu
  =
  \sigma s_{1,\lambda/\sigma}^\nu(x/\sigma)~.
  \end{align}
  For eq.~\eqref{eq:shrink_anti_symmetry}, we use the change of variable $t \to -t$
  \begin{align}
  s_{\sigma,\lambda}^\nu(-x)
  &=
  \uargmin{t\in \RR}
  \frac{(-x - t)^2}{2 \sigma^2}
  +
  \lambda_\nu^{-\nu}|t|^\nu
  =
  -
  \uargmin{t\in \RR}
  \frac{(-x + t)^2}{2 \sigma^2}
  +
  \lambda_\nu^{-\nu}|t|^\nu~,
  \\
  &=
  -
  \uargmin{t\in \RR}
  \frac{(x - t)^2}{2 \sigma^2}
  +
  \lambda_\nu^{-\nu}|t|^\nu
  =
  -s_{\sigma,\lambda}^\nu(x)~.
  \end{align}
  We now prove eq.~\eqref{eq:shrinkage}.
  Let $t = s_{\sigma,\lambda}^\nu(x)$, and since
 $t$ minimizes the objective, then
  \begin{align}
    \lambda_\nu^{-\nu}|t|^\nu
    \leq
    \frac{(x - t)^2}{2 \sigma^2} + \lambda_\nu^{-\nu}|t|^\nu
    \leq
    \frac{(x - x)^2}{2 \sigma^2} + \lambda_\nu^{-\nu}|x|^\nu
    = \lambda_\nu^{-\nu}|x|^\nu~,
  \end{align}
  which implies that $|t| \leq |x|$.
  Let $x > 0$ and assume $t = s_{\sigma,\lambda}^\nu(x) < 0$.
  Since $t$ minimizes the objective, then
  \begin{align}
    \frac{(x - t)^2}{2 \sigma^2} + \lambda_\nu^{-\nu}|t|^\nu
    \leq
    \frac{(x + t)^2}{2 \sigma^2} + \lambda_\nu^{-\nu}|t|^\nu
  \end{align}
  which implies that $-x t \leq x t$ and leads to a contradiction.
  Then for $x > 0$, $s_{\sigma,\lambda}^\nu(x) \in [0, x]$,
  which concludes the proof since $s_{\sigma,\lambda}^\nu$ is odd.

  We now prove \eqref{eq:shrink_increasing}.
  Let $x_1 > x_2$ and define $t_1 = s_{\sigma,\lambda}^\nu(x_1)$ and
  $t_2 = s_{\sigma,\lambda}^\nu(x_2)$.
  Since $t_1$ and $t_2$ minimize their respective objectives, the following
  two statements hold
  \begin{align}
    \frac{(x_1 - t_1)^2}{2 \sigma^2} + \lambda_\nu^{-\nu}|t_1|^\nu
    & \leq
    \frac{(x_1 - t_2)^2}{2 \sigma^2} + \lambda_\nu^{-\nu}|t_2|^\nu~,
    \\
    \qandq
    \frac{(x_2 - t_2)^2}{2 \sigma^2} + \lambda_\nu^{-\nu}|t_2|^\nu
    & \leq
    \frac{(x_2 - t_1)^2}{2 \sigma^2} + \lambda_\nu^{-\nu}|t_1|^\nu~.
  \end{align}
  Summing both inequalities lead to
  \begin{align}
    & (x_1 - t_1)^2 + (x_2 - t_2)^2
    \leq
    (x_1 - t_2)^2 + (x_2 - t_1)^2~,
    \\
    \Rightarrow \quad &
    -2 x_1 t_1 - 2 x_2 t_2
    \leq
    -2 x_1 t_2 - 2 x_2 t_1~,
    \\
    \Rightarrow \quad &
    t_1 (x_1 - x_2)
    \geq
    t_2 (x_1 - x_2)
    \quad
    \Rightarrow \quad
    t_1 \geq t_2 \quad \text{(since $x_1 > x_2$)}~.
  \end{align}

  We now prove \eqref{eq:shrink_increasing2}.
  Let $\lambda_1 > \lambda_2$ and define
  $t_1 = s_{\sigma,\lambda_1}^\nu(x)$ and
  $t_2 = s_{\sigma,\lambda_2}^\nu(x)$.
  Since $t_1$ and $t_2$ minimize their respective objectives, the following
  expressions hold
  \begin{align}
    \frac{(x - t_1)^2}{2 \sigma^2} + \lambda_{\nu,1}^{-\nu}|t_1|^\nu
    & \leq
    \frac{(x - t_2)^2}{2 \sigma^2} + \lambda_{\nu,1}^{-\nu}|t_2|^\nu~,
    \\
    \qandq
    \frac{(x - t_2)^2}{2 \sigma^2} + \lambda_{\nu,2}^{-\nu}|t_2|^\nu
    & \leq
    \frac{(x - t_1)^2}{2 \sigma^2} + \lambda_{\nu,2}^{-\nu}|t_1|^\nu~.
  \end{align}
  Again, summing both inequalities lead to
  \begin{align}
    &
    \lambda_{\nu,1}^{-\nu}|t_1|^\nu + \lambda_{\nu,2}^{-\nu}|t_2|^\nu
    \leq
    \lambda_{\nu,1}^{-\nu}|t_2|^\nu + \lambda_{\nu,2}^{-\nu}|t_1|^\nu~,
    \\
    \Rightarrow \quad
    &
    (\lambda_{\nu,1}^{-\nu} - \lambda_{\nu,2}^{-\nu}) |t_1|^\nu
    \leq
    (\lambda_{\nu,1}^{-\nu} - \lambda_{\nu,2}^{-\nu})|t_2|^\nu~,
    \\
    \Rightarrow \quad
    &
    |t_1|^\nu
    \geq
    |t_2|^\nu
    \quad \text{(since $\lambda_1 > \lambda_2$ and $\nu > 0$)}~.
  \end{align}

  We now prove \eqref{eq:shrink_keep}. Consider $x > 0$. Since $\lambda \mapsto s_{\sigma,\lambda}^\nu(x)$ is a monotonic function and $s_{\sigma,\lambda}^\nu(x) \in [0, x]$ for all $\lambda$,
  it converges for $\lambda \to \infty$ to a value $\omega \in [0, x]$.
  Assume $0 < \omega < x$ and
  let $0 < \epsilon < \max(\omega, x - \omega)$. By definition of the limit,
  for $\lambda$ big enough
  \begin{align}
    0 < \omega - \epsilon < t \triangleq s_{\sigma,\lambda}^\nu(x) < \omega + \epsilon~.
  \end{align}
  It follows that $x - t > x - (w + \epsilon) > 0$, and then
  \begin{align}
    \frac{(x - (\omega + \epsilon))^2}{2 \sigma^2} + \lambda_\nu^{-\nu}|\omega - \epsilon|^\nu
    <
    \frac{(x - t)^2}{2 \sigma^2} + \lambda_\nu^{-\nu}|t|^\nu~.
  \end{align}
  Moreover, since $\omega + \epsilon \ne x$,
  we have for $\lambda$ big enough
  \begin{align}
    \lambda_\nu^{-\nu}|x|^\nu
    <
    \frac{(x - (\omega + \epsilon))^2}{2 \sigma^2} + \lambda_\nu^{-\nu}|\omega - \epsilon|^\nu~.
  \end{align}
  Combining the two last inequalities shows that
  \begin{align}
    \frac{(x - x)^2}{2 \sigma^2} + \lambda_\nu^{-\nu}|x|^\nu
    <
    \frac{(x - t)^2}{2 \sigma^2} + \lambda_\nu^{-\nu}|t|^\nu~,
  \end{align}
  which is in contradiction with the fact that $t$ minimizes the objective.
  As a consequence, $\omega = x$, which concludes the proof since
  $s_{\sigma,\lambda}^\nu(x)$ is odd and satisfies \eqref{eq:shrink_reduction}.

  We now prove \eqref{eq:shrink_kill}. Consider $x > 0$. Since $\lambda \mapsto s_{\sigma,\lambda}^\nu(x)$ is a monotonic function and $s_{\sigma,\lambda}^\nu(x) \in [0, x]$ for all $\lambda$,
  it converges for $\lambda \to 0^+$ to a value $\omega \in [0, x]$.
  Assume $0 < \omega < x$ and
  let $0 < \epsilon < \max(\omega, x - \omega)$.
  Again, we have for $\lambda$ small enough
  \begin{align}
    \frac{(x - (\omega + \epsilon))^2}{2 \sigma^2} + \lambda_\nu^{-\nu}|\omega - \epsilon|^\nu
    <
    \frac{(x - t)^2}{2 \sigma^2} + \lambda_\nu^{-\nu}|t|^\nu~.
  \end{align}
  Moreover, since $\omega \ne \epsilon$,
  we have for $\lambda$ small enough
  \begin{align}
    \frac{x^2}{2 \sigma^2}
    <
    \frac{(x - (\omega + \epsilon))^2}{2 \sigma^2} + \lambda_\nu^{-\nu}|\omega - \epsilon|^\nu~.
  \end{align}
  Combining the two last inequalities shows that
  \begin{align}
    \frac{(x - 0)^2}{2 \sigma^2} + \lambda_\nu^{-\nu}|0|^\nu
    <
    \frac{(x - t)^2}{2 \sigma^2} + \lambda_\nu^{-\nu}|t|^\nu~,
  \end{align}
  which is in contradiction with the fact that $t$ minimizes the objective.
  As a consequence, $\omega = 0$, which concludes the proof since
  $s_{\sigma,\lambda}^\nu(x)$ is odd and satisfies \eqref{eq:shrink_reduction}.
\end{proof}


\bibliographystyle{siamplain}
\bibliography{references}
\end{document}

%% file: ex_shared.tex
\usepackage{mystyle}
\usepackage{lipsum}
\usepackage{amsfonts}
\usepackage{graphicx}
\usepackage{epstopdf}
\usepackage{algorithmic}
\ifpdf
  \DeclareGraphicsExtensions{.eps,.pdf,.png,.jpg}
\else
  \DeclareGraphicsExtensions{.eps}
\fi

\numberwithin{theorem}{section}

\newcommand{\TheTitle}{Image denoising with generalized Gaussian mixture model
  patch priors
}
\newcommand{\TheAuthors}{C. Deledalle, S. Parameswaran, and T.Q. Nguyen}

\headers{\TheTitle}{\TheAuthors}

\title{{\TheTitle}} 

\author{
  {Charles-Alban Deledalle%
  \thanks{Institut de Math\'ematiques de Bordeaux, CNRS, Universit\'e de Bordeaux, Bordeaux INP, Talence, France
    (\email{charles-alban.deledalle@math.u-bordeaux.fr}).
  }}
  \footnotemark[2]
  \and
  Shibin Parameswaran%
  \thanks{Department of Electrical and Computer Engineering,
    University of California, San Diego, La Jolla, USA
    (\email{cdeledalle@eng.ucsd.edu},
    \email{sparames@ucsd.edu},
    \email{tqn001@eng.ucsd.edu}).}
  \and
  Truong Q. Nguyen%
  \footnotemark[2]
}

\usepackage{amsopn}
\DeclareMathOperator{\diag}{diag}


%% file: table_3algos_main.tex
\begin{table*}[!t]
  \centering
  \caption{Image denoising performance comparison of EPLL algorithm  with GMM and GGMM priors.
  			PSNR and SSIM values are obtained on the BSDS test set (average over 60 images),
    		        on six standard images corrupted with 5 different levels of noise (average over 10 noise realizations), and finally
                        an average over these 66 images. BM3D algorithm results are also included
    		for reference purposes.}
  \begin{tabular}{ll cccccccc}
    \hline\\[-.9em]
$\sigma$& Algo. & BSDS & barbara &    \shortstack{camera\\man} &    hill &     house &   lena & mandrill & Avg.\\
    \hline
    \hline
    \\[-1em]
    & & \multicolumn{8}{c}{PSNR}\\
    \cline{3-10}\\[-.9em]
    \rowcolor{Gray}
    \multirow{3}{*}{\cellcolor{white} $5$}
&       BM3D  & 37.33 & 38.30 & 38.28 & 36.04 & 39.82 & 38.70 & 35.26 & 37.36 \\
&        GMM  & 37.25 & 37.60 & 38.07 & 35.93 & 38.81 & 38.49 & 35.22 & 37.26 \\
&       GGMM  & \bf 37.33 & \bf 37.73 & \bf 38.12 & \bf 35.95 & \bf 38.94 & \bf 38.52 & \bf 35.23 & \bf 37.33 \\
\hline
\rowcolor{Gray}
\multirow{3}{*}{\cellcolor{white} $10$}
&       BM3D  & 33.06 & 34.95 & 34.10 & 31.88 & 36.69 & 35.90 & 30.58 & 33.15 \\
&        GMM  & 33.02 & 33.65 & 33.91 & 31.79 & 35.56 & 35.46 & 30.55 & 33.06 \\
&       GGMM  & \bf 33.10 & \bf 33.87 & \bf 34.01 & \bf 31.81 & \bf 35.72 & \bf 35.59 & \bf 30.58 & \bf 33.15 \\
\hline
\rowcolor{Gray}
\multirow{3}{*}{\cellcolor{white} $20$}
&       BM3D  & 29.38 & 31.73 & 30.42 & 28.56 & 33.81 & 33.02 & 26.60 & 29.50 \\
&        GMM  & 29.36 & 29.76 & 30.16 & 28.46 & 32.77 & 32.40 & 26.60 & 29.42 \\
&       GGMM  & \bf 29.43 & \bf 30.02 & \bf 30.24 & \bf 28.48 & \bf 33.03 & \bf 32.59 & \bf 26.64 & \bf 29.50 \\
\hline
\rowcolor{Gray}
\multirow{3}{*}{\cellcolor{white} $40$}
&       BM3D  & 26.28 & 27.97 & 27.16 & 25.89 & 30.69 & 29.81 & 23.07 & 26.38 \\
&       GMM  & 26.21 & 26.02 & 26.93 & 25.68 & 29.60 & 29.18 & \bf 23.25 & 26.26 \\
&      GGMM  & \bf 26.26 & \bf 26.17 & \bf 27.03 & \bf 25.70 & \bf 29.89 & \bf 29.42 & 23.21 & \bf 26.32 \\
\hline
\rowcolor{Gray}
\multirow{3}{*}{\cellcolor{white} $60$}
&       BM3D  & 24.81 & 26.31 & 25.24 & 24.52 & 28.74 & 28.19 & 21.71 & 24.90 \\
&       GMM   & 24.57 & 23.95 & 25.10 & 24.21 & 27.53 & 27.28 & \bf 21.57 & 24.61 \\
&      GGMM   & \bf 24.64 & \bf 24.03 & \bf 25.17 & \bf 24.25 & \bf 27.80 & \bf 27.52 & 21.50 & \bf 24.67 \\
    \hline
    \\[-1em]
    & & \multicolumn{8}{c}{SSIM}\\
    \cline{3-10}\\[-.9em]
    \rowcolor{Gray}
    \multirow{3}{*}{\cellcolor{white} \cellcolor{white} $5$}
&       BM3D  &.9619 &.9643 &.9613 &.9508 &.9571 &.9436 &.9588 &.9614 \\
&       GMM   &.9626 &.9616 & \bf.9604 & \bf.9511 & \bf.9475 & \bf.9434 & \bf.9597 &.9618 \\
&      GGMM   & \bf.9628 & \bf.9617 &.9602 &.9507 &.9469 &.9425 &.9593 & \bf.9620 \\
\hline
\rowcolor{Gray}
\multirow{3}{*}{\cellcolor{white} $10$}
&       BM3D  &.9115 &.9410 &.9286 &.8821 &.9215 &.9155 &.8983 &.9117 \\
&       GMM   & \bf.9155 &.9298 &.9307 & \bf.8858 & \bf.8999 &.9107 & \bf.9022 & \bf.9150 \\
&      GGMM   &.9154 & \bf .9313 & \bf.9309 &.8839 &.8992 & \bf.9112 &.9007 &.9149 \\
\hline
\rowcolor{Gray}
\multirow{3}{*}{\cellcolor{white} $20$}
&       BM3D  &.8236 &.9036 &.8685 &.7789 &.8741 &.8763 &.7943 &.8260 \\
&       GMM   & \bf.8315 &.8687 & \bf.8704 & \bf.7812 &.8596 &.8639 & \bf.8030 & \bf.8324 \\
&      GGMM   &.8297 & \bf .8715 &.8699 &.7766 & \bf .8629 & \bf .8669 &.7991 &.8308 \\
\hline
\rowcolor{Gray}
\multirow{3}{*}{\cellcolor{white} $40$}
&        BM3D  &.7074 &.8196 &.7954 &.6599 &.8276 &.8143 &.6184 &.7118 \\
&        GMM   & \bf.7054 &.7509 &.7780 & \bf.6496 &.8025 &.7918 & \bf.6341 & \bf.7081 \\
&       GGMM   &.7018 & \bf .7526 & \bf.7842 &.6430 & \bf .8112 & \bf .7995 &.6192 &.7048 \\
\hline
\rowcolor{Gray}
\multirow{3}{*}{\cellcolor{white} $60$}
&        BM3D &.6375 &.7581 &.7496 &.5859 &.7956 &.7784 &.4993 &.6427 \\
&        GMM  & \bf.6212 &.6534 &.7174 & \bf.5661 &.7507 &.7350 & \bf.5001 & \bf.6241 \\
&       GGMM  &.6174 & \bf.6544 & \bf.7266 &.5592 & \bf .7622 & \bf .7438 &.4782 &.6207 \\
\hline
\end{tabular}
  \label{tab:psnr_denoising}
  \vskip1em
\end{table*}

%% file: table_nu_main.tex
\begin{table*}[!t]
  \centering
  \caption{Image denoising performance comparison of EPLL algorithm with different priors.
  			PSNR values are obtained on the BSDS test set (average over 60 images),
    		        on six standard images corrupted with 3 different levels of noise (average over 10 noise realizations), and finally
                        an average over these 66 images.}
  \begin{tabular}{ll cccccccc}
    \hline\\[-.9em]
$\sigma$& Algo. & BSDS & barbara &    \shortstack{camera\\man} &    hill &     house &   lena & mandrill & Avg.\\
    \hline
    \hline
    \\[-1em]
    & & \multicolumn{8}{c}{PSNR}\\
    \cline{3-10}\\[-.9em]
    \multirow{4}{*}{$5$}
&        GMM  & 37.25 & 37.60 & 38.07 & 35.93 & 38.81 & 38.49 & 35.22 & 37.26 \\
&     LMM  		& 37.31 & \bf 37.83 & 38.11 & 35.89 & 38.93 & 38.49 & 35.18 & 37.32 \\
&    HLMM  		& 36.85 & 37.42 & 37.66 & 35.39 & 38.37 & 38.08 & 34.77 & 36.86 \\
&       GGMM  	& \bf 37.33 & 37.73 & \bf 38.12 & \bf 35.95 & \bf 38.94 & \bf 38.52 & \bf 35.23 & \bf 37.33 \\
\hline

\multirow{4}{*}{$20$}
&       GMM  & 29.36 & 29.76 & 30.16 & 28.46 & 32.77 & 32.40 & 26.60 & 29.42 \\
&    LMM  	  & 29.30 & \bf 30.18 & 30.04 & 28.36 & \bf 33.22 & \bf 32.72 & 26.43 & 29.37 \\
&   HLMM  	  & 28.48 & 29.28 & 29.04 & 27.72 & 32.50 & 32.10 & 25.44 & 28.56 \\
&      GGMM  & \bf 29.43 & 30.02 & \bf 30.24 & \bf 28.48 & 33.03 & 32.59 & \bf 26.64 & \bf 29.50 \\
\hline

\multirow{4}{*}{$60$}
&        GMM  & 24.57 & 23.95 & 25.10 & 24.21 & 27.53 & 27.28 & \bf 21.57 & 24.61 \\
&     LMM  	  & 24.55 & 23.94 & 24.96 & 24.23 & \bf 27.91 & \bf 27.58 & 21.35 & 24.59 \\
&    HLMM  	  & 23.95 & 23.16 & 23.72 & 23.84 & 27.10 & 26.94 & 20.67 & 23.97 \\
&       GGMM  & \bf 24.64 & \bf 24.03 & \bf 25.17 & \bf 24.25 & 27.80 & 27.52 & 21.50 & \bf 24.67 \\
\hline
\end{tabular}
  \label{tab:psnr_denoising_nu}
  \vskip1em
\end{table*}

%% file: visual_results.tex
\begin{figure}
  \centering%
   \subfigure[Reference $\bx$]{%
     \begin{tikzpicture}[spy using outlines={red,magnification=1.5,size=2.25cm, connect spies}]
     \node{\includegraphicsnopsnr{.45\linewidth}{viewport=0 238 320 450,clip}{castle_sig20_T5/subimg1}};%
     \end{tikzpicture}%
  }\hfill%
  \subfigure[Noisy $\by$]{%
    \begin{tikzpicture}[spy using outlines={red,magnification=1.5,size=2.25cm, connect spies}]
    \node{\includegraphicspsnr{.45\linewidth}{viewport=0 238 320 450,clip}{castle_sig20_T5/subimg2}{22.12}{0.3679}};%
    \end{tikzpicture}%
  }\\
  \subfigure[GMM ($\nu = 2$)]{%
    \begin{tikzpicture}[spy using outlines={red,magnification=1.5,size=2.25cm, connect spies}]
    \node{\includegraphicspsnr{.45\linewidth}{viewport=0 238 320 450,clip}{castle_sig20_T5/subimg3}{30.43}{0.8678}};%
    \spy on (-1.1,.8) in node [left] at (-1.25,3.6);
    \spy on (.5,1.6) in node [left] at (1.125,3.6);
    \spy on (2.7,-.1) in node [left] at (3.5,3.6);
    \end{tikzpicture}%
  }\hfill
  \subfigure[GGMM ($.3 \leq \nu \leq 2$)]{%
    \begin{tikzpicture}[spy using outlines={red,magnification=1.5,size=2.25cm, connect spies}]%
    \node{\includegraphicspsnr{.45\linewidth}{viewport=0 238 320 450,clip}{castle_sig20_T5/subimg4}{30.47}{0.8686}};%
    \spy on (-1.1,.8) in node [left] at (-1.25,3.6);
    \spy on (.5,1.6) in node [left] at (1.125,3.6);
    \spy on (2.7,-.1) in node [left] at (3.5,3.6);
    \end{tikzpicture}%
  }%
  \\
  \subfigure[LMM ($\nu=1$)]{%
    \begin{tikzpicture}[spy using outlines={red,magnification=1.5,size=2.25cm, connect spies}]
    \node{\includegraphicspsnr{.45\linewidth}{viewport=0 238 320 450,clip}{castle_sig20_T5/subimg5}{30.39}{0.8666}};%
    \spy on (-1.1,.8) in node [left] at (-1.25,3.6);
    \spy on (.5,1.6) in node [left] at (1.125,3.6);
    \spy on (2.7,-.1) in node [left] at (3.5,3.6);
    \end{tikzpicture}%
  }\hfill
  \subfigure[HLMM ($\nu=0.5$)]{%
    \begin{tikzpicture}[spy using outlines={red,magnification=1.5,size=2.25cm, connect spies}]
    \node{\includegraphicspsnr{.45\linewidth}{viewport=0 238 320 450,clip}{castle_sig20_T5/subimg6}{28.48}{0.8522}};%
    \spy on (-1.1,.8) in node [left] at (-1.25,3.6);
    \spy on (.5,1.6) in node [left] at (1.125,3.6);
    \spy on (2.7,-.1) in node [left] at (3.5,3.6);
    \end{tikzpicture}%
  }
  \\
  \caption{%
    (a) Close in on the image {\it Castle} from the BSDS testing dataset,
    (b) a noisy version degraded by additive white Gaussian noise with standard
    deviation $\sigma=20$ and (c)-(f) results of EPLL under four patch priors:
    GMM, GGMM, LMM and HLMM, respectively. PSNR and SSIM are given
    in the bottom-left corner.}
  \label{fig:castle}
  \vskip1em
\end{figure}
\begin{figure*}
  \centering%
  \subfigure[Reference $\bx$]{%
    \begin{tikzpicture}[spy using outlines={red,magnification=1.5,size=2.25cm, connect spies}]
      \node{\includegraphicsnopsnr{.45\linewidth}{viewport=0 66 256 236,clip}{cameraman_sig20_T5/subimg1}};%
    \end{tikzpicture}
  }\hfill%
    \subfigure[Noisy $\by$]{%
    \begin{tikzpicture}[spy using outlines={red,magnification=1.5,size=2.25cm, connect spies}]
      \node{\includegraphicspsnr{.45\linewidth}{viewport=0 66 256 236,clip}{cameraman_sig20_T5/subimg2}{22.14}{0.3972}};%
    \end{tikzpicture}
  }\\
  \subfigure[GMM ($\nu = 2$)]{%
    \begin{tikzpicture}[spy using outlines={red,magnification=1.5,size=2.25cm, connect spies}]
    \node{\includegraphicspsnr{.45\linewidth}{viewport=0 66 256 236,clip}{cameraman_sig20_T5/subimg3}{30.17}{0.8690}};%
    \spy on (-2.6,1.2) in node [left] at (-1.25,3.6);
    \spy on (-.3,-1.1) in node [left] at (1.125,3.6);
    \spy on (0.7,.9) in node [left] at (3.5,3.6);
    \end{tikzpicture}
  }\hfill
  \subfigure[GGMM ($.3 \leq \nu \leq 2$)]{%
    \begin{tikzpicture}[spy using outlines={red,magnification=1.5,size=2.25cm, connect spies}]
    \node{\includegraphicspsnr{.45\linewidth}{viewport=0 66 256 236,clip}{cameraman_sig20_T5/subimg4}{30.28}{0.8681}};%
    \spy on (-2.6,1.2) in node [left] at (-1.25,3.6);
    \spy on (-.3,-1.1) in node [left] at (1.125,3.6);
    \spy on (0.7,.9) in node [left] at (3.5,3.6);
    \end{tikzpicture}
  }\\
  \subfigure[LMM ($\nu=1$)]{%
    \begin{tikzpicture}[spy using outlines={red,magnification=1.5,size=2.25cm, connect spies}]
    \node{\includegraphicspsnr{.45\linewidth}{viewport=0 66 256 236,clip}{cameraman_sig20_T5/subimg5}{30.04}{0.8578}};%
    \spy on (-2.6,1.2) in node [left] at (-1.25,3.6);
    \spy on (-.3,-1.1) in node [left] at (1.125,3.6);
    \spy on (0.7,.9) in node [left] at (3.5,3.6);
    \end{tikzpicture}
  }\hfill
  \subfigure[HLMM ($\nu=0.5$)]{%
    \begin{tikzpicture}[spy using outlines={red,magnification=1.5,size=2.25cm, connect spies}]
    \node{\includegraphicspsnr{.45\linewidth}{viewport=0 66 256 236,clip}{cameraman_sig20_T5/subimg6}{29.05}{0.8421}};%
    \spy on (-2.6,1.2) in node [left] at (-1.25,3.6);
    \spy on (-.3,-1.1) in node [left] at (1.125,3.6);
    \spy on (0.7,.9) in node [left] at (3.5,3.6);
    \end{tikzpicture}
  }\\
  \caption{%
    (a) Close in on the standard image {\it Cameraman}.
    (b) a noisy version degraded by additive white Gaussian noise with standard
    deviation $\sigma=20$ and (c)-(f) results of EPLL under four patch priors:
    GMM, GGMM, LMM and HLMM, respectively. PSNR and SSIM are given
    in the bottom-left corner.}
  \label{fig:cameraman}
  \vskip1em
\end{figure*}


\begin{figure*}
  \centering%
  \subfigure[Reference $\bx$]{%
    \begin{tikzpicture}[spy using outlines={red,magnification=1.5,size=2.25cm, connect spies}]
      \node{\includegraphicsnopsnr{.45\linewidth}{viewport=100 220 512 492,clip}{barbara_sig20_T5/subimg1}};%
    \end{tikzpicture}
  }\hfill%
  \subfigure[Noisy $\by$]{%
    \begin{tikzpicture}[spy using outlines={red,magnification=1.5,size=2.25cm, connect spies}]
      \node{\includegraphicspsnr{.45\linewidth}{viewport=100 220 512 492,clip}{barbara_sig20_T5/subimg2}{22.11}{0.4764}};%
    \end{tikzpicture}
  }\\
  \subfigure[GMM ($\nu = 2$)]{%
    \begin{tikzpicture}[spy using outlines={red,magnification=1.5,size=2.25cm, connect spies}]
      \node{\includegraphicspsnr{.45\linewidth}{viewport=100 220 512 492,clip}{barbara_sig20_T5/subimg3}{29.75}{0.8694}};%
      \spy on (-2.7,-1) in node [left] at (-1.25,3.6);
      \spy on (-1,0.7) in node [left] at (1.125,3.6);
      \spy on (1.2,1) in node [left] at (3.5,3.6);
    \end{tikzpicture}
  }\hfill%
  \subfigure[GGMM  ($.3 \leq \nu \leq 2$)]{%
    \begin{tikzpicture}[spy using outlines={red,magnification=1.5,size=2.25cm, connect spies}]
      \node{\includegraphicspsnr{.45\linewidth}{viewport=100 220 512 492,clip}{barbara_sig20_T5/subimg4}{30.03}{0.8729}};%
      \spy on (-2.7,-1) in node [left] at (-1.25,3.6);
      \spy on (-1,0.7) in node [left] at (1.125,3.6);
      \spy on (1.2,1) in node [left] at (3.5,3.6);
    \end{tikzpicture}
  }\\
  \subfigure[LMM ($\nu=1$)]{%
    \begin{tikzpicture}[spy using outlines={red,magnification=1.5,size=2.25cm, connect spies}]
      \node{\includegraphicspsnr{.45\linewidth}{viewport=100 220 512 492,clip}{barbara_sig20_T5/subimg5}{30.21}{0.8737}};%
      \spy on (-2.7,-1) in node [left] at (-1.25,3.6);
      \spy on (-1,0.7) in node [left] at (1.125,3.6);
      \spy on (1.2,1) in node [left] at (3.5,3.6);
    \end{tikzpicture}
  }\hfill
  \subfigure[HLMM ($\nu=0.5$)]{%
    \begin{tikzpicture}[spy using outlines={red,magnification=1.5,size=2.25cm, connect spies}]
      \node{\includegraphicspsnr{.45\linewidth}{viewport=100 220 512 492,clip}{barbara_sig20_T5/subimg6}{29.27}{0.8495}};%
      \spy on (-2.7,-1) in node [left] at (-1.25,3.6);
      \spy on (-1,0.7) in node [left] at (1.125,3.6);
      \spy on (1.2,1) in node [left] at (3.5,3.6);
    \end{tikzpicture}
  }\\
  \caption{%
    (a) Close in on the standard image {\it Barbara}.
    (b) a noisy version degraded by additive white Gaussian noise with standard
    deviation $\sigma=20$ and (c)-(f) results of EPLL under four patch priors:
    GMM, GGMM, LMM and HLMM, respectively. PSNR and SSIM are given
    in the bottom-left corner.}
  \label{fig:barbara}
  \vskip1em
\end{figure*}